\documentclass{article}
\usepackage[letterpaper, margin=1in]{geometry} 
\usepackage{amssymb}
\usepackage{amsmath}
\allowdisplaybreaks
\usepackage{amsthm}
\usepackage{xspace}
\usepackage{physics}
\usepackage{xcolor}
\usepackage{tikz}
\usetikzlibrary{positioning, calc}
\usepackage{mathrsfs}
\usepackage[dvipsnames]{xcolor}
\usepackage{mathtools}
\usepackage{thmtools}
\usepackage{dirtytalk}
\usepackage{stmaryrd}
\usepackage{enumitem}
\usepackage{relsize}
\usepackage{graphicx}
\usepackage[normalem]{ulem}
\usepackage{bbm}
\usepackage[colorlinks=true, allcolors=blue]{hyperref}
\usepackage[capitalise,nameinlink]{cleveref}
\usepackage{mathtools}
\usepackage{tcolorbox}
\usepackage{thmtools}
\usepackage{thm-restate}
\tcbuselibrary{skins, breakable}

\usepackage[font=small]{caption}

\newtheorem{theorem}{Theorem}[section]

\newtheorem{lemma}[theorem]{Lemma}
\newtheorem{definition}[theorem]{Definition}

\crefname{assumption}{Assumption}{Assumptions}
\newtheorem{proposition}[theorem]{Proposition}

\crefname{conjecture}{Conjecture}{Conjectures}
\newtheorem{observation}[theorem]{Observation}
\newtheorem{question}[theorem]{Question}

\newcommand{\cOp}[1]{\ensuremath{\mathrm{#1}}\xspace}
\newcommand{\cpO}{\cOp{cpO}}
\newcommand{\pC}{\cOp{pC}}
\newcommand{\cpwt}{\cOp{cpWT}}
\newcommand{\mcpwt}{\cOp{\widetilde{cpWT}}}
\newcommand{\cpwtZ}{\cOp{cpWT_{\circ}}}
\newcommand{\cpLabInv}{\cOp{cpO^{inv}_{lab}}}
\newcommand{\cpLabFor}{\cOp{cpO^{for}_{lab}}}
\newcommand{\cpPerm}{\cOp{cpO_{edg}}}
\newcommand{\xorOP}{\cOp{P}}
\newcommand{\wtO}{\ensuremath{O^{(\mathrm{wt})}}\xspace}
\newcommand{\wtOZ}{\ensuremath{O^{\mathrm{wt}}_{\circ}}}

\newcommand{\wtF}{\ensuremath{F^{(\mathrm{wt})}}\xspace}

\newcommand{\middleCol}{\ensuremath{\mathrm{B}_{\mathrm{mid}}}\xspace}
\newcommand{\middleEdge}{\mathrm{B}_{\mathrm{edg}}\xspace}
\newcommand{\middleNodes}{\mathrm{V}_{\mathrm{mid}}\xspace}
\newcommand{\midBoundNodes}{\delta\mathrm{V}_{\mathrm{mid}}\xspace}
\newcommand{\edgeIndic}{\mathrm{b_{\mathrm{edg}}}\xspace}
\newcommand{\cmfWT}{\cOp{cmfWT}}
\newcommand{\cmfO}{\cOp{cmfO}}

\newcommand{\edgeVertex}[2]{\mathbf{e}_{#1}(#2)\xspace}
\newcommand{\probAssignEdge}[2]{\ensuremath{\kappa}_{#1}(#2)}
\newcommand{\colSeqSet}{\ensuremath{\mathcal{H}}}
\newcommand{\emptySeq}{\epsilon}
\newcommand{\colorVertMap}{\phi}
\newcommand{\mapLabels}{\ensuremath{\mathcal{P}}}
\newcommand{\outsideColorDb}{\ensuremath{C^{\mathrm{out}}}}
\newcommand{\outsideLabelDb}{\ensuremath{L^{\mathrm{out}}}}
\newcommand{\reprTuple}{\ensuremath{\mathcal{D}}}
\newcommand{\nbhood}{\ensuremath{\mathcal{N}}}
\newcommand{\goodMap}{G}
\newcommand{\shapeToDb}{j_{\colorVertMap}}
\newcommand{\canVec}[1]{\ket{g_{#1}}}
\newcommand{\assignVar}{\ensuremath{\mathcal{Z}}}
\newcommand{\famQParts}{\ensuremath{\mathfrak{B}}}
\newcommand{\typeInv}{(\mathrm{i})}
\newcommand{\typeFor}{(\mathrm{ii})}
\newcommand{\typeEdgOut}{(\mathrm{iii})}
\newcommand{\typeEdgeIn}{\mathrm{(iv)}}

\newcommand{\vecRem}{\xi^{\mathrm{rem}}}
\newcommand{\vecUnif}{\xi^{\mathrm{unif}}}

\newcommand{\pCLabelInv}{\cOp{pC_{lab}^{inv}}}
\newcommand{\pCPerm}{\cOp{pC_{edg}}}
\newcommand{\pCLabel}{\cOp{pC_{lab}^{for}}}
\newcommand{\pCfromPlus}{\cOp{pCF}}
\newcommand{\pcwt}{\cOp{pCWT}}
\newcommand{\mpCPerm}{\cOp{\widetilde{pC_{edg}}}}
\newcommand{\stepFunct}{\mathrm{sF}}

\newcommand{\gcoin}{\mathrm{G}}
\newcommand{\shiftOp}{\mathrm{S}}
\newcommand{\shiftComp}{\mathrm{S^{\circ}}}
\newcommand{\SWAP}{\mathrm{SWAP}}

\newcommand{\labelPerm}{\ensuremath{\sigma_{\textup{label}}}\xspace}
\newcommand{\colorPerm}{\ensuremath{\sigma_{\textup{edg}}}\xspace}

\newcommand{\colorPresSet}{\ensuremath{S_{\textup{edg}}}\xspace}
\newcommand{\head}{\ensuremath{\mathrm{head}}}
\newcommand{\tail}{\ensuremath{\mathrm{tail}}}
\newcommand{\treeEdges}{E_{\mathrm{tr},c,\alpha}}
\newcommand{\weldedEdges}{E_{\mathrm{wld}, \alpha}}

\newcommand{\multDistr}{\ensuremath{\mathcal{D}_{\mathrm{prod}}}\xspace}
\newcommand{\singDistr}{\ensuremath{\mathcal{D}_{\mathrm{sing}}}\xspace}
\newcommand{\evSingle}{\ensuremath{\mathcal{S}}\xspace}

\newcommand{\register}[1]{\ensuremath{\mathrm{#1}}\xspace}
\newcommand{\X}{\register{X}}
\newcommand{\Y}{\register{Y}}
\newcommand{\Z}{\register{Z}}
\newcommand{\A}{\register{A}}
\newcommand{\I}{\register{I}}

\newcommand{\C}{\register{C}}
\newcommand{\regO}{\register{O}}
\newcommand{\Lab}{\register{L}}
\newcommand{\regP}{\register{P}}

\newcommand{\db}{\ensuremath{\mathbf{I}}\xspace}

\newcommand{\col}{\ensuremath{\textup{col}}\xspace}

\newcommand{\rG}{\ensuremath{\hat{G}}}
\newcommand{\rV}{\ensuremath{\hat{V}}}
\newcommand{\rE}{\ensuremath{\hat{E}}}
\newcommand{\vPerm}{\ensuremath{V_\mathrm{perm}}}

\newcommand{\InjLabel}{\ensuremath{\I\mathrm{lab}}\xspace}

\newcommand{\InjColPres}{\ensuremath{\I\mathrm{pres}}\xspace}

\newcommand{\Sym}{\ensuremath{\mathrm{Sym}}}

\newcommand{\blocks}{\mathrm{blocks}}
\newcommand{\partition}{\mathrm{partition}}
\newcommand{\intPart}{\ensuremath{\calV_{\partition}}}
\newcommand{\intBlocks}{\ensuremath{\calV_{\blocks}}}
\newcommand{\dbJ}{\mathbf{J}}
\newcommand{\inter}{\ensuremath{\calI_S}}

\newcommand{\prCross}{\ensuremath{\Pi_{\mathrm{Cross}}}}
\newcommand{\prCanonicalCross}{\ensuremath{\Pi_{\rho}^{\mathrm{can}}}}
\newcommand{\prWeld}{\ensuremath{\Pi_\rho}}
\newcommand{\prWeldOrtho}{\ensuremath{\Pi_\rho^{\bot}}}
\newcommand{\prFuture}{\ensuremath{W_{\rho, t}}}
\newcommand{\prBoundary}{\Pi^{\mathrm{bdry}}}
\newcommand{\prUnif}[1]{\ensuremath{\Pi_{#1}^{\mathrm{unif}}}}
\newcommand{\prGoodAssign}{\ensuremath{\Pi_{\colorVertMap}}}

\newcommand{\probWTO}{\ensuremath{p_{\mathrm{wt}}}}
\newcommand{\probCWT}{\ensuremath{p_{\cpwt}}}

\newcommand{\eventCrossEdge}{\ensuremath{\textup{Cross}_E}}
\newcommand{\eventBoundary}{\ensuremath{\textup{Bdry}}}

\newcommand{\insertDir}{\ensuremath{A_{\rho, t}^{\sigma}}}

\newcommand{\lastAddTime}{\ensuremath{\lambda}}

\DeclareMathOperator{\EX}{\mathbb{E}}
\newcommand{\poly}{\ensuremath{\mathrm{poly}}}

\DeclareMathOperator{\dom}{\mathrm{Dom}}
\DeclareMathOperator{\im}{\mathrm{Im}}
\DeclareMathOperator{\ran}{\mathrm{Ran}}
\DeclareMathOperator{\Id}{\mathbbm{1}}
\DeclarePairedDelimiter\ceil{\lceil}{\rceil}
\DeclarePairedDelimiter\floor{\lfloor}{\rfloor}

\newcommand{\groot}{\ensuremath{\mathrm{r}}} % I am g(raph) root :)
\newcommand{\depth}{\ensuremath{\mathrm{depth}}}
\newcommand{\strucIdx}{\ensuremath{\mathrm{w}}\xspace}
\newcommand{\alg}{\ensuremath{\mathcal{A}}\xspace}

\newcommand{\calH}{\ensuremath{\mathcal{H}}}
\newcommand{\calR}{\ensuremath{\mathcal{R}}}

\newcommand{\calA}{\ensuremath{\mathcal{A}}}

\newcommand{\calI}{\ensuremath{\mathcal{I}}}
\newcommand{\calM}{\ensuremath{\mathcal{M}}}
\newcommand{\calP}{\ensuremath{\mathcal{P}}}
\newcommand{\calV}{\ensuremath{\mathcal{V}}}
\newcommand{\calK}{\ensuremath{\mathcal{K}}}
\newcommand{\calS}{\ensuremath{\mathcal{S}}}
\newcommand{\calD}{\ensuremath{\mathcal{D}}}

\newcommand{\pairPath}{\mathbf{x_{\textup{path}}}}

\providecommand{\Span}{\mathsf{Span}}
\newcommand{\supp}{\mathrm{Supp}}
\newcommand{\decrFact}[2]{(#1)_{#2}}

\newcommand{\Ext}{\ensuremath{\mathrm{Ext}}}

\colorlet{colorA}{RoyalBlue}
\colorlet{colorB}{ForestGreen}
\colorlet{colorC}{BrickRed}
\newcommand{\colorSet}{\ensuremath{\mathcal{C}}}
\newcommand{\upcolor}{\ensuremath{\chi}}
\newcommand{\colF}{\ensuremath{\mathrm{Color}}\xspace}
\newcommand{\ncF}{\ensuremath{\mathrm{Nb}}}

\newcommand{\colClass}[1]{\ensuremath{\mathcal{C}_{#1}}}
\newcommand{\colClassSize}[1]{\ensuremath{\mathrm{N}_{#1}}}

\usepackage{mdframed}
  {%
  \begin{mdframed}%
  [leftmargin=.5em,rightmargin=0em,#1,bottomline=false,topline=false,rightline=false,linewidth=2pt,linecolor=black!20,innerleftmargin=.5em,innertopmargin=0pt,innerrightmargin=0pt]\begingroup\vspace*{0pt}}%
  {\endgroup
  \end{mdframed}
  }

  {%
  \begin{mdframed}%
  [leftmargin=.5em,rightmargin=0em,#1,bottomline=false,topline=false,rightline=false,linewidth=2pt,linecolor=red!20,innerleftmargin=.5em,innertopmargin=0pt,innerrightmargin=0pt]\begingroup\vspace*{-2pt}}%
  {\endgroup
  \end{mdframed}
  }

\definecolor{algbg}{HTML}{F7F7F7}
\definecolor{algframe}{HTML}{CCCCCC}
\definecolor{alglabel}{HTML}{555555}
\definecolor{inputlabel}{HTML}{888888}

\newcounter{algoctr}[section]

\crefname{algoctr}{Algorithm}{Algorithms}
\Crefname{algoctr}{Algorithm}{Algorithms}

\newcommand{\alglabel}[1]{%
  \setcounter{algoctr}{\value{theorem}}%
  \addtocounter{algoctr}{-1}%
  \refstepcounter{algoctr}%
  \label{#1}%
}

\newtcolorbox{algobox}[2][]{%          #1 = tcolorbox keys, #2 = title text
  enhanced,
  breakable,
  colback    = algbg,
  colframe   = algframe,
  boxrule    = 0.3pt,
  arc        = 1.5pt,
  left       = 10pt,
  right      = 10pt,
  top        = 2pt,
  bottom     = 10pt,
  fontupper  = \small,
  title      = {\refstepcounter{theorem}%
                \small\textsc{Algorithm \thetheorem}\;\;#2},
  coltitle   = alglabel,
  colbacktitle = algbg,
  toptitle   = 6pt,
  bottomtitle = 6pt,
  #1
}
 
\title{Hardness of Pathfinding in a Welded Tree}

\author{David Miloschewsky\\\small{Department of Computer Science, }\\
\small{Stony Brook University}\\
\small{dmiloschewsk@cs.stonybrook.edu}\and Supartha Podder
\\
\small{Department of Computer Science, }\\
\small{Stony Brook University}\\\small{supartha@cs.stonybrook.edu}}

\date{}

\begin{document}

\maketitle

\begin{abstract}
    Starting from the entrance of a welded tree, a quantum walk algorithm can find its exit vertex exponentially faster than any classical algorithm~\cite{CCD+03}. However, it has been an open question whether any quantum algorithm is able to efficiently find a path from the entrance to the exit~\cite{Aar21,CCG23}. We answer this by proving an exponential quantum query lower bound for finding such path. Specifically, for trees of height $n$, any quantum query algorithm requires at least $\Omega(2^{n/24})$ queries in order to succeed with constant probability.

    Our proof uses the compressed permutation oracle technique of~\cite{Car25} in order to construct databases which track the graph information an algorithm has learned and forgotten, and show that no efficient quantum algorithm can build an entrance-to-exit path in these records.
\end{abstract}

\tableofcontents

\section{Introduction}

The study of quantum algorithms is motivated by the search for problems which quantum computers can solve faster than their classical counterparts. While polynomial speedups such as those for search~\cite{Gro97} and amplitude estimation~\cite{BHMT02} are valuable, we are primarily interested in superpolynomial speedups. These include factoring~\cite{Sho97}, Simon's problem~\cite{Sim97}, forrelation~\cite{AA18} and the code intersection problem~\cite{YZ24}. Another example is the welded trees problem in which a quantum walk algorithm can find the exit vertex exponentially faster than any classical algorithm~\cite{CCD+03}. In this paper, we ask whether this speedup extends to outputting a path from the entrance to the exit.

A welded tree graph $G_n$ consists of two binary trees of height $n+1$ whose leaves are joined by a cycle which alternates between the two sides (see \cref{fig:basic_welded_tree}). We call the roots of the trees the \emph{entrance} and \emph{exit} and the edges on the cycle the \emph{welded edges}. %Letting $N\coloneq 2^n$, notice that $G_n$ contains $\Theta(N)$ vertices and each tree has $2^n$ leaves.
In the welded trees problem, each vertex is assigned a random label, an algorithm is given the label of the entrance vertex and oracle access to neighboring labels and their task is to output the exit label. Notice that once found, the exit label can be immediately identified since it is the only node besides the entrance with two neighbors.

\begin{figure}
    \center
    % The graph of glued_tree.tex, in the layout of fresh_oracle.tex.
\begin{tikzpicture}[
    scale=0.8,
    x=1cm, y=1cm,
    basicVertex/.style={circle, draw=black!75, fill=black!20,
        minimum size=14pt, inner sep=0pt, line width=0.9pt},
    basicEdge/.style={draw=black!80, line width=0.9pt}
]

% Two binary trees facing one another across the weld.
\foreach \v/\x/\y in {
    0/0/0, 1/1.8/1.8, 2/1.8/-1.8,
    3/3.6/2.7, 4/3.6/0.9, 5/3.6/-0.9, 6/3.6/-2.7,
    7/5.4/3.15, 8/5.4/2.25, 9/5.4/1.35, 10/5.4/0.45,
    11/5.4/-0.45, 12/5.4/-1.35, 13/5.4/-2.25, 14/5.4/-3.15,
    22/8.2/3.15, 23/8.2/2.25, 24/8.2/1.35, 25/8.2/0.45,
    26/8.2/-0.45, 27/8.2/-1.35, 28/8.2/-2.25, 29/8.2/-3.15,
    18/10/2.7, 19/10/0.9, 20/10/-0.9, 21/10/-2.7,
    16/11.8/1.8, 17/11.8/-1.8, 15/13.6/0
} { \node[basicVertex] (v\v) at (\x,\y) {}; }

% Tree edges.
\foreach \u/\v in {
    0/1, 0/2, 1/3, 1/4, 2/5, 2/6,
    3/7, 3/8, 4/9, 4/10, 5/11, 5/12, 6/13, 6/14,
    15/16, 15/17, 16/18, 16/19, 17/20, 17/21,
    18/22, 18/23, 19/24, 19/25, 20/26, 20/27, 21/28, 21/29%
} { \draw[basicEdge] (v\u) -- (v\v); }

% The same weld edges as in glued_tree.tex.
% White underlays keep crossing edges visually distinct.
\foreach \u/\v in {
    7/24, 7/28, 8/23, 8/27, 9/22, 9/28, 10/26, 10/29,
    11/25, 11/29, 12/23, 12/27, 13/22, 13/24, 14/25, 14/26%
} {
    \draw[basicEdge, preaction={draw=white, line width=2.9pt}]
        (v\u) -- (v\v);
}
% Distinct colors for the entrance, exit, and leaves of the two trees.
\node[basicVertex, draw=colorA, fill=colorA!25] at (v0) {};
\node[basicVertex, draw=colorB, fill=colorB!25] at (v15) {};
\foreach \v in {7,8,9,10,11,12,13,14,22,23,24,25,26,27,28,29} {
    \node[basicVertex, draw=colorC, fill=colorC!25] at (v\v) {};
}
\node[anchor=south, yshift=14pt, font=\small\bfseries] at (v0.north) {Entrance};
\node[anchor=south, yshift=14pt, font=\small\bfseries] at (v15.north) {Exit};
\end{tikzpicture}
    \caption{An instance of a welded tree graph. The \textcolor{colorA}{blue} node is the entrance, the \textcolor{colorB}{green} node is the exit and the \textcolor{colorC}{red} nodes are the leaves. The edges between the leaves are the weld edges.}
    \label{fig:basic_welded_tree}
\end{figure}
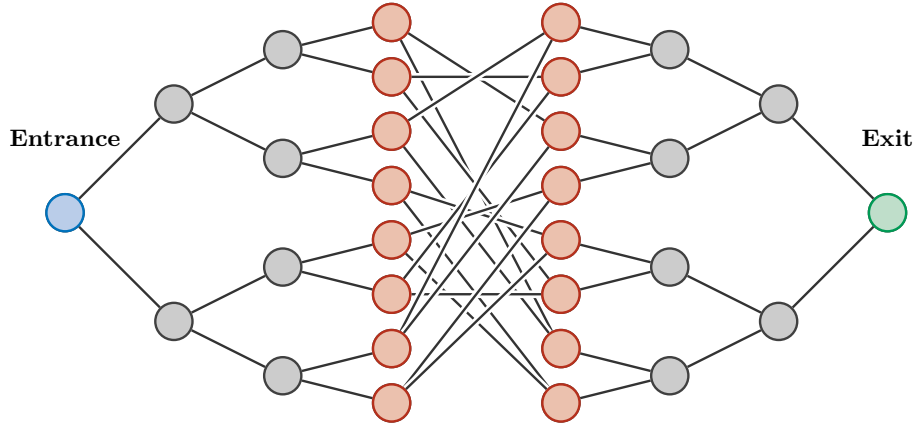

\cite{CFG02} first showed that quantum and random walks can have exponentially different traversal times using a version of $G_n$ with the leaves glued together. Subsequently,~\cite{CCD+03} established the exponential query separation between a quantum walk traversing $G_n$ and any classical algorithm. Intuitively, the classical hardness is due to the fact that after reaching the weld, the local information at each node does not tell the algorithm which way is towards the exit vertex. On the other hand, the quantum walk relies on the fact that its evolution from the entrance stays in the span of uniform superpositions over columns. This reduces the analysis to that of a walk on a line with $2n+2$ nodes, which allows one to show that outputting an exit with high probability is possible using $\poly(n)$ queries~\cite{CCD+03}. This illustrates how the structure of a graph can make a quantum walk evolve in a much smaller space than the vertex space and potentially lead to a speedup~\cite{KB07}.

Subsequent work has analysed the problem further. Jeffery and Zur introduce multidimensional quantum walks and obtain an $O(n)$-query quantum algorithm~\cite{JZ25}, Li, Li, and Luo give a succinct exact algorithm using coined quantum walks and amplitude amplification~\cite{LLL23}, and Belovs gives a short proof of the optimality of the linear hitting time using transducers~\cite{Bel24}. Finally, it has been shown that similar speedups hold for some generalizations of the welded trees graph~\cite{BLH25}. Nonetheless, the original result leaves a natural question.

\begin{question}\label{que:main_question}
    Does there exist an efficient quantum algorithm which outputs a path between the entrance and exit vertices of a welded tree?
\end{question}

\cref{que:main_question} is natural due to the fact that any random walk algorithm which outputs some vertex can output a path to it by tracking its steps. However, as mentioned in~\cite{Aar21}, there does not seem to be a trivial way to do this quantumly as keeping track of previously visited nodes changes the interference which the quantum walk relies on. Rosmanis explored the idea of tracking a path over some fixed length during a walk in order to finding an exit path, but this did not yield an efficient algorithm~\cite{Ros11}.

A partial negative answer to \cref{que:main_question} was established by Childs, Coudron, and Gilani~\cite{CCG23}. They showed that efficient quantum algorithms which are \emph{genuine} and \emph{rooted} can be simulated classically, meaning they cannot find an entrance-exit path. Roughly, rooted algorithms are those which always remember a path to the entrance, while genuine algorithms are restricted to certain operations on the vertex labels. However, general algorithms do not need to obey these restrictions. For example, an algorithm could try to first find the exit vertex using the quantum walk algorithm and then try to join two paths, each grown from a different root. Ruling out such strategies requires an argument which tracks what the algorithm retains during its computation.

Other papers focused on the quantum query complexity of pathfinding in various settings. For example,~\cite{LT25} give an $\tilde{O}(N^{1/3})$ quantum pathfinding algorithm for the problem, while~\cite{Li23,LZ25} construct families of graphs based on welded trees for which quantum algorithms produce a path exponentially faster than classical algorithms.

We give a negative answer to \cref{que:main_question} by proving an exponential quantum query lower bound for finding a path from the entrance to the exit.

\begin{theorem}[Informal version of \cref{thm:pathfinding_hard}]\label{thm:main_result}
    For any quantum algorithm $\alg$ making $q\leq \sqrt{N}/100$ queries to a welded trees graph oracle,
    \begin{align*}
        \Pr[\alg \text{ outputs an entrance-exit path}] = O\left( \frac{q^6}{N^{1/4}} \right).
    \end{align*}
\end{theorem}

In particular, any quantum algorithm must make at least $\Omega(N^{1/24})$ queries to the oracle to produce a path with constant probability. We note that the query model and distribution over instances in \cref{thm:main_result} are similar but not exactly the same as in~\cite{CCD+03}. We use an oracle model where one queries a label and color and receives another label if such neighbor exists, while the oracle in~\cite{CCD+03} returns all neighbors. Since we only use $3$ colors, the label-color oracle can query all colors to get a list of neighbors. In terms of the distribution,~\cite{CCD+03} samples a uniformly random cycle over welded edges, our distribution samples a cycle compatible with some fixed edge colorings and conditions on them being a single cycle.

\subsection{Technical overview}

Our proof is built on tracking the information an algorithm remembers about the graph during computation. The central challenge is showing that an algorithm cannot use interference to record a path from the entrance to the exit efficiently, even though a quantum walk can reach the exit while recording only its current position. For intuition, we focus on the case when $q=\poly(n)$, although the final bound applies beyond.

\paragraph{Recording graph information.} In order to record what an algorithm knows about the graph, we use the recording/compressed oracle framework of~\cite{Zha19}. Introduced as a technique which captures both the quantum query lower bound of search and collision problems, it has found a variety of other uses~\cite{CFHL21,HM23,BKW26,MH25,ACMT25,BHNZ25,GWWZ26,Car25}. For an introduction see~\cite{Ham25}. Intuitively, the technique describes the way a quantum algorithm interacts with a random oracle by tracking different input-output records in a database which represent the information the algorithm has received and remembers. Note that the database itself is quantum as the recording is being done in superposition. Unlike other approaches which track all the information an algorithm has queried, here we also allow to unquery, and hence forget, queries.

We apply the recording framework to welded trees using the recent construction of compressed permutation oracles of~\cite{Car25}, whose properties are used repeatedly throughout the paper. Specifically, we construct a graph oracle \wtO using two independent permutations, one which assigns random \emph{labels} to vertices, denoted $\labelPerm$, and one which permutes \emph{edges} with the same column and color in the middle of the graph, which are the nodes between columns $n/2$ and $3n/2$, denoted as $\colorPerm$. A single query to the graph oracle takes in some label $l$ and color $\alpha$, reverses the labeling, follows the edge along color $\alpha$ and then re-applies the label on the new vertex. We formally define this model in \cref{sec:construction_welded_trees}. Note that this construction is based on the construction in~\cite{CCG23} where they used a color-preserving permutation on the welded vertices.

Replacing $\labelPerm$ and $\colorPerm$ with compressed permutation oracles gives us two databases, a label database $L$ and an edge-color database $C$. We define the compressed welded trees oracle \cpwt in \cref{sec:compressed_welded_tree}, where we show that it closely approximates the original oracle. Specifically, we show that after $q$ queries the databases $(L,C)$ contain only $O(q)$ records and the probability of outputting a valid path is bounded by the probability such path is recorded in the database, up to a small error.

\paragraph{Fresh oracle.} Notice that the databases $(L,C)$ represent the subgraph $\rG$ which the algorithm has explored. Using a modified oracle \mcpwt we call \emph{fresh}, formally defined in \cref{sec:fresh_oracle}, we show that $\rG$ is always a forest, meaning it does not contain any cycles, and new edges cannot join distinct connected components in $\rG$. The oracle \mcpwt is defined by restricting the edge assignments in $C$ to endpoints which are outside of $\rG$. In the middle of the graph, each column contains at least $\sqrt{N}$ nodes, meaning that when the size of the database is much smaller, this excludes only a small fraction of the assignments. We show that the effect of this modification is exponentially small for $q=\poly(n)$. Note that the database records can still be erased or replaced, but any new edge assignment must be fresh. An example of how the oracle behaves can be seen in \cref{fig:fresh_oracle}.

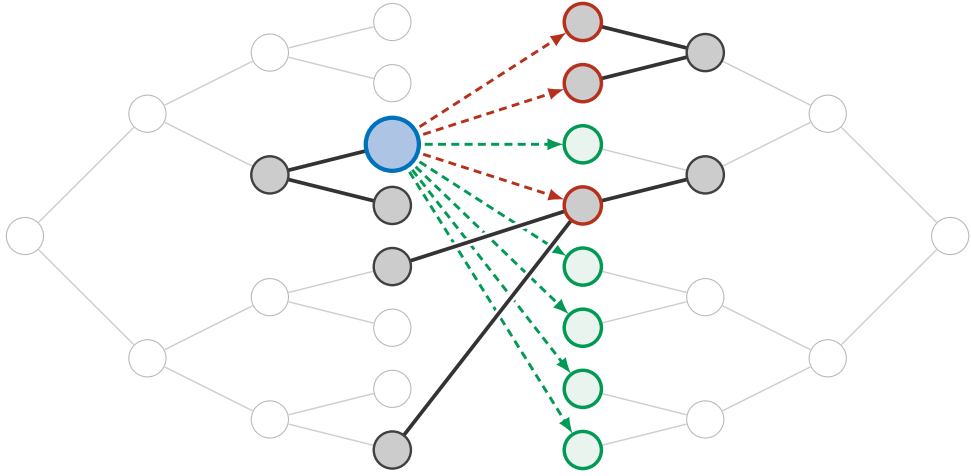
\begin{figure}
    \centering
    % The two-tree layout and internal vertex indices follow glued_tree.tex.
% Unrecorded weld edges are random: arrows show alternative assignments.
% The recorded database has identity assignments for the seven bold
% edges, with labels on their endpoints. Its support is
% {4,9,10,11,14,18,19,22,23,25}. The additional point 25 and its three
% neighbors enter the support through the recorded star at 25.
% At vertex 9 the beta candidate endpoints are
% {22,23,24,27,28}, and the gamma candidates are {22,24,25,26,28,29}.
% We draw their union. Destinations 22, 23, and 25 are excluded.
\begin{tikzpicture}[
    scale=0.9,
    x=1cm, y=1cm, >=latex,
    freshVertex/.style={circle, draw=black!28, fill=white,
        minimum size=14pt, inner sep=0pt},
    freshKnown/.style={freshVertex, draw=black!75, fill=black!20,
        line width=0.9pt},
    freshBlocked/.style={freshKnown, draw=colorC, line width=1.3pt},
    freshQuery/.style={freshKnown, draw=colorA, fill=colorA!30,
        line width=1.6pt, minimum size=20pt},
    freshRecorded/.style={draw=black!80, line width=1.4pt},
    freshContext/.style={draw=black!19, line width=0.55pt},
    freshCandidate/.style={->, densely dashed, line width=1.1pt}
]

% The two binary trees, with the weld between columns 3 and 4.
% Coordinates are structural positions, not information given to the algorithm.
\foreach \v/\x/\y in {
    0/0/0, 1/1.8/1.8, 2/1.8/-1.8,
    3/3.6/2.7, 4/3.6/0.9, 5/3.6/-0.9, 6/3.6/-2.7,
    7/5.4/3.15, 8/5.4/2.25, 9/5.4/1.35, 10/5.4/0.45,
    11/5.4/-0.45, 12/5.4/-1.35, 13/5.4/-2.25, 14/5.4/-3.15,
    22/8.2/3.15, 23/8.2/2.25, 24/8.2/1.35, 25/8.2/0.45,
    26/8.2/-0.45, 27/8.2/-1.35, 28/8.2/-2.25, 29/8.2/-3.15,
    18/10/2.7, 19/10/0.9, 20/10/-0.9, 21/10/-2.7,
    16/11.8/1.8, 17/11.8/-1.8, 15/13.6/0
} { \node[freshVertex] (v\v) at (\x,\y) {}; }

% All tree edges of G_3.
\foreach \u/\v in {
    0/1, 0/2, 1/3, 1/4, 2/5, 2/6,
    3/7, 3/8, 4/9, 4/10, 5/11, 5/12, 6/13, 6/14,
    15/16, 15/17, 16/18, 16/19, 17/20, 17/21,
    18/22, 18/23, 19/24, 19/25, 20/26, 20/27, 21/28, 21/29%
} { \draw[freshContext] (v\u) -- (v\v); }

% Two recorded three-vertex components.
\foreach \u/\v in {4/9, 4/10, 18/22, 18/23} {
    \draw[freshRecorded] (v\u) -- (v\v);
}

% Alternative random weld assignments across the two unrecorded colors.
% Every opposite leaf is eligible for at least one of these colors.
\foreach \v in {24,26,27,28,29} {
    \draw[freshCandidate, colorB] (v9) -- (v\v);
}
\foreach \v in {22,23,25} {
    \draw[freshCandidate, colorC] (v9) -- (v\v);
}

% A third recorded component centered at a right-hand leaf.
% Its incident edge records put all three neighbors in the support.
% White underlays distinguish recorded weld edges at line crossings.
\foreach \u/\v in {19/25,11/25,14/25} {
    \draw[freshRecorded, preaction={draw=white, line width=3.4pt}]
        (v\u) -- (v\v);
}

% Red is relative to the blue query: only forbidden candidate destinations.
% Other recorded vertices stay neutral, including the revisitable parent.
\foreach \v in {4,10,11,14,18,19} {
    \node[freshKnown] at (v\v) {};
}
\foreach \v in {22,23,25} {
    \node[freshBlocked] at (v\v) {};
}
\node[freshQuery] at (v9) {};
\foreach \v in {24,26,27,28,29} {
    \node[freshVertex, draw=colorB, fill=colorB!8,
        line width=1.3pt] at (v\v) {};
}
\end{tikzpicture}
    \caption{Example of a \emph{fresh} query of a welded edge. The \textcolor{colorA}{blue} node is the origin, \textcolor{black!75}{gray} nodes and edges represent the recorded information in $(L,C)$, \textcolor{colorB}{green} nodes are available connections and \textcolor{colorC}{red} nodes are forbidden under the freshness condition. Note that in the actual oracle, every column in the middle of the graph behaves this way.}
    \label{fig:fresh_oracle}
\end{figure}

Intuitively, this formalizes the idea that walking down from both roots of the oracle and connecting in the middle is not a viable strategy. Suppose that at the end of an algorithm, the database $(L,C)$ contains an entrance-to-exit path. Consider a welded edge which is in this path and the last time it was inserted (using a Feynman history path). The endpoint of this crossing was fresh, so the final path from it must have grown afterwards. Along this path segment, the edges have increasing times of last insertion. In \cref{sec:reduction_to_fixed_weld}, we use this idea to reduce the probability of recording a path using \mcpwt to the probability that after crossing a welded edge, such component reaches a boundary in the middle region, with an $O(q^2)$ multiplicative overhead due to a union bound over all paths. We call the event of reaching the column $n/2$ or $3n/2$ \say{escaping}.

\paragraph{Probability of escaping.} In \cref{sec:escaping_probability_bound}, we show that if one fixes a weld record, the insertion time and direction, the probability its component escapes while the record remains present is exponentially small even for a quantum algorithm.

For a classical algorithm, this is due to the fact that when going from a weld to a root, you must always move upwards as if you go downwards once (assuming you do not retrace), you will go back to the weld. For a fixed sequence of colors under random label and edge assignments, the probability that each step goes upwards is exponentially close to $\frac{1}{2}$. Extending this over $n/2$ upward steps means that escaping is exponentially unlikely. By union-bounding over the possible ways to start such a run, we have that the probability of escaping is exponentially small.

We quantize this by splitting different databases into groups. Fix some weld $\rho$. We split the graph $\rG$ into two parts, an inside and outside. The \emph{inside} part contains a general description of the tree which starts in $\rho$ using only sequences of colors and outside labels. We emphasize that we do not describe the actual vertices inside. The \emph{outside} part contains the subset of $(L,C)$ which is not captured by the inside. Notice that the representation of the inside part captures what the algorithm actually knows about the graph as the actual vertices are hidden.

For such inside/outside description, we consider a superposition over possible vertex assignments to the inside color sequences. Each has an amplitude equal to the square root of its probability under a random walk which starts from $\rho$. Furthermore, we retain only assignments which stay in the middle region and avoid collisions, which occurs with overwhelming probability due to the arguments above. The key technical step is showing that applications of \mcpwt approximately preserve this superposition, even as records are added, erased or replaced. Specifically, we consider different types of queries based on how they affect the database and show that mixing is not possible and errors from different groups are orthogonal, hence controlling the error even for different superpositions of database groups. Intuitively, keeping an oracle answer in the algorithm's workspace prevents the database record from being erased. However, in order to output a path, one has to actually write it down. 

Lastly, the initial distribution over welded tree graphs is not exactly the same as that in~\cite{CCD+03}. Specifically, we allow the welded edges to be assigned randomly as long as they keep the local structure the same. This means that the welded edges can be a union of multiple cycles as opposed to a single cycle. However, sampling pairs of graphs from both distributions and showing that they are hard to distinguish using the hybrid method, we find that the two cases are approximately equivalent. Combining all of these steps gives us \cref{thm:main_result}.

\paragraph{Recovering the quantum walk speedup.} Before diving into the formal proof, it is useful to understand how the database develops under a quantum walk. It was shown in~\cite{CFG02,CCD+03} that a continuous quantum walk traverses a welded tree in $\poly(n)$ steps by using the power of quantum interference. In \cref{sec:quantum_walk_analysis}, we formalize the idea of the power of \say{forgetting} by showing that a quantum walk using the compressed welded tree oracle only remembers its current label in the database, up to exponentially small error, by uncomputing the previous label and temporary edge record. This allows different computation paths to interfere.

\subsection*{AI Disclosure}

The main ideas and original proofs were developed by the authors, who wrote the manuscript in its entirety. ChatGPT 6 Astra and its predecessors were used in the creation of this document. Specifically, it found a much simpler way to prove the statements in \cref{sec:single_cycle,sec:general_intertwiner_lemma}, created \cref{fig:basic_welded_tree,fig:welded_tree,fig:color_preserving,fig:fresh_oracle} and was used to discuss whether our proofs are correct. The authors verified the correctness and originality of all content including references.

\section{Preliminaries}

Let us describe the notation used in this paper. For positive integers $a,b$ such that $a\leq b$, we use $[a]$ to denote the set $\{1,2,\dots,a\}$ and $[a,b]$ for $\{a,a+1,\dots,b\}$. The symbol $\rho$ will be used to denote density operators and for a pure state $\ket{\psi}$, $\rho(\ket{\psi}) = \ketbra{\psi}{\psi}$. All Hilbert spaces considered here are finite-dimensional. We use $\norm{\,\cdot \,}$ for the Euclidean norm of a vector or the induced operator norm. All projectors are orthogonal projectors. We use $\Sym(n)$ to denote the set of permutations of $[n]$, and $\phi\sim \Sym(n)$ denotes drawing a uniformly random permutation. Similarly, $\Sym(D)$ for a set $D$ is the set of permutations over $D$.

A graph $G = (V,E)$ consists of vertices $V$ and edges $E$ between the vertices. When $G$ is a rooted tree, we let $\groot(G)$ be the root of $G$ and for any $v\in V$, let $\depth(v)$ be the distance between $v$ and $\groot(G)$.

A quantum query algorithm $\alg^{(O)}$ which makes $q$ queries to the oracle $O$ is specified by a sequence of oracle-independent unitaries $A_0,\dots,A_q$ on its workspace, producing the final state $A_q O\cdots A_1 O A_0\ket{0}$. Access to $O$ is provided using an XOR oracle, defined as,
\begin{align*}
    \ket{x,y} \xrightarrow{O_f} \ket{x, y\oplus f(x)}.
\end{align*}

Note that we are assuming that every algorithm makes at least $1$ query. The following result will be useful later in the paper.

\begin{lemma}\label{lem:operator_norm_max}
    Let $\calH = \calH_1\oplus \dots \oplus \calH_k$ be a Hilbert space and $O: \calH\rightarrow \calH$ an operator such that $O = O_1 \oplus \dots \oplus O_k$ where for all $i\in [k]$, $O_i: \calH_i\rightarrow \calH_i$. Then,
    \begin{align*}
        \norm{O} = \max_{i\in [k]} \norm{O_i}
    \end{align*}
\end{lemma}

As our result extensively uses the compressed oracle in~\cite{Car25}, we follow their notation described below.

\paragraph{Injective databases.} We use $\db_N$ to denote the set of partial injective functions from $[N]$ to $[N]$ (i.e., functions $I : [N] \to [N] \cup \{\perp\}$ whose outputs in $[N]$ are pairwise distinct). For $I \in \db_N$, we write $\dom(I)$ and $\im(I)$ for its domain and image (both excluding $\perp$), and define its size $|I| = |\dom(I)|$. We write $\db_{\leq t} = \{I \in \db_N : |I| \leq t\}$. For $x\in[N]$ and $y \in [N] \cup \{\perp\}$, we define $I[x \to y]$ as,
\begin{align*}
    I[x \to y](x^\prime) \coloneq \begin{cases}
        I(x^\prime) &\text{If $x\neq x^\prime$ and either $I(x^\prime)\neq y$ or $y= \perp$}\\
        \perp &\text{If $x\neq x^\prime, y\neq \perp$ and $I(x^\prime)=y$}\\
        y &\text{If $x=x^\prime$.}
    \end{cases}
\end{align*}

The Hilbert space $\calH(\db_N)$ is spanned by orthonormal basis states $\{\ket{I} : I \in \db_N\}$, and we store the injective database in the register $\I$. We write $\Pi_{\leq t,\I}$ to be the following projector,
\begin{align*}
    \Pi_{\leq t, \I} \coloneq \sum_{I: \abs{I} \leq t} \ketbra{I}.
\end{align*}
The notation $\norm{\,\cdot \,}_{\leq t}$ means $\norm{\,\cdot \, \Pi_{\leq t, \I}}$.
When comparing two injective databases $I,J$, we say that $J\succeq I$ when for all $x\in \dom(I)$, we have $I(x)=J(x)$. We denote the set of databases which extend $I$ as $\Ext(I) \coloneq \{J: J \succeq I\}$. The projector $\Pi_{\succeq I}$ on all databases which \emph{include} $I$ is defined as,
\begin{align*}
    \Pi_{\succeq I} \coloneq \sum_{J: J\in \Ext(I)} \ketbra{J}.
\end{align*}

\subsection{Permutation oracles}\label{sec:permutation_oracles_definition}

Let us present permutation oracles and the compressed permutation oracle framework of~\cite{Car25}. We want to emphasize that throughout the paper, symbols such as $\I, \C$ and $\Lab$ symbolize registers, while $I,C$ and $L$ represent instances of databases (usually in these registers). Hence $\ket{I}_{\I}$ represents the database $I$ in register~$\I$.

\paragraph{Permutation oracle.} First, let us describe the oracle $O_\phi$ for $\phi\sim \Sym(N)$ on $b\in \{0,1\}$, $x\in[N]$, and $y\in\{0,1\}^d$,
\begin{align*}
    O_{\phi} \ket{b, x}_{\X} \ket{y}_\Y \coloneq \ket{b,x}_\X \ket{y \oplus \phi^{1-2b}(x)}_\Y,
\end{align*}
where the registers $\X,\Y$ describe the query input and output registers. In this case, the bit $b$ controls the direction of querying $x$. Let $\alg^{O_\phi}$ be a quantum algorithm which makes $q$ queries to the oracle $O_\phi$. The final mixed state of the algorithm is described as,

\begin{align*}
    \ket{\psi_q^{O_\phi}} &= A_q O_\phi \dots A_1 O_\phi A_0 \ket{0}\\
    \rho_{\alg}^{(O_{\phi})} &= \EX_{\phi \sim \Sym(N)}[\ketbra{\psi_q^{\phi}}{\psi_q^{\phi}}]
\end{align*}

\paragraph{Compressed permutation oracle.} Let us begin by defining several key terms. Let $\ket{\psi_1}, \ket{\psi_2}$ be arbitrary pure states such that $\braket{\psi_1}{\psi_2} =0$. Then $\pCfromPlus(\ket{\psi_1}, \ket{\psi_2})$ represents the compression primitive defined as,
\begin{align*}
    \pCfromPlus(\ket{\psi_1}, \ket{\psi_2}) \coloneq \Id - \ketbra{\psi_1} - \ketbra{\psi_2} + \ketbra{\psi_1}{\psi_2} + \ketbra{\psi_2}{\psi_1}.
\end{align*}
This operator swaps $\ket{\psi_1}$ and $\ket{\psi_2}$ and fixes their common orthogonal complement and is a self-adjoint unitary.

The database register $\I$ holds states in $\calH(\db_N)$, initialized to $\ket{\perp}_\I$, the trivial database mapping all inputs to $\perp$. In the forward operators below, we suppress the direction bit $b$. For $I \in \db_N$ and $x \notin \dom(I)$, let
\begin{align*}
    \ket{+_{x,I}} \coloneqq \frac{1}{\sqrt{N - |I|}} \sum_{y \in [N] \setminus \im(I)} \ket{I[x \to y]}
\end{align*}
be the uniform superposition over all non-colliding extensions of $I$ at $x$. Define
\begin{align*}
    \calH_{x,I} \coloneqq \Span\bigl(\{\ket{I}\}\cup\{\ket{I[x\to y]}:y\in[N]\setminus\im(I)\}\bigr).
\end{align*}
These subspaces form an orthogonal decomposition of $\calH(\db_N)$ for fixed $x$. On $\calH_{x,I}$, the \emph{permutation compression} operator $\pC_{x,I}$ swaps the trivially undefined state $\ket{I}$ with this uniform superposition:
\begin{align}
    \pC_{x,I} &\coloneqq \Id - \ketbra{+_{x,I}}{+_{x,I}} - \ketbra{I}{I} + \ketbra{+_{x,I}}{I} + \ketbra{I}{+_{x,I}}\label{eq:pc_x_I}\\
    &=\pCfromPlus(\ket{+_{x,I}}, \ket{I})\big|_{\calH_{x,I}}.
\end{align}
The key difference from the standard (function) compression operator is that the superposition ranges only over $[N] \setminus \im(I)$, enforcing injectivity at each step. The full operator $\pC$ is defined as,
\begin{align}
    \pC_x &\coloneq \bigoplus_{I: x \notin \dom(I)} \pC_{x,I}\\
    \pC &\coloneq \sum_{x\in[N]}\ketbra{x}_{\X}\otimes\pC_{x,\I}\label{eq:pc_def}
\end{align}
The \emph{extended purified oracle} $\xorOP$ XORs the current database value into the output register when the queried input is defined, and acts as the identity otherwise:
\begin{align*}
    \xorOP\ket{x}_\X\ket{y}_\Y\ket{I}_\I \coloneqq \begin{cases}
        \ket{x}_\X\ket{y \oplus I(x)}_\Y\ket{I}_\I & \text{if $x \in \dom(I)$,} \\
        \ket{x}_\X\ket{y}_\Y\ket{I}_\I & \text{otherwise.}
    \end{cases}
\end{align*}
Finally, the \emph{flip} operator $\cOp{F}$ inverts the injective database, $\cOp{F}\ket{I}_\I = \ket{I^{-1}}_\I$, where $I^{-1}(y) = x$ whenever $I(x) = y$. This allows inverse queries to be answered by first flipping the database, applying the forward-direction oracle, then flipping back.

The compressed permutation oracle \cpO is defined as,
\begin{align}
    \cpO \ket{b} \coloneq \ket{b} \otimes \begin{cases}
        \pC \cdot \xorOP \cdot \pC^{\dagger} &\text{If $b=0$}\\
        \cOp{F} \cdot \pC \cdot \xorOP \cdot \pC^{\dagger} \cOp{F}^{\dagger} &\text{otherwise.}
    \end{cases}\label{eq:cpo_definition}
\end{align}
The compressed oracle experiment initializes $\I$ to $\ket{\perp}_\I$ and replaces each call to $O_\phi$ with $\cpO$. The unitaries $A_i$ act only on $\A$, with the identity on $\I$ implicit. The joint state of the algorithm and database registers after $q$ queries is
\begin{align*}
    \ket*{\psi_q^{(\cpO)}}_{\A\I} \coloneqq A_q \cpO \cdots A_1 \cpO A_0 \ket{0}_\A\ket{\perp}_\I.
\end{align*}
Tracing out $\I$ gives the algorithm's reduced state
\begin{align*}
    \rho_\A^{(\cpO)} = \Tr_\I\!\left[\ketbra{\psi_q^{(\cpO)}}{\psi_q^{(\cpO)}}_{\A\I}\right].
\end{align*}
Next, let us state several properties of $\cpO$. 
\begin{theorem}[Soundness; Theorem 5.19 in~\cite{Car25}]\label{thm:soundness}
    For any $q$ query quantum algorithm,
    \begin{align*}
        \frac{1}{2}\norm{\rho_{\A}^{(O_{\phi})} - \rho_{\A}^{(\cpO)}}_1 = O\left(\frac{q^3}{N^{1/4}} \right)
    \end{align*}
\end{theorem}
While we will not use \cref{thm:soundness} directly, it is useful to keep in mind that we may essentially replace the standard permutation oracle with its compressed version at the cost of small error.
\begin{lemma}[Bounded growth; Lemma 3.2 in~\cite{Car25}]\label{lem:bounded}
    The final state of a $q$ query quantum algorithm using the compressed permutation oracle is supported on databases of size at most $q$,
    \begin{align*}
        \Pi_{\leq q, \I} \ket*{\psi_q^{(\cpO)}}_{\A\I} = \ket*{\psi_q^{(\cpO)}}_{\A\I}
    \end{align*}
\end{lemma}
\cref{lem:bounded} represents the intuitive idea that after $q$ queries, an algorithm only \say{knows} at most $q$ points.

To state the third property, write $\Pi_{(x,y) \in \I} \coloneq \sum_{I : I(x) = y} \ketbra{I}{I}$ for the projector onto databases containing the pair $(x, y)\in[N]^2$. For a list $\mathbf{z}=((x_1,y_1),\dots,(x_l,y_l))$ output by the algorithm, define
\begin{align*}
    \Pi_{\mathbf{z},\I}^{(\cpO)} &\coloneq \Pi_{(x_1,y_1)\in\I}\cdots\Pi_{(x_l,y_l)\in\I},\\
    \Pi_{\A\I}^{(\cpO)} &\coloneq \sum_{\mathbf{z}} \Pi_{\mathbf{z},\I}^{(\cpO)} \otimes \ketbra{\mathbf{z}}_{\A} \\
    \Pi_{\mathbf{z},\I}^{O_\phi} &\coloneq
    \left(\pC_{x_1}\Pi_{(x_1,y_1)\in\I}\pC_{x_1}^\dagger\right)\cdots
    \left(\pC_{x_l}\Pi_{(x_l,y_l)\in\I}\pC_{x_l}^\dagger\right).\\
    \Pi_{\A\I}^{O_\phi}&\coloneq \sum_{\mathbf{z}} \Pi_{\mathbf{z},\I}^{(O_\phi)} \otimes \ketbra{\mathbf{z}}_{\A}
\end{align*}
We find the following relationship between compressed states with these operators.
\begin{lemma}[Fundamental lemma; Lemma 3.3 in~\cite{Car25}]\label{lem:fundamental}
    For any algorithm which makes $q$ queries and outputs a list of pairs whose length is at most $l$,
    \begin{align*}
        \norm{\Pi_{\A\I}^{O_\phi}\ket*{\psi_q^{(\cpO)}}_{\A\I} } \leq \norm{\Pi_{\A\I}^{(\cpO)}\ket*{\psi_q^{(\cpO)}}_{\A\I} } + \frac{\sqrt{2}l}{\sqrt{N-q-l}}.
    \end{align*}
\end{lemma}
The following observation is useful to understand how querying the same value may behave.

\begin{observation}\label{obs:compressed_permutation_oracle_already_queried}
    Let $I\in\db_N$, $x\notin\dom(I)$, and $y\in[N]\setminus\im(I)$, and let $m=N - \abs{I}$. Then,
    \begin{align*}
        \pC_x \ket{I[x\to y]}_{\I} = \ket{I[x\to y]} - \frac{1}{\sqrt{m}} \ket{+_{x,I}} + \frac{1}{\sqrt{m}}\ket{I}
    \end{align*}
\end{observation}

Furthermore, it will be useful to be aware of the fact that two $\pC$ operators with similar $\ket{+}$ states behave similarly.

\begin{observation}\label{obs:superposition_similarity}
    Consider two finite sets $\calA, \calA^\prime$ of orthonormal basis labels such that $\abs{\calA}, \abs{\calA^\prime} \geq c>0$ and $\abs{(\calA \cup \calA^\prime) \setminus (\calA \cap \calA^\prime)} \leq b$. Letting $\ket{\calA}$ and $\ket{\calA^\prime}$ denote the normalized uniform superpositions over the sets,
    \begin{align*}
        \norm{\ket{\calA} - \ket{\calA^\prime}} \leq \sqrt{\frac{b}{c}}.
    \end{align*}
    Furthermore, for any database basis state $\ket{I}$ which is orthogonal to $\ket{\calA}$ and $\ket{\calA^\prime}$,
    \begin{align*}
        \norm{\pCfromPlus(\ket{\calA}, \ket{I}) - \pCfromPlus(\ket{\calA^\prime}, \ket{I})}  \leq 2 \sqrt{\frac{b}{c}}.
    \end{align*}
\end{observation}
\begin{proof}
    By expanding the square of the norm, we have,
    \begin{align*}
        \norm{\ket{\calA} - \ket{\calA^\prime}}^2 &= 2 - 2\Re \braket{\calA}{\calA^\prime}\\
        &= 2 -\frac{2 \abs{\calA \cap \calA^\prime}}{\sqrt{\abs{\calA} \abs{\calA^\prime}}}\\
        &= 2 -\frac{\abs{\calA} + \abs{\calA^\prime} - \abs{(\calA \cup \calA^\prime) \setminus (\calA \cap \calA^\prime)}}{\sqrt{\abs{\calA} \abs{\calA^\prime}}}\\
        &=\frac{\abs{(\calA \cup \calA^\prime) \setminus (\calA \cap \calA^\prime)} - (\sqrt{\abs{\calA}} - \sqrt{\abs{\calA^\prime}})^2}{\sqrt{\abs{\calA} \abs{\calA^\prime}}}\\
        &\leq \frac{\abs{(\calA \cup \calA^\prime) \setminus (\calA \cap \calA^\prime)}}{\sqrt{\abs{\calA} \abs{\calA^\prime}}} \leq \frac{b}{c}
    \end{align*}
    Taking the square root completes the first statement. For the second statement, expanding the compression operators and applying the triangle inequality gives,
    \begin{align*}
        \norm{\pCfromPlus(\ket{\calA}, \ket{I}) - \pCfromPlus(\ket{\calA^\prime}, \ket{I})} &\leq \norm{\ketbra{\calA^\prime} - \ketbra{\calA}} + \norm{(\ket{\calA} - \ket{\calA^\prime})\bra{I} + \ket{I}(\bra{\calA} - \bra{\calA^\prime})}\\
        &= \sqrt{1-\abs{\braket{\calA}{\calA^\prime}}^2} + \norm{\ket{\calA} - \ket{\calA^\prime}}\\
        &\leq 2\norm{\ket{\calA} - \ket{\calA^\prime}} \leq 2\sqrt{\frac{b}{c}},
    \end{align*}
    where the equality uses the fact that $\ket{I}$ is orthogonal to both states.
\end{proof}

Another useful statement is the following.

\begin{proposition}\label{prop:recompression_unlikely}
    Let $x,y\in[N]$ and $I$ be an injective database over $[N]$ such that $I(x)=y$, and $I_\bot = I\setminus \{(x,y)\}$. Then,
    \begin{align*}
        \norm{(\pC_x - \Id)\ket{I}}^2 = \frac{2}{N - \abs{I}+1}.
    \end{align*}
    More generally, fix $x$ and let $\zeta$ be a vector supported on databases $I$ with $x\in\dom(I)$ and $\abs{I}\leq r$, with an orthogonal record $\ket{I(x)}$ in a subregister of the algorithm register. Then,
    \begin{align*}
        \norm{(\pC_x - \Id)\zeta} \leq \sqrt{\frac{2}{N - r+1}}\norm{\zeta}.
    \end{align*}
\end{proposition}
\begin{proof}
    By \cref{obs:compressed_permutation_oracle_already_queried},
    \begin{align}
        \norm{(\pC_x - \Id)\ket{I}}^2 = \norm{\frac{\ket{I_\bot} - \ket{+_{x,I_\bot}}}{\sqrt{N-\abs{I_\bot}}}}^2 = \frac{2}{N - \abs{I_\bot}}=\frac{2}{N-|I|+1},\label{eq:recompression_unlikely_pt1}
    \end{align}
    where in the second equality we used the fact that $\ket{I_\bot}$ and $\ket{+_{x,I_\bot}}$ are orthogonal. For the second part of the statement, we may write $\zeta$ as,
    \begin{align*}
        \zeta = \sum_{y, I_\bot} (\zeta_{y,I_\bot}\otimes\ket{y})_{\A} \ket{I_\bot[x\to y]}_{\I},
    \end{align*}
    Therefore,
    \begin{align*}
        \norm{(\pC_x - \Id)\zeta}^2 &= \sum_{y,I_\bot} \frac{2}{N-\abs{I_\bot}} \norm{\zeta_{y,I_\bot}}^2 &\mbox{(By \cref{eq:recompression_unlikely_pt1})}\\
        &\leq \frac{2}{N - r + 1} \sum_{y,I_\bot} \norm{\zeta_{y,I_\bot}}^2 = \frac{2}{N - r + 1}\norm{\zeta}^2.
    \end{align*}
    Taking the square root completes the statement.
\end{proof}

\paragraph{Expanding the compression terms.} Let us expand the compression terms in more detail. Write $\db|_x=\{I\in\db_N:x\notin\dom(I)\}$. In this paragraph, we also allow $\pC_x$ to use a normalized uniform superposition over any nonempty set $\Omega(I)\subseteq[N]\setminus\im(I)$ in each block (for example, in \cpO, $\Omega(I) = [N] \setminus \im(I)$). Let us decompose it using the insertion operator $P_x^{+}$ and its adjoint $P_x^{-}$, which are defined as follows,
\begin{align*}
    P_x^{+}\ket{I} &\coloneq \begin{cases}
        \frac{1}{\sqrt{\abs{\Omega(I)}}} \sum_{y\in \Omega(I)} \ket{I[x\to y]} &\text{if $x\not\in \dom(I)$}\\
        0 &\text{if $x\in \dom(I)$.}
    \end{cases}\\
    P_x^{+} &=\sum_{I \in \db|_x} \ketbra{+_{x,I}}{I}\\
    P_x^{-} &\coloneq (P_x^{+})^{\dagger} = \sum_{I \in \db|_x} \ketbra{I}{+_{x,I}},
\end{align*}
where $\ket{+_{x,I}}$ now denotes the normalized uniform superposition over $\{\ket{I[x\to y]}:y\in\Omega(I)\}$.
Notice the following behavior of $P_x^{-}$ for $I\in\db|_x$ and $y\in\Omega(I)$,
\begin{align}
    P_x^{-}\ket{I[x\to y]} = \frac{1}{\sqrt{\abs{\Omega(I)}}}\ket{I}.\label{eq:a_x_destroy}
\end{align}
For $y\in([N]\setminus\im(I))\setminus\Omega(I)$, this action is zero. Furthermore, let us define the following projectors on the injective databases $\Pi_{x}^{\in}$ and $\Pi_{x}^{\not\in}$,
\begin{align*}
    \Pi_{x}^{\in} &\coloneq \sum_{I: x\in \dom(I)} \ket{I}\bra{I}\\
    \Pi_{x}^{\not\in} &\coloneq  \sum_{I: x\not\in \dom(I)} \ket{I}\bra{I} = \Id - \Pi_{x}^{\in} 
\end{align*}
Let us show how these operators decompose $\pC_{x}$. Note that we are fixing some specific database (such as $\InjColPres_{c,\alpha}$ or $\InjLabel$) onto which we are applying the compression operator $\pC$, which for each $x$ and $I$ such that $x\not\in \dom(I)$, switches between $\ket{I}$ and $\ket{+_{x,I}}$ while fixing everything orthogonal to $\Span\{\ket{I}, \ket{+_{x,I}}\}$.
\begin{lemma}\label{lem:expansion_pc_operator}
    Fix some operator $\pC_{x}$ being applied in $\pC$. Then we have that,
    \begin{align*}
        \pC_{x} \Pi_{x}^{\not\in} &= P_x^{+} \Pi_{x}^{\not\in}\\
        \pC_{x} \Pi_{x}^{\in} &= (\Id + P_x^{-} - P_x^{+} P_x^{-})\Pi_{x}^{\in}
    \end{align*}
\end{lemma}
\begin{proof}
    Let us fix some database $I$, assume that $x\not\in \dom(I)$ and recall the definition of $\pC_{x,I}$ in \cref{eq:pc_x_I}. Therefore,
    \begin{align*}
        \pC_{x}\ket{I} &= \pC_{x,I}\ket{I} = \ket{+_{x,I}} = P_x^{+}\ket{I}.
    \end{align*}
    Here $\ket{I}$ is the only undefined basis state in $\calH_{x,I}$, and $\pC_x$ is a direct sum over these orthogonal blocks. Next, let us focus on the second case. Assigned basis states with $y\notin\Omega(I)$ are fixed by $\pC_{x,I}$ and annihilated by $P_x^{-}$, so the second identity holds on those states. It remains to consider an arbitrary vector supported on the assigned states with $y\in\Omega(I)$ in one block,
    \begin{align*}
        \ket{\Psi} = \sum_{y\in \Omega(I)} \ket{\psi_y}_{\A} \otimes \ket{I[x\to y]}_{\I},
    \end{align*}
    where the $\A$ is some additional register (in our case, it contains the algorithm and other injective databases).\footnote{We are assuming that $\ket{\psi_y}$ is some arbitrary vector, so we can omit normalization factors.} Furthermore, let us define a $+$ state $\ket{\psi_+}$ which we use to decompose $\ket{\Psi}$,
    \begin{align*}
        \ket{\psi_+} &\coloneq \frac{1}{\sqrt{\abs{\Omega(I)}}} \sum_{y\in \Omega(I)} \ket{\psi_y},\\
        \ket{\Psi} &= \ket{\psi_+} \ket{+_{x,I}} + \sum_{y\in \Omega(I)} \left(\ket{\psi_y} - \frac{\ket{\psi_+}}{\sqrt{\abs{\Omega(I)}}} \right)_{\A} \otimes \ket{I[x\to y]}_{\I}.
    \end{align*}
    Notice that for any state $\ket{\phi}$ in $\A$, the second part of the equation is orthogonal to $\ket{+_{x,I}}$,
    \begin{align*}
        \left(\bra{\phi}_{\A} \otimes \bra{+_{x,I}}_{\I}\right)\sum_{y\in \Omega(I)} &\left(\ket{\psi_y} - \frac{\ket{\psi_+}}{\sqrt{\abs{\Omega(I)}}} \right)_{\A} \otimes \ket{I[x\to y]}_{\I}\\
        &= \frac{1}{\sqrt{\abs{\Omega(I)}}} \sum_{y\in \Omega(I)} \bra{\phi}\left(\ket{\psi_y} - \frac{\ket{\psi_+}}{\sqrt{\abs{\Omega(I)}}}\right)\\
        &=\frac{1}{\sqrt{\abs{\Omega(I)}}} \bra{\phi}\left(\sum_{y\in \Omega(I)} \ket{\psi_y} - \sqrt{\abs{\Omega(I)}} \cdot \ket{\psi_+}\right)\\
        &=\frac{1}{\sqrt{\abs{\Omega(I)}}} \bra{\phi}\left(\sqrt{\abs{\Omega(I)}} \cdot \ket{\psi_+} -  \sqrt{\abs{\Omega(I)}} \cdot \ket{\psi_+}\right) = 0.
    \end{align*}
    Therefore, applying $\pC_{x,I}$ on $\ket{\Psi}$ gives us,
    \begin{align*}
        \pC_{x,I}\ket{\Psi} = \ket{\psi_+}\ket{I} + \sum_{y\in \Omega(I)} \left(\ket{\psi_y} - \frac{\ket{\psi_+}}{\sqrt{\abs{\Omega(I)}}} \right)_{\A} \otimes \ket{I[x\to y]}_{\I}.
    \end{align*}
    On the other hand, let us consider the applications of $P_x^{-}$ and $P_x^{+} P_x^{-}$ on $\ket{\Psi}$,
    \begin{align*}
        (\Id \otimes P_x^{-}) \ket{\Psi} &= \sum_{y\in \Omega(I)} \ket{\psi_y}_{\A} \otimes \left(\ket{I}\braket{+_{x,I}}{I[x\to y]}\right)_{\I}\\
        &= \left(\frac{1}{\sqrt{\abs{\Omega(I)}}} \sum_y \ket{\psi_y} \right) \ket{I} = \ket{\psi_+} \ket{I}.\\
        (\Id \otimes P_x^{+} P_x^{-}) \ket{\Psi} &= \ket{\psi_+} P_x^{+}\ket{I}\\
        &=\sum_y \frac{\ket{\psi_+}}{\sqrt{\abs{\Omega(I)}}} \ket{I[x\to y]}
    \end{align*}
    Combining this, we have that,
    \begin{align*}
        (\Id \otimes (\Id + P_x^{-} - P_x^{+} P_x^{-}))\ket{\Psi} &= \sum_{y} \ket{\psi_y} \otimes \ket{I[x\to y]} + \ket{\psi_+} \ket{I} - \sum_y \frac{\ket{\psi_+}}{\sqrt{\abs{\Omega(I)}}} \ket{I[x\to y]}\\
        &= \ket{\psi_+} \ket{I} + \sum_{y\in \Omega(I)} \left(\ket{\psi_y} - \frac{\ket{\psi_+}}{\sqrt{\abs{\Omega(I)}}} \right) \otimes \ket{I[x\to y]}\\
        &=\pC_{x,I}\ket{\Psi}.\qedhere
    \end{align*}
\end{proof}

Therefore, the applications of the $\pC$ operators can be decomposed using these terms.

\paragraph{Intertwiner Lemma.}

The proof of \cref{thm:soundness} in~\cite{Car25} uses an \emph{intertwiner} to compare a purification of the real-oracle computation with $\ket*{\psi_q^{(\cpO)}}_{\A\I}$. Note that the construction in~\cite{Car25} assumes that the permutation $\phi$ is over an alphabet whose size is a power of $4$. As our oracles cannot satisfy this condition in general, we require the following generalization.

\begin{lemma}[General-input Intertwiner]\label{lem:general_intertwiner_lemma}
    For a finite nonempty set $S$, let $O_S$ denote oracle access to a uniformly random permutation of $S$ which is sampled at the start. Let $I_S$ store partial injective maps from $S$ to $S$, and let $\cpO_S$ be the corresponding compressed permutation oracle. There exist a purified realization $\cmfO_S$ of $O_S$, with private register $\regP_S$ initialized to $\ket{\bot}_{\regP_S}$, and an isometry $\calI_S: \calH(\regP_S) \to \calH(\regP_S) \otimes \calH(\I_S)$ such that, for an absolute constant $K$ and any state $\ket{\psi_q}$ obtained from $\ket{0}_{\A}\ket{\bot}_{\regP_S}$ using $q$ calls to $\cmfO_S$,
    \begin{align*}
        \norm{(\calI_S \cmfO_S - \cpO_S \calI_S) \ket{\psi_q}} \leq K\frac{(q+1)^2}{\abs{S}^{1/4}}.
    \end{align*}
    Additionally, $\calI_S$ is such that $\calI_S \ket{\bot}_{\regP_S} = \ket{\bot}_{\regP_S}\ket{\bot}_{I_S}$.
\end{lemma}

The proof of \cref{lem:general_intertwiner_lemma} may be found in \cref{sec:general_intertwiner_lemma}. We note that as shown in~\cite[Theorem 5.14]{Car25}, \cref{lem:general_intertwiner_lemma} implies \cref{thm:soundness} for an arbitrary domain.

\section{Welded trees oracles}

In this section, we use the permutation compressed oracle techniques in order to construct a welded trees instance. This allows us to bound the query complexity of finding a path between the start and exit vertex of a welded tree.

\subsection{Standard welded trees oracle}\label{sec:construction_welded_trees}

Assume that $n$ is even and $N\coloneq 2^n$. An instance of a welded tree $G_n = (V_n, E_n)$ consists of two binary trees $T_1,T_2$ of height $n+1$ whose leaves are joined by a bipartite cycle. We denote the set of edges on the cycle $E_W$ and call them the \say{welded edges}. For $c\in [0,2n+1]$, we define the column nodes $\col(c) \subseteq V_n$ as,
\begin{align*}
    \col(c) \coloneq \begin{cases}
        \{v\in V(T_1): \depth(v) = c\} &\text{if $c\leq n$.}\\
        \{v\in V(T_2): \depth(v) = 2n+1-c\} &\text{if $c\geq n+1$.}
    \end{cases}
\end{align*}

Notice that as $T_1, T_2$ are binary trees, $\abs{\col(c)} = 2^{\min(c, 2n+1-c)}$. Each vertex $v\in V_n$ is assigned a structural index $\strucIdx: V_n \to [0,2^{n+2}-3]$ as follows,
\begin{align*}
    \strucIdx(v) \coloneq \begin{cases}
        0 &\text{if $v=\groot(T_1)$.}\\
        2^{n+1} - 1 &\text{if $v=\groot(T_2)$.}\\
        2\cdot \strucIdx(v^\prime) + 1 &\text{$\col(v)\in [1,n]$ and $v$ is the \emph{left} child of $v^\prime$.}\\
        2\cdot \strucIdx(v^\prime) + 2 &\text{$\col(v)\in [1,n]$ and $v$ is the \emph{right} child of $v^\prime$.}\\
        2\cdot \strucIdx(v^\prime) - 2^{n+1} + 2 &\text{$\col(v)\in [n+1,2n]$ and $v$ is the \emph{left} child of $v^\prime$.}\\
        2\cdot \strucIdx(v^\prime) - 2^{n+1} + 3  &\text{$\col(v)\in [n+1,2n]$ and $v$ is the \emph{right} child of $v^\prime$.}
    \end{cases}
\end{align*}

For the remainder of the paper, we will use the vertex label $v\in V_n$ and the structural index $\strucIdx(v)$ interchangeably. We will use $p: V_n \to V_n\cup\{\bot\}$ to denote the parent of a vertex,
\begin{align*}
    p(w) \coloneq \begin{cases}
        \floor{(w - 1)/2} &\text{if $\col(w) \in [1,n]$.}\\
        \floor{(w-2^{n+1})/2} + 2^{n+1} - 1 &\text{if $\col(w) \in [n+1, 2n]$}\\
        \bot &\text{otherwise.}
    \end{cases}
\end{align*}

Next, let us assign colors from the set $\colorSet = \{{\color{colorA}\alpha},{\color{colorB}\beta},{\color{colorC}\gamma}\}$ to the edges.

\begin{observation}[Edge coloring $G_n$]\label{lem:color_welded_graph}
    For any $G_n$, there exists a proper edge coloring $\colF: E_n \to \colorSet$.
\end{observation}
\begin{proof}
    Notice that $G_n$ is a bipartite graph as the odd columns may form one partition and the even columns the other. As the maximum degree of $G_n$ is 3, Kőnig's theorem implies that the graph is 3-edge-colorable~\cite{Konig16}.
\end{proof}

For the remainder of the paper, we will fix some edge coloring $\colF$ which was obtained using \cref{lem:color_welded_graph}. An example of a welded tree instance using the procedures above may be found in \cref{fig:welded_tree}.

\begin{figure}
    \center
    \begin{tikzpicture}[
    every node/.style={circle, draw, fill=white, inner sep=0pt, minimum size=14pt, font=\tiny},
    treelabel/.style={rectangle, draw=none, fill=none, font=\normalsize\bfseries},
    collabel/.style={rectangle, draw=none, fill=none, font=\footnotesize, text=gray},
    >=latex
]

% ============================================================
%   Color definitions  (change these to recolor all edges)
% ============================================================

% Edge styles for each color
\tikzset{
    edgeA/.style={thick, colorA},
    edgeB/.style={thick, colorB},
    edgeC/.style={thick, colorC},
}

% ============================================================
%   Layout parameters
% ============================================================
\def\xsep{1.8}   % horizontal distance between columns
 
% ============================================================
%   LEFT TREE  (columns 0..3,  depth 3)
%   BFS labeling from root: 1, 2, 3, ..., 15
%
%   Tree-edge coloring (parent color → child colors):
%     L0:  →L1a = A,  →L1b = B
%     L1a: parent A,  →L2a = B,  →L2b = C
%     L1b: parent B,  →L2c = A,  →L2d = C
%     L2a: parent B,  →L3a = A,  →L3b = C
%     L2b: parent C,  →L3c = A,  →L3d = B
%     L2c: parent A,  →L3e = B,  →L3f = C
%     L2d: parent C,  →L3g = A,  →L3h = B
% ============================================================
 
% --- Col 0: root (entrance) ---
\node (L0) at (0, 0) {0};
 
% --- Col 1 ---
\node (L1a) at (\xsep,  1.0) {1};
\node (L1b) at (\xsep, -1.0) {2};
 
% --- Col 2 ---
\node (L2a) at (2*\xsep,  1.8) {3};
\node (L2b) at (2*\xsep,  0.6) {4};
\node (L2c) at (2*\xsep, -0.6) {5};
\node (L2d) at (2*\xsep, -1.8) {6};
 
% --- Col 3 (leaves) ---
\node (L3a) at (3*\xsep,  2.45) {7};
\node (L3b) at (3*\xsep,  1.75) {8};
\node (L3c) at (3*\xsep,  1.05) {9};
\node (L3d) at (3*\xsep,  0.35) {10};
\node (L3e) at (3*\xsep, -0.35) {11};
\node (L3f) at (3*\xsep, -1.05) {12};
\node (L3g) at (3*\xsep, -1.75) {13};
\node (L3h) at (3*\xsep, -2.45) {14};
 
% --- Left tree edges ---
\draw[edgeA] (L0) -- (L1a);
\draw[edgeB] (L0) -- (L1b);
 
\draw[edgeB] (L1a) -- (L2a);
\draw[edgeC] (L1a) -- (L2b);
\draw[edgeA] (L1b) -- (L2c);
\draw[edgeC] (L1b) -- (L2d);
 
\draw[edgeA] (L2a) -- (L3a);
\draw[edgeC] (L2a) -- (L3b);
\draw[edgeA] (L2b) -- (L3c);
\draw[edgeB] (L2b) -- (L3d);
\draw[edgeB] (L2c) -- (L3e);
\draw[edgeC] (L2c) -- (L3f);
\draw[edgeA] (L2d) -- (L3g);
\draw[edgeB] (L2d) -- (L3h);

% ============================================================
%   RIGHT TREE  (columns 4..7,  mirror of left tree)
%   BFS labeling from root: 16, 17, 18, ..., 30
%
%   Tree-edge coloring (same scheme as left tree):
%     R0:  →R1a = A,  →R1b = B
%     R1a: parent A,  →R2a = B,  →R2b = C
%     R1b: parent B,  →R2c = A,  →R2d = C
%     R2a: parent B,  →R3a = A,  →R3b = C
%     R2b: parent C,  →R3c = A,  →R3d = B
%     R2c: parent A,  →R3e = B,  →R3f = C
%     R2d: parent C,  →R3g = A,  →R3h = B
% ============================================================
 
% --- Col 4 (leaves) ---
\node (R3a) at (4*\xsep,  2.45) {22};
\node (R3b) at (4*\xsep,  1.75) {23};
\node (R3c) at (4*\xsep,  1.05) {24};
\node (R3d) at (4*\xsep,  0.35) {25};
\node (R3e) at (4*\xsep, -0.35) {26};
\node (R3f) at (4*\xsep, -1.05) {27};
\node (R3g) at (4*\xsep, -1.75) {28};
\node (R3h) at (4*\xsep, -2.45) {29};
 
% --- Col 5 ---
\node (R2a) at (5*\xsep,  1.8) {18};
\node (R2b) at (5*\xsep,  0.6) {19};
\node (R2c) at (5*\xsep, -0.6) {20};
\node (R2d) at (5*\xsep, -1.8) {21};
 
% --- Col 6 ---
\node (R1a) at (6*\xsep,  1.0) {16};
\node (R1b) at (6*\xsep, -1.0) {17};
 
% --- Col 7: root (exit) ---
\node (R0) at (7*\xsep, 0) {15};
 
% --- Right tree edges ---
\draw[edgeA] (R0) -- (R1a);
\draw[edgeB] (R0) -- (R1b);
 
\draw[edgeB] (R1a) -- (R2a);
\draw[edgeC] (R1a) -- (R2b);
\draw[edgeA] (R1b) -- (R2c);
\draw[edgeC] (R1b) -- (R2d);
 
\draw[edgeA] (R2a) -- (R3a);
\draw[edgeC] (R2a) -- (R3b);
\draw[edgeA] (R2b) -- (R3c);
\draw[edgeB] (R2b) -- (R3d);
\draw[edgeB] (R2c) -- (R3e);
\draw[edgeC] (R2c) -- (R3f);
\draw[edgeA] (R2d) -- (R3g);
\draw[edgeB] (R2d) -- (R3h);

% ============================================================
%   GLUING EDGES
%   Designed so that every degree-3 leaf vertex has all
%   three colors {A, B, C} on its incident edges.
%
%   Left leaf tree colors:           Right leaf tree colors:
%     L3a=A  L3b=C  L3c=A  L3d=B      R3a=A  R3b=C  R3c=A  R3d=B
%     L3e=B  L3f=C  L3g=A  L3h=B      R3e=B  R3f=C  R3g=A  R3h=B
%
%   Each leaf needs its two gluing edges to use the other
%   two colors.  The gluing below satisfies this at every
%   vertex (verified exhaustively).
% ============================================================
 
\draw[edgeB] (L3a) -- (R3c);
\draw[edgeC] (L3a) -- (R3g);
 
\draw[edgeA] (L3b) -- (R3b);
\draw[edgeB] (L3b) -- (R3f);
 
\draw[edgeC] (L3c) -- (R3a);
\draw[edgeB] (L3c) -- (R3g);
 
\draw[edgeA] (L3d) -- (R3e);
\draw[edgeC] (L3d) -- (R3h);
 
\draw[edgeC] (L3e) -- (R3d);
\draw[edgeA] (L3e) -- (R3h);
 
\draw[edgeB] (L3f) -- (R3b);
\draw[edgeA] (L3f) -- (R3f);
 
\draw[edgeB] (L3g) -- (R3a);
\draw[edgeC] (L3g) -- (R3c);
 
\draw[edgeA] (L3h) -- (R3d);
\draw[edgeC] (L3h) -- (R3e);
 
% ============================================================
%   LABELS
% ============================================================
 
% Entrance / Exit labels
\node[treelabel, above=4pt of L0] {Entrance};
\node[treelabel, above=4pt of R0] {Exit};
 
% Column labels
\node[collabel] at (0,       -3.2) {col 0};
\node[collabel] at (\xsep,   -3.2) {col 1};
\node[collabel] at (2*\xsep, -3.2) {col 2};
\node[collabel] at (3*\xsep, -3.2) {col 3};
\node[collabel] at (4*\xsep, -3.2) {col 4};
\node[collabel] at (5*\xsep, -3.2) {col 5};
\node[collabel] at (6*\xsep, -3.2) {col 6};
\node[collabel] at (7*\xsep, -3.2) {col 7};

% Title
% \node[rectangle, draw=none, fill=none, font=\large\bfseries] at (3.5*\xsep, 3.8)
%     {Glued Trees ($n = 3$)};

\end{tikzpicture}
    \caption{An instance of a welded tree $G_3$.}\label{fig:welded_tree}
\end{figure}
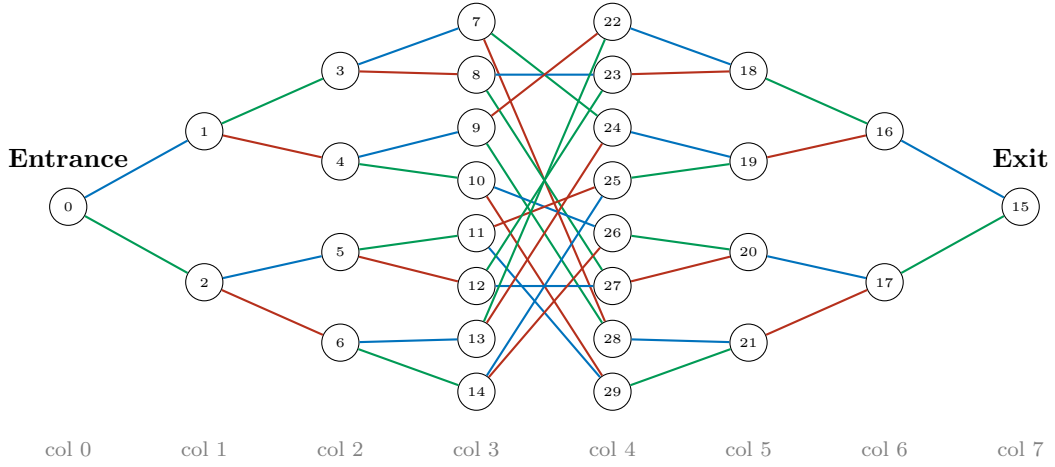

For all nodes except the roots of $T_1,T_2$, we define the parent color $\upcolor$ as the color that leads to a node's parent, $\upcolor(v) = \colF( \{v, p(v) \})$. We will partition the nodes $v$ in $G_n$ based on $\col(v)$ and $\upcolor(v)$. Specifically, for all $c\in [2n]$ and $\alpha\in \colorSet$, define the \emph{color class} $\colClass{c,\alpha}$ and its size $\colClassSize{c,\alpha}$ as,
\begin{align*}
    \colClass{c,\alpha} \coloneq \{v: \col(v) = c, \upcolor(v) = \alpha \}, \;\; \colClassSize{c,\alpha} \coloneq \abs{\colClass{c,\alpha}}
\end{align*}

We call the edges in column $c$ those between vertices in columns $c-1$ and $c$. The following lemma shows that the number of edges with a color is approximately $\tfrac{1}{3}$ of all edges in the column.

\begin{lemma}[Lemma 59 in~\cite{CCG23}]\label{lem:colors_per_column}
    Let $\alpha_1$ be the color which is not used by any edges to $\groot(T_1)$. For any color $\alpha^\prime \in \colorSet$, let $\gamma(\alpha^\prime, c)$ denote the number of edges colored $\alpha^\prime$ in column $c\in [n]$ of $T_1$. Then,
    \begin{align}
        \gamma(\alpha^\prime, c) = \begin{cases}
            \floor*{\frac{2^c}{3}} &\text{if $c$ is odd and $\alpha^\prime=\alpha_1$ or $c$ is even and $\alpha^\prime \neq \alpha_1$}\\
            \ceil*{\frac{2^c}{3}} &\text{otherwise.}
        \end{cases}\label{eq:colors_per_column}
    \end{align}
    For $T_2$ (i.e. $c\in [n+2,2n+1]$), replace $c$ in \cref{eq:colors_per_column} with $2n-c+2$.
\end{lemma}
\begin{proof}
    The proof is by induction on the edge-column $c$. When $c=1$, the graph contains two differently colored edges neither of which is equal to $\alpha_1$, meaning \cref{eq:colors_per_column} holds. Assume \cref{eq:colors_per_column} holds for some $c\in [n-1]$, $c$ is odd and $\alpha$ is some color. For edge-column $c+1$, the parent $v$ of each edge colored $\alpha$ must have $\upcolor(v) \neq \alpha$. Therefore,
    \begin{align*}
        \gamma(\alpha, c+1) &= \sum_{\alpha^\prime \neq \alpha} \gamma(\alpha^\prime, c)\\
        &=\begin{cases}
            2\cdot \ceil{\frac{2^c}{3}} &\text{if $\alpha = \alpha_1$.}\\
            \ceil{\frac{2^c}{3}} + \floor{\frac{2^c}{3}} &\text{otherwise,}
        \end{cases}\\
        &=\begin{cases}
            \ceil{\frac{2^{c+1}}{3}} &\text{if $\alpha = \alpha_1$.}\\
            \floor{\frac{2^{c+1}}{3}} &\text{otherwise.}
        \end{cases}
    \end{align*}
    The result when $c$ is even or $c\in [n+2,2n+1]$ follows analogously.
\end{proof}

Furthermore, it will be useful to partition the edges in $E_n$ based on their colors. For an arbitrary edge $e = \{u,v\}$ where $\col(u) > \col(v)$, we use $\head(e)$ and $\tail(e)$ to denote their direction,
\begin{align*}
    \head(e) &\coloneq \begin{cases}
        v &\text{if $e\in E_{T_1}$,}\\
        u &\text{if $e \in E_{T_2}$,}\\
        u &\text{if $e\in E_W$,}
    \end{cases}\\
    \tail(e) &\coloneq \begin{cases}
        u &\text{if $\head(e) = v$,}\\
        v &\text{otherwise.}
    \end{cases}
\end{align*}

Thus a tree edge is always directed from a child to its parent. Formally, for any edge $e\in E_{T_1} \cup E_{T_2}$, $\upcolor(\tail(e)) = \colF(e)$. For welded edges $E_W$, the order is arbitrary, so we set it to go from column $n$ to $n+1$. Letting $c\in [2n+1]$ and $\alpha\in \colorSet$,
\begin{align}
    \treeEdges &\coloneq \left\{e = \{u,v\} \in E_{T_1} \cup E_{T_2}: \col(\tail(e)) = c, \colF(e) = \alpha \right\}\\
    &= \left\{ \{v,p(v)\}: v\in \colClass{c,\alpha} \right\}\label{eq:tree_edges_def_alt} \\
    \weldedEdges &\coloneq \left\{ e\in E_W: \colF(e) = \alpha \right\}.\label{eq:welded_edges_def}
\end{align}

By \cref{eq:tree_edges_def_alt,eq:welded_edges_def}, we have that $\abs{\treeEdges} = \colClassSize{c,\alpha}$ and $\abs{\weldedEdges} = \sum_{\beta \in \colorSet \setminus \{\alpha\}} \colClassSize{n,\beta}$.

\paragraph{Constructing the oracle.} For the remainder of the section, we will construct the welded trees oracle $\wtO$ and describe its distribution. There are two independent sources of randomness, both of which will use permutations. First, we randomize the edges using $\colorPerm$ and then the labels using \labelPerm.

Let us begin by defining the edge permutation \colorPerm. Simply put, \colorPerm permutes over edges in each column which have the same color. Formally, we consider permutations $\pi: E_n \to E_n$ such that,
\begin{align}
    \forall c\in [2n+1], \alpha\in \colorSet: \pi(\treeEdges) = \treeEdges \text{ and } \pi(\weldedEdges) = \weldedEdges.\label{eq:color_preserving_definition}
\end{align}
We call permutations which satisfy \cref{eq:color_preserving_definition} \emph{color-preserving}, as all vertices $v$ remain on the same column and their parent color $\upcolor(v)$ does not change. 

We may decompose color-preserving permutations based on the permutations over each set of edges separately. We will be solely interested in permutations between columns $c\in [n/2,3n/2]$, which we call the \emph{middle columns} and denote by $\middleCol = [n/2, 3n/2]$. The entire set of permuted edges is denoted by $\middleEdge$,
\begin{align*}
    \middleEdge &\coloneq \{ (\mathrm{tr}, c, \alpha): c\in \middleCol, \alpha\in \colorSet \} \cup \{ (\mathrm{wld}, \alpha): \alpha\in \colorSet \}.
\end{align*}
We also define the interior vertices and boundary of the middle region by
\begin{align*}
    \middleNodes &\coloneq \{v\in V_n : n/2 < \col(v) < 3n/2\},\\
    \midBoundNodes &\coloneq \col(n/2) \cup \col(3n/2).
\end{align*}
For $b\in \middleEdge$, we let $\sigma_b: E_b \to E_b$ denote the permutation over the corresponding edges and $S_b$ be the set of all permutations. Finally, we define the edge permutation $\colorPerm$ and the set of all permutations $\colorPresSet$ as,
\begin{align*}
    \colorPerm &\coloneq \{ c_b\}_{b\in \middleEdge}\\
    \colorPresSet &\coloneq \prod_{b\in \middleEdge} S_b.
\end{align*}
Notice that when we permute edges using $\colorPerm$, we do not change the columns of the vertices. Therefore, if our oracle returns the structural index, we are leaking the current position in the graph, making the pathfinding problem trivial. In order to avoid this leak, we randomly assign each node with a label. Let $M \coloneq 2^{2n}$. The corresponding set of permutations is,
\begin{align*}
    S_M^{(0)} &\coloneq \{\pi \in S_{M \cup \{0\}}: \pi(0) = 0\}.
\end{align*}
We will always sample the labeling permutation $\labelPerm \sim S_M^{(0)}$. Notice that as $\labelPerm(0)=0$, we have that the entrance label is fixed and thus known to the algorithm. Furthermore, since $M \gg \abs{V_n}$, randomly guessing the label of a structural vertex has a negligible probability.

For $v\in [0,M]$ and $\alpha\in\colorSet$, we may define the edge-color function $\edgeVertex{\alpha}{v}$ to be the edge which contains $v$ in direction $\alpha$,
\begin{align}
    \edgeVertex{\alpha}{v} & \coloneq \begin{cases}
        e &\text{if $\exists e = \{u,v\} \in E_n$ s.t. $\colF(e) = \alpha$}\\
        \bot &\text{otherwise.}
    \end{cases}
\end{align}
Notice that $\edgeVertex{\alpha}{v}$ returns $\bot$ either when the label is not in $V_n$ or when we query the one color which the edges from the roots of the trees were not assigned by \cref{lem:color_welded_graph}. The edge-cell indicator $\edgeIndic(e)$ describes which permutation an edge belongs to,
\begin{align*}
    \edgeIndic(e) &\coloneq \begin{cases}
        (\mathrm{tr}, \col(\tail(e)), \colF(e)) &\text{if $e\in E_{T_1} \cup E_{T_2}$ and $\col(\tail(e)) \in \middleCol$,}\\
        (\mathrm{wld}, \colF(e)) &\text{if $e\in E_W$,}\\
        \bot &\text{otherwise.}
    \end{cases}
\end{align*}

Assume that we have some fixed $\colorPerm \in \colorPresSet$. Letting $v\in [M]$, $\alpha\in \colorSet$ and $e=\edgeVertex{\alpha}{v}$, the neighbor function $\ncF_{\colorPerm }$ is defined as,
\begin{align}
    \ncF_{\colorPerm}(v,\alpha) &\coloneq \begin{cases}
        v &\text{if $v\notin V_n$ or $e = \bot$,}\\
        \head(e)  &\text{if $e \neq \bot$, $\edgeIndic(\edgeVertex{\alpha}{v}) = \bot$ and $v=\tail(e)$,}\\
        \tail(e)  &\text{if $e \neq \bot$, $\edgeIndic(\edgeVertex{\alpha}{v}) = \bot$ and $v=\head(e)$,}\\
        \head(\sigma_{\edgeIndic(e)}(e)) &\text{if $\edgeIndic(e) \neq \bot$ and $v=\tail(e)$,}\\
        \tail(\sigma_{\edgeIndic(e)}^{-1}(e)) &\text{if $\edgeIndic(e) \neq \bot$ and $v=\head(e)$.}
    \end{cases}\label{eq:def_ncf_perm}
\end{align}
Notice that for a fixed color, the function is an involution,
\begin{align*}
    \ncF_{\sigma}(\ncF_{\sigma}(v,\alpha),\alpha) = v.
\end{align*}
Furthermore, this handles the boundary of \middleCol on columns $n/2$ and $3n/2$. For some $\sigma = (\labelPerm, \colorPerm) \in S_M^{(0)} \times \colorPresSet$, $l,l^\prime \in [0,M]$ and $\alpha\in\colorSet$, the welded trees oracle $\wtO_{\sigma}$ is defined as,
\begin{align}
    \wtF_{\sigma}(l, \alpha) &\coloneq \labelPerm( \ncF_{\sigma}( \labelPerm^{-1}(l), \alpha))\label{eq:oracle_function_definition} \\
    \wtO_\sigma \ket{l, \alpha, l^\prime} &\coloneq \ket*{l, \alpha, l^\prime \oplus \wtF_{\sigma}(l, \alpha)}\nonumber
\end{align}

When it is obvious that $\sigma$ is fixed, we will omit it from the notation and write $\wtF$ and $\wtO$.
A visualization of the effect of \colorPerm may be found in \cref{fig:color_preserving}.

\begin{figure}[h]
    \center
    \begin{tikzpicture}[
    every node/.style={circle, draw, fill=white, inner sep=0pt, minimum size=14pt, font=\tiny},
    treelabel/.style={rectangle, draw=none, fill=none, font=\small\bfseries},
    >=latex
]
 
\tikzset{
    edgeA/.style={semithick, colorA},
    edgeB/.style={semithick, colorB},
    edgeC/.style={semithick, colorC},
}
 
\def\xA{0}
\def\xB{0.8}
\def\xC{1.6}
\def\xD{2.4}
\def\xE{4.3}
\def\xF{5.1}
\def\xG{5.9}
\def\xH{6.7}
 
\def\yshrink{0.85}
\def\hgap{8.0}
 
% %%%%%%%%%%%%%%%%%%%%%%%%%%%%%%%%%%%%%%%%%%%%%%%%%%%%%%%%%%%%
%   PANEL 1:  Before pi  (original edges)
% %%%%%%%%%%%%%%%%%%%%%%%%%%%%%%%%%%%%%%%%%%%%%%%%%%%%%%%%%%%%
 
\begin{scope}
 
\node[treelabel] at ({(\xD+\xE)/2}, \yshrink*2.45+0.55) {Before};
 
% -- Nodes --
\node (L0) at (\xA, 0) {0};
\node (L1a) at (\xB,  \yshrink*1.0) {1};
\node (L1b) at (\xB, -\yshrink*1.0) {2};
\node (L2a) at (\xC,  \yshrink*1.8) {3};
\node (L2b) at (\xC,  \yshrink*0.6) {4};
\node (L2c) at (\xC, -\yshrink*0.6) {5};
\node (L2d) at (\xC, -\yshrink*1.8) {6};
\node (L3a) at (\xD,  \yshrink*2.45) {7};
\node (L3b) at (\xD,  \yshrink*1.75) {8};
\node (L3c) at (\xD,  \yshrink*1.05) {9};
\node (L3d) at (\xD,  \yshrink*0.35) {10};
\node (L3e) at (\xD, -\yshrink*0.35) {11};
\node (L3f) at (\xD, -\yshrink*1.05) {12};
\node (L3g) at (\xD, -\yshrink*1.75) {13};
\node (L3h) at (\xD, -\yshrink*2.45) {14};
\node (R3a) at (\xE,  \yshrink*2.45) {22};
\node (R3b) at (\xE,  \yshrink*1.75) {23};
\node (R3c) at (\xE,  \yshrink*1.05) {24};
\node (R3d) at (\xE,  \yshrink*0.35) {25};
\node (R3e) at (\xE, -\yshrink*0.35) {26};
\node (R3f) at (\xE, -\yshrink*1.05) {27};
\node (R3g) at (\xE, -\yshrink*1.75) {28};
\node (R3h) at (\xE, -\yshrink*2.45) {29};
\node (R2a) at (\xF,  \yshrink*1.8) {18};
\node (R2b) at (\xF,  \yshrink*0.6) {19};
\node (R2c) at (\xF, -\yshrink*0.6) {20};
\node (R2d) at (\xF, -\yshrink*1.8) {21};
\node (R1a) at (\xG,  \yshrink*1.0) {16};
\node (R1b) at (\xG, -\yshrink*1.0) {17};
\node (R0)  at (\xH,  0)            {15};
 
% -- Original left tree edges --
\draw[edgeA] (L0) -- (L1a);  \draw[edgeB] (L0) -- (L1b);
\draw[edgeB] (L1a) -- (L2a); \draw[edgeC] (L1a) -- (L2b);
\draw[edgeA] (L1b) -- (L2c); \draw[edgeC] (L1b) -- (L2d);
\draw[edgeA] (L2a) -- (L3a); \draw[edgeC] (L2a) -- (L3b);
\draw[edgeA] (L2b) -- (L3c); \draw[edgeB] (L2b) -- (L3d);
\draw[edgeB] (L2c) -- (L3e); \draw[edgeC] (L2c) -- (L3f);
\draw[edgeA] (L2d) -- (L3g); \draw[edgeB] (L2d) -- (L3h);
 
% -- Original right tree edges --
\draw[edgeA] (R0) -- (R1a);  \draw[edgeB] (R0) -- (R1b);
\draw[edgeB] (R1a) -- (R2a); \draw[edgeC] (R1a) -- (R2b);
\draw[edgeA] (R1b) -- (R2c); \draw[edgeC] (R1b) -- (R2d);
\draw[edgeA] (R2a) -- (R3a); \draw[edgeC] (R2a) -- (R3b);
\draw[edgeA] (R2b) -- (R3c); \draw[edgeB] (R2b) -- (R3d);
\draw[edgeB] (R2c) -- (R3e); \draw[edgeC] (R2c) -- (R3f);
\draw[edgeA] (R2d) -- (R3g); \draw[edgeB] (R2d) -- (R3h);
 
% -- Original gluing edges --
\draw[edgeB] (L3a) -- (R3c); \draw[edgeC] (L3a) -- (R3g);
\draw[edgeA] (L3b) -- (R3b); \draw[edgeB] (L3b) -- (R3f);
\draw[edgeC] (L3c) -- (R3a); \draw[edgeB] (L3c) -- (R3g);
\draw[edgeA] (L3d) -- (R3e); \draw[edgeC] (L3d) -- (R3h);
\draw[edgeC] (L3e) -- (R3d); \draw[edgeA] (L3e) -- (R3h);
\draw[edgeB] (L3f) -- (R3b); \draw[edgeA] (L3f) -- (R3f);
\draw[edgeB] (L3g) -- (R3a); \draw[edgeC] (L3g) -- (R3c);
\draw[edgeA] (L3h) -- (R3d); \draw[edgeC] (L3h) -- (R3e);
 
\end{scope}

% %%%%%%%%%%%%%%%%%%%%%%%%%%%%%%%%%%%%%%%%%%%%%%%%%%%%%%%%%%%%
%   PANEL 2:  After pi  (edges rewired by pi . ncF . pi^{-1})
%
%   pi:  5<->7,  8->10->14->8,  9<->13,  11->15->12->11,
%        20<->22,  23->29->25->23,  24<->28,  26->30->27->26
% %%%%%%%%%%%%%%%%%%%%%%%%%%%%%%%%%%%%%%%%%%%%%%%%%%%%%%%%%%%%
 
\begin{scope}[xshift=\hgap cm]
 
\node[treelabel] at ({(\xD+\xE)/2}, \yshrink*2.45+0.55) {After};
 
% -- Nodes (same positions and labels) --
\node (L0p) at (\xA, 0) {0};
\node (L1ap) at (\xB,  \yshrink*1.0) {1};
\node (L1bp) at (\xB, -\yshrink*1.0) {2};
\node (L2ap) at (\xC,  \yshrink*1.8) {3};
\node (L2bp) at (\xC,  \yshrink*0.6) {4};
\node (L2cp) at (\xC, -\yshrink*0.6) {5};
\node (L2dp) at (\xC, -\yshrink*1.8) {6};
\node (L3ap) at (\xD,  \yshrink*2.45) {7};
\node (L3bp) at (\xD,  \yshrink*1.75) {8};
\node (L3cp) at (\xD,  \yshrink*1.05) {9};
\node (L3dp) at (\xD,  \yshrink*0.35) {10};
\node (L3ep) at (\xD, -\yshrink*0.35) {11};
\node (L3fp) at (\xD, -\yshrink*1.05) {12};
\node (L3gp) at (\xD, -\yshrink*1.75) {13};
\node (L3hp) at (\xD, -\yshrink*2.45) {14};
\node (R3ap) at (\xE,  \yshrink*2.45) {22};
\node (R3bp) at (\xE,  \yshrink*1.75) {23};
\node (R3cp) at (\xE,  \yshrink*1.05) {24};
\node (R3dp) at (\xE,  \yshrink*0.35) {25};
\node (R3ep) at (\xE, -\yshrink*0.35) {26};
\node (R3fp) at (\xE, -\yshrink*1.05) {27};
\node (R3gp) at (\xE, -\yshrink*1.75) {28};
\node (R3hp) at (\xE, -\yshrink*2.45) {29};
\node (R2ap) at (\xF,  \yshrink*1.8) {18};
\node (R2bp) at (\xF,  \yshrink*0.6) {19};
\node (R2cp) at (\xF, -\yshrink*0.6) {20};
\node (R2dp) at (\xF, -\yshrink*1.8) {21};
\node (R1ap) at (\xG,  \yshrink*1.0) {16};
\node (R1bp) at (\xG, -\yshrink*1.0) {17};
\node (R0p)  at (\xH,  0)            {15};
 
% -- Rewired left tree edges --
% {1,2}A  {1,3}B  (unchanged)
\draw[edgeA] (L0p) -- (L1ap);  \draw[edgeB] (L0p) -- (L1bp);
% {2,4}B (same)   {2,7}C (was {2,5})
\draw[edgeB] (L1ap) -- (L2ap); \draw[edgeC] (L1ap) -- (L2dp);
% {3,6}A (same)   {3,5}C (was {3,7})
\draw[edgeA] (L1bp) -- (L2cp); \draw[edgeC] (L1bp) -- (L2bp);
% {4,10}A (was {4,8})   {4,13}C (was {4,9})
\draw[edgeA] (L2ap) -- (L3cp); \draw[edgeC] (L2ap) -- (L3fp);
% {5,8}A (was {5,10})   {5,12}B (was {5,11})
\draw[edgeA] (L2bp) -- (L3ap); \draw[edgeB] (L2bp) -- (L3ep);
% {6,11}B (was {6,12})  {6,9}C (was {6,13})
\draw[edgeB] (L2cp) -- (L3dp); \draw[edgeC] (L2cp) -- (L3bp);
% {7,14}A (same)  {7,15}B (same)
\draw[edgeA] (L2dp) -- (L3gp); \draw[edgeB] (L2dp) -- (L3hp);
 
% -- Rewired right tree edges --
% {16,17}A  {16,18}B  (unchanged)
\draw[edgeA] (R0p) -- (R1ap);  \draw[edgeB] (R0p) -- (R1bp);
% {17,19}B (same)   {17,22}C (was {17,20})
\draw[edgeB] (R1ap) -- (R2ap); \draw[edgeC] (R1ap) -- (R2dp);
% {18,21}A (same)   {18,20}C (was {18,22})
\draw[edgeA] (R1bp) -- (R2cp); \draw[edgeC] (R1bp) -- (R2bp);
% {19,29}A (was {19,23})   {19,28}C (was {19,24})
\draw[edgeA] (R2ap) -- (R3cp); \draw[edgeC] (R2ap) -- (R3fp);
% {20,25}A (same)   {20,27}B (was {20,26})
\draw[edgeA] (R2bp) -- (R3ap); \draw[edgeB] (R2bp) -- (R3ep);
% {21,26}B (was {21,27})   {21,24}C (was {21,28})
\draw[edgeB] (R2cp) -- (R3dp); \draw[edgeC] (R2cp) -- (R3bp);
% {22,23}A (was {22,29})   {22,30}B (same)
\draw[edgeA] (R2dp) -- (R3gp); \draw[edgeB] (R2dp) -- (R3hp);
 
% -- Rewired gluing edges --
% {8,23}C   {8,29}B
\draw[edgeC] (L3ap) -- (R3ap); \draw[edgeB] (L3ap) -- (R3gp);
% {9,24}A   {9,28}B   (unchanged)
\draw[edgeA] (L3bp) -- (R3bp); \draw[edgeB] (L3bp) -- (R3fp);
% {10,23}B   {10,25}C
\draw[edgeB] (L3cp) -- (R3ap); \draw[edgeC] (L3cp) -- (R3cp);
% {11,27}A   {11,30}C   (unchanged)
\draw[edgeA] (L3dp) -- (R3ep); \draw[edgeC] (L3dp) -- (R3hp);
% {12,26}C   {12,30}A   (unchanged)
\draw[edgeC] (L3ep) -- (R3dp); \draw[edgeA] (L3ep) -- (R3hp);
% {13,24}B   {13,28}A   (unchanged)
\draw[edgeB] (L3fp) -- (R3bp); \draw[edgeA] (L3fp) -- (R3fp);
% {14,25}B   {14,29}C
\draw[edgeB] (L3gp) -- (R3cp); \draw[edgeC] (L3gp) -- (R3gp);
% {15,26}A   {15,27}C   (unchanged)
\draw[edgeA] (L3hp) -- (R3dp); \draw[edgeC] (L3hp) -- (R3ep);
 
\end{scope}
 
\end{tikzpicture}
    \caption{An example of an application of \colorPerm on $G_3$ (\labelPerm is $\Id$).}\label{fig:color_preserving}
\end{figure}
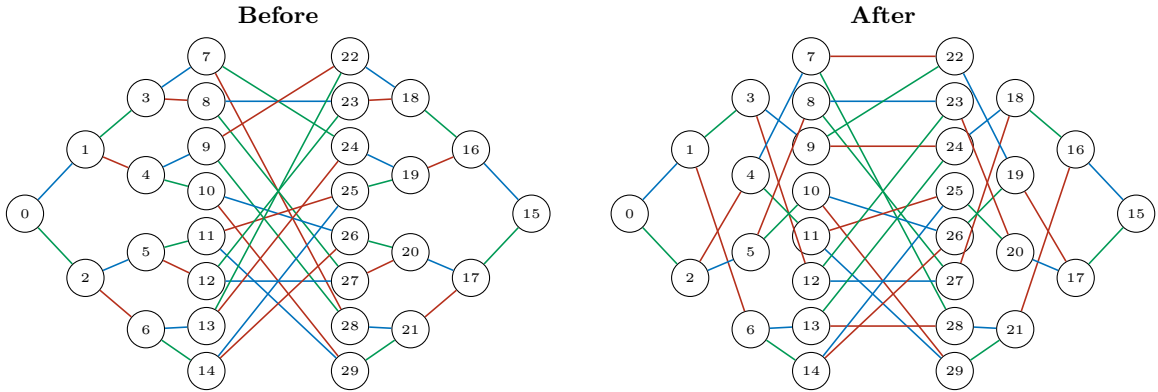

Suppose a query algorithm $\alg$ has access to \wtO. The final pure and mixed states of the algorithm are,
\begin{align}
    \ket{\psi^{(\wtO)}} &\coloneq A_q \wtO A_{q-1} \dots \wtO A_0 \ket{0}\\
    \rho_{\alg}^{(\wtO)} &\coloneq \EX_{\substack{\labelPerm \sim S_M\\\colorPerm \sim \colorPresSet}}[\rho(\ket{\psi}^{(\wtO)})]\label{eq:mixed_state_gtO_definition}
\end{align}

The construction of $\wtO$ provides us with the following observations.

\begin{observation}\label{obs:min_num_nodes_per_perm}
    Assume that $n\geq 16$. For every $b\in \middleEdge$ and every $\sigma_b: E_b \to E_b$,
    \begin{align*}
        \abs{E_b} \geq \frac{\sqrt{N}}{4}.
    \end{align*} 
\end{observation}
\begin{proof}
    The smallest permutations used are the ones over columns $c\in \{n/2, 3n/2\}$ and some color $\alpha\in \colorSet$. By \cref{lem:colors_per_column},
    \begin{align*}
        \abs{\treeEdges} \geq \colClassSize{c,\alpha} &\geq \floor{\tfrac{2^{n/2}}{3}} \geq \tfrac{\sqrt{N}}{4}.\qedhere
    \end{align*}
\end{proof}

\subsection{Compressed welded trees oracle}\label{sec:compressed_welded_tree}

In the following subsection, we construct the compressed welded trees oracle \cpwt and prove its soundness, bounded growth and fundamental properties. The construction will replace $\labelPerm$ and $\colorPerm$ with their compressed versions.

A label database is the partial injection $L: [M] \rightharpoonup [M]$. As with $\labelPerm$, we use the convention $L(0)=0$ to fix the entrance vertex and do not count it in $\dom(L)$. An edge (also called color) database $C$ is a family of partial injections $C = \{C_{b}\}_{b\in \middleEdge}$ where each $C_{b}:E_b \rightharpoonup E_b$. We denote the cardinality of $C$ by $\abs{C} = \sum_b \abs{C_b}$. We denote their respective registers by $\Lab$ and $\C$. A pair of the label and color databases is denoted using $I=(L,C)$ under the register $\I$.

For $e_1, e_2 \in E_b$, we write $C[e_1\to e_2]$ to denote extending the injection $C_b$ with $e\to f$ and keeping the rest of $C$ the same. As each edge has a unique label $\edgeIndic(e)$, we may omit the subscript $b$. Furthermore, $C_1 \preceq C_2$ denotes that for all $b\in \middleEdge$, $C_{1,b} \preceq C_{2,b}$.

Following the definition of $\wtF$ in \cref{eq:oracle_function_definition}, \cpwt first \emph{unlabels} using $\pCLabelInv$ (compression operator for $\labelPerm^{-1}$), followed by moving between columns using $\pCPerm$ (compression operator for $\colorPerm$ in $F^{\textup{perm}}$) and finally \emph{relabels} using $\pCLabel$ (compression operator for $\labelPerm$). All three operators perform the recording action. The query-input and query-output registers $\X$ and $\Y$ are written as $\ket{l,\alpha}_{\X} \ket{a}_{\Y}$ where $l,a\in [M]$ and $\alpha\in \colorSet$. The operator $\xorOP$ XORs the label of the adjacent structural vertex into $\Y$.

For a fixed partial edge database $C$, a structural index $v\in [0,M]$ and a color $\alpha\in \colorSet$, let $e = \edgeVertex{\alpha}{v}$ and $b=\edgeIndic(e)$. The partial neighbor map is defined as,
\begin{align}
    \ncF_C(v,\alpha) &\coloneq \begin{cases}
        v &\text{if $v\notin V_n$ or $e=\bot$,}\\
        \head(e) &\text{if $b=\bot$ and $v=\tail(e)$,}\\
        \tail(e) &\text{if $b=\bot$ and $v=\head(e)$,}\\
        \head(C_b(e)) &\text{if $b\in \middleEdge$, $v=\tail(e)$ and $e\in \dom(C_b)$,}\\
        \tail(C_b^{-1}(e)) &\text{if $b\in \middleEdge$, $v=\head(e)$ and $e\in \im(C_b)$,}\\
        \bot &\text{otherwise.}
    \end{cases}
\end{align}
Thus the map $\ncF_C$ uses $C_b$ exactly as $\ncF_{\colorPerm}$ uses $\sigma_b$ in \cref{eq:def_ncf_perm}. Furthermore, notice that when $\ncF_C(v,\alpha)\neq \bot$, we have that $\ncF_C(\ncF_C(v,\alpha),\alpha) = v$.

Let us define the inverse-label compression operator \pCLabelInv, which is essentially just an inverse call to the $\pC$ operator defined in \cref{eq:pc_def}. Letting $l\in [M]$ and $L$ being some label database such that $l\notin \im(L)$,
\begin{align}
    \ket{+_{l,L}^{\mathrm{inv}}} &\coloneq \frac{1}{\sqrt{M - \abs{L}}} \sum_{v\in [M]\setminus \dom(L)} \ket{L[v \to l]}\label{eq:plus_state_label_inv} \\
    \pC^{\mathrm{inv}}_{l,L} &\coloneq \pCfromPlus(\ket{+_{l,L}^{\mathrm{inv}}}, \ket{L})\nonumber\\ %\Id - \ketbra{+_{l,L}^{\mathrm{inv}}} - \ketbra{L} + \ketbra{+_{l,L}^{\mathrm{inv}}}{L} + \ketbra{L}{+_{l,L}^{\mathrm{inv}}}\\
    \pC^{\mathrm{inv}}_{l} &\coloneq \bigoplus_{L: l\notin \im(L)} \pC^{\mathrm{inv}}_{l,L}\nonumber\\
    \pCLabelInv \ket{l,\alpha}_{\X}\ket{L}_{\Lab} &\coloneq \begin{cases}
        \ket{l,\alpha}_{\X}\left( \pC^{\mathrm{inv}}_{l, \Lab} \ket{L}_{\Lab} \right) &\text{if $l\in [M]$,}\\
        \ket{0,\alpha}_{\X} \ket{L}_{\Lab} &\text{if $l=0$.}
    \end{cases}.\nonumber
\end{align}

Next, we define the compressed edge-permutation operator. Fix $b\in \middleEdge$ and $e\in E_b$. For an edge database $C$, we define the following,
\begin{align}
    \ket{+_{e,b,C}^{\rightarrow}} &\coloneq \frac{1}{\sqrt{\abs{E_b} - \abs{C_b}}} \sum_{e^\prime \in E_b \setminus \im(C_b)} \ket{C[e \to e^\prime]}\label{eq:plus_state_upcolor}\\
    \pC^{\rightarrow}_{e,b} &\coloneq \bigoplus_{C: e\notin \dom(C_b)} \pCfromPlus \left(\ket{+_{e,b,C}^{\rightarrow}}, \ket{C} \right)\label{eq:pC_right_direction}\\
    \ket{+_{e,b,C}^{\leftarrow}} &\coloneq \frac{1}{\sqrt{\abs{E_b} - \abs{C_b}}} \sum_{e^\prime \in E_b \setminus \dom(C_b)} \ket{C[e^\prime \to e]}\label{eq:plus_state_downcolor}\\
    \pC^{\leftarrow}_{e,b} &\coloneq \bigoplus_{C: e\notin \im(C_b)} \pCfromPlus \left(\ket{+_{e,b,C}^{\leftarrow}}, \ket{C} \right)\label{eq:pC_left_direction}
\end{align}
For $v\in V_n$ and $\alpha \in \colorSet$, let $e=\edgeVertex{\alpha}{v}$ and $b=\edgeIndic(e)$. Define,
\begin{align}
    \pC^\prime_{v,\alpha} &\coloneq \begin{cases}
        \pC^{\rightarrow}_{e,b} &\text{if $e\neq \bot$, $b\in \middleEdge$ and $v=\tail(e)$,}\\
        \pC^{\leftarrow}_{e,b} &\text{if $e\neq \bot$, $b\in \middleEdge$ and $v=\head(e)$,}\\
        \Id &\text{otherwise.}
    \end{cases}
\end{align}
The controlled edge-compression operator is defined as,
\begin{align}
    \pCPerm \ket{l,\alpha}_{\X} \ket{L}_{\Lab} \ket{C}_{\C} &\coloneq \begin{cases}
        \ket{l,\alpha}_{\X}\ket{L}_{\Lab} \left( \pC^\prime_{v,\alpha} \ket{C}_{\C} \right) &\text{if $l\in \im(L)$ and $v=L^{-1}(l)$,}\\
        \ket{l,\alpha}_{\X}\ket{L}_{\Lab} \ket{C}_{\C} &\text{otherwise.}
    \end{cases}
\end{align}

The forward-label compression operator is defined using a partial map $\stepFunct(l,\alpha, L, C)$,
\begin{align}
    \stepFunct(l,\alpha, L, C) &\coloneq \begin{cases}
        \ncF_C(L^{-1}(l), \alpha) &\text{if $l\in \im(L) \cup \{0\}$,}\\
        \bot &\text{if $l\notin \im(L) \cup \{0\}$.}
    \end{cases}\\
    \pC^{\mathrm{for}}_{w} &\coloneq \bigoplus_{L: w\notin \dom(L)} \pC_{w,L}\\
    \pCLabel\ket{l,\alpha}_{\X} \ket{L}_{\Lab} \ket{C}_{\C} &\coloneq \begin{cases}
        \ket{l,\alpha}_{\X} \left( \pC^{\mathrm{for}}_{w}\ket{L}_{\Lab}\right) \ket{C}_{\C} &\text{if $w=\stepFunct(l,\alpha, L, C) \notin \{\bot, L^{-1}(l), 0\}$,}\\
        \ket{l,\alpha}_{\X} \ket{L}_{\Lab} \ket{C}_{\C} &\text{if $w\in \{\bot, L^{-1}(l), 0\}$.}
    \end{cases}
\end{align}

Note that we are using the definition of $\pC_{x,I}$ from \cref{eq:pc_x_I}. Finally, we define the XOR operator $\xorOP$ as,
\begin{align*}
    \xorOP\ket{l,\alpha}_{\X}\ket{a}_{\Y} \ket{L,C}_{\I} \coloneq \begin{cases}
        \ket{l,\alpha}_{\X} \ket{a\oplus L(w)}_{\Y} \ket{L,C}_{\I} &\text{if $w = \stepFunct(l,\alpha, L, C) \in \dom(L) \cup \{0\}$,}\\
        \ket{l,\alpha}_{\X}\ket{a}_{\Y} \ket{L,C}_{\I} &\text{otherwise.}
    \end{cases}
\end{align*}

Therefore, we may define the main oracle we will be working with.

\begin{definition}[Compressed Welded Trees Oracle] Let $\I = (\Lab, \C)$ be the injective database registers which are initialized as $\ket{\bot}_{\I}$. The compressed welded trees oracle \cpwt is defined as,
    \begin{align}
        \cpwt \coloneq \pCLabelInv \cdot \pCPerm \cdot \pCLabel \cdot \xorOP \cdot \pCLabel^\dagger \cdot \pCPerm^\dagger \cdot \pCLabelInv^\dagger\label{eq:cpwt_definition}
    \end{align}
\end{definition}
For simplicity of notation, we use $\pcwt$ to denote the three compression operators in \cref{eq:cpwt_definition},
\begin{align*}
    \pcwt \coloneq \pCLabelInv \cdot \pCPerm \cdot \pCLabel.
\end{align*}
The final mixed and pure state of some $q$-query algorithm \alg using \cpwt is denoted by,
\begin{align*}
    \ket*{\psi_{\alg,q}^{(\cpwt)}} &\coloneq A_{q,\A} \cpwt_{\A\I} A_{q-1, \A} \dots \cpwt_{\A\I} A_{0,\A}\ket{0}_{\A} \ket{\bot}_{\I},\\
    \rho_{\alg,q}^{(\cpwt)} &\coloneq \Tr_{\I}\left[\rho\left( \ket*{\psi_{\alg,q}^{(\cpwt)}} \right)\right].
\end{align*}

\subsubsection{Properties of \cpwt}

Let us show a bounded growth lemma for \cpwt.\footnote{In case the stricter proof fails, we have a weaker version: Each $\pC$ operator in \cref{eq:cpwt_definition} adjusts the size of the corresponding database it is altering at most by $1$. As $L$ may be altered at most $4$ times and $C$ at most twice, $\abs{L} \leq 4q$ and $\abs{C} \leq 2q$.}

\begin{lemma}[Bounded growth of \cpwt]\label{lem:cpwt_bounded_growth} After $q$ queries to \cpwt, the sizes of the databases $L$ and $C$ in the support of the final state are bounded by $2q$ and $q$ respectively. Formally,
    \begin{align*}
        \Pi_{\leq 2q, \Lab} \Pi_{\leq q, \C}\ket{\psi_{\alg, q}} = \ket{\psi_{\alg, q}}.
    \end{align*}
\end{lemma}
\begin{proof}
    The proof may be shown by induction as at the start both databases are empty. Suppose that the query register is in the state $\ket{l,\alpha}_{\X}$ and let $L^\prime$ and $C^\prime$ be some fixed database pair basis after applying $\pCLabelInv$ and $\pCPerm$. Hence, the value $w = \stepFunct(l,\alpha, L^\prime,C^\prime)$ is fixed. If $w\notin \{\bot, L^{\prime,-1}(l)\}$, then $\pCLabel$ cannot alter $\stepFunct(l,\alpha, L^\prime,C^\prime)$ as $C^\prime$ is invariant under $\pCLabel$ and $L^{\prime,-1}(l)$ cannot be altered. On the other hand, if $w\in \{\bot, L^{\prime, -1}(l)\}$, the compression operator functions as the identity. Hence, during $\pCLabel \cdot \xorOP \cdot \pCLabel^\dagger$ we preserve the subspace $\calH(L|^w, C)$ and thus the database $L^\prime$ changes size by at most $1$.

    Next, let us consider the effect of $\pCPerm$. In this case, the specific operator depends on the value $v = L^{-1}(l)$. Per the argument above, this value cannot be changed by $\pCLabel$. As $C$ is invariant under $\pCLabel$ and $\xorOP$, we have that $\pCPerm \cdot \pCLabel \cdot \xorOP \cdot \pCLabel^\dagger \cdot \pCPerm^\dagger$ may only increase the size of $C$ by at most $1$. Similarly, as $L^{\prime,-1}(l)$ is not altered during this process, the $\pC$ operator in $\pCLabelInv$ always calls the same value $l$, and thus the size of the database during these calls may only change by $1$. Therefore, $L$ may change its size by $2$ at most, while $C$ only does so by $1$ at each application of $\cpwt$.
\end{proof}

We will also obtain a version of \cref{lem:fundamental} for \cpwt. Let us first define the necessary variables. Let $\pairPath = (l_0, \alpha_1, l_1,\alpha_2,\dots,\alpha_m, l_m)$ where $l_i\in [0,M]$ and $\alpha_i\in \colorSet$ be some proposed output path of length $m\geq 1$. We assume that $l_0=0$ to start from the entrance vertex and that the labels are pairwise distinct in order to avoid cycles. In particular, $l_i\neq 0$ for $i\geq 1$. Let us consider the following operators,
\begin{align*}
    \Pi^{(l_{j-1}, \alpha_j, l_j)} &= \sum_{L,C: L(\stepFunct(l_{j-1}, \alpha_j, L,C)) = l_j} \ketbra{L,C}_{\I} \\
    \Pi_{(\pairPath)}^{\cpwt} &\coloneq \prod_{j\in [m]} \Pi^{(l_{j-1}, \alpha_j, l_j)}\\
    \Pi_{(\pairPath)}^{\mathrm{wt}} &\coloneq  \prod_{j=m}^{1} \pcwt_{l_{j-1}, \alpha_j} \Pi^{(l_{j-1}, \alpha_j, l_j)} \pcwt_{l_{j-1}, \alpha_j}^{\dagger}.
\end{align*}
These operators verify that specific input-output pairs (which, in our case are determined by a path) agree with the oracle. The joint verification operators are,
\begin{align}
    \Pi_{\A\I}^{\cpwt} &\coloneq \sum_{\pairPath} \Pi_{(\pairPath), \I}^{\cpwt} \otimes \ketbra{l_0, \alpha_1, l_1,\dots, \alpha_m, l_{m}}{l_0, \alpha_1, l_1,\dots, \alpha_m, l_{m}}_{\A}\\
    \Pi_{\A\I}^{\mathrm{wt}} &\coloneq \sum_{\pairPath} \Pi_{(\pairPath), \I}^{\mathrm{wt}} \otimes \ketbra{l_0, \alpha_1, l_1,\dots, \alpha_m, l_{m}}{l_0, \alpha_1, l_1,\dots, \alpha_m, l_{m}}_{\A}\label{eq:project_wt_oracle_outputs}
\end{align}
These operators let us compare verification using the oracle with checking the compressed database directly.
\begin{lemma}[Fundamental lemma for \cpwt]\label{lem:fundamental_cpwt}
    For any state after $q$ queries with $q+m\leq\sqrt{N}/8$,
    \begin{align*}
        \norm{\Pi_{\A\I}^{\mathrm{wt}}\ket{\psi_{\alg, q}^{(\cpwt)}}} \leq \norm{\Pi_{\A\I}^{\cpwt}\ket{\psi_{\alg, q}^{(\cpwt)}}} + O\left(\frac{m}{N^{1/4}}\right).
    \end{align*}
\end{lemma}
\begin{proof}
    In this proof, write $\norm{B}_{\leq r}\coloneq\norm{B\Pi_{\leq2r,\Lab}\Pi_{\leq r,\C}}$, in accordance with \cref{lem:cpwt_bounded_growth}. Notice that,
    \begin{align*}
        \norm{\Pi_{\A\I}^{\mathrm{wt}}\ket{\psi_{\alg, q}^{(\cpwt)}}} - \norm{\Pi_{\A\I}^{\cpwt}\ket{\psi_{\alg, q}^{(\cpwt)}}} &\leq \norm{\left(\Pi_{\A\I}^{\mathrm{wt}} - \Pi_{\A\I}^{\cpwt} \right) \ket{\psi_{\alg, q}^{(\cpwt)}}}\\
        &\leq \norm{\Pi_{\A\I}^{\mathrm{wt}} - \Pi_{\A\I}^{\cpwt}}_{\leq q}.
    \end{align*}
    Both operators are controlled on some specific $\pairPath$, so by \cref{lem:operator_norm_max}, it suffices to fix some specific path of length $m$. Telescoping the products with the diagonal consistency projectors on the right, which preserve database size, gives
    \begin{align}
        \norm{\Pi_{(\pairPath)}^{\mathrm{wt}} - \Pi_{(\pairPath)}^{\cpwt}}_{\leq q} \leq \sum_{j\in [m]} \norm{\pcwt_{l_{j-1}, \alpha_j} \Pi^{(l_{j-1}, \alpha_j, l_j)} \pcwt_{l_{j-1}, \alpha_j}^{\dagger} - \Pi^{(l_{j-1}, \alpha_j, l_j)}}_{\leq q}\label{eq:sum_over_proj}.
    \end{align}
    Consider any arbitrary injective database $D$ and inputs $x,y \in [S]$. Then we have that,
    \begin{align*}
        \bra{D[x\to y]} \pC_{x,D} \ket{D[x\to y]} = 1 - \frac{1}{S- \abs{D}}.
    \end{align*}
    Therefore,
    \begin{align}
        \norm{\pC_{x,D} \Pi^{(x,y)\in D}\pC_{x,D} - \Pi^{(x,y)\in D}} \leq \sqrt{\frac{2}{S-\abs{D}}}.\label{eq:pc_step_interlace}
    \end{align}
    For each compression primitive, fix all database entries except the queried assignment. Since $l_i\neq l_{i-1}$ and the partial neighbor map is an involution, the consistency condition selects at most one assigned basis state in each block. These blocks are orthogonal, so applying \cref{eq:pc_step_interlace} blockwise and using the triangle inequality gives the bound below. The intermediate compression operators add at most one label and one edge assignment, so the primitive bounds below use the size cutoff $q+m$.
    \begin{align*}
        &\norm{\pcwt_{l_{j-1}, \alpha_j} \Pi^{(l_{j-1}, \alpha_j, l_j)} \pcwt_{l_{j-1}, \alpha_j}^{\dagger} - \Pi^{(l_{j-1}, \alpha_j, l_j)}}_{\leq q}\\ &\leq \norm{\pCLabelInv_{l_{j-1},\alpha_j} \Pi^{(l_{j-1}, \alpha_j, l_j)} \pCLabelInv_{l_{j-1},\alpha_j}^\dagger - \Pi^{(l_{j-1}, \alpha_j, l_j)}}_{\leq q+m}
        \\ &\quad + \norm{\pCPerm_{l_{j-1},\alpha_j} \Pi^{(l_{j-1}, \alpha_j, l_j)} \pCPerm_{l_{j-1},\alpha_j}^\dagger - \Pi^{(l_{j-1}, \alpha_j, l_j)}}_{\leq q+m}\\ &\quad +\norm{\pCLabel_{l_{j-1},\alpha_j} \Pi^{(l_{j-1}, \alpha_j, l_j)} \pCLabel_{l_{j-1},\alpha_j}^\dagger - \Pi^{(l_{j-1}, \alpha_j, l_j)}}_{\leq q+m} \\
        &\leq 2\sqrt{\frac{2}{M - 2(q+m)}} + \sqrt{\frac{2}{\tfrac{\sqrt{N}}{4} - q - m}} = O\left(\frac{1}{N^{1/4}}\right),
    \end{align*}
    where the last inequality follows due to \cref{obs:min_num_nodes_per_perm}. The statement follows by combining the result above with \cref{eq:sum_over_proj}.
\end{proof}

Next, we will use the following statement which shows a connection between the standard and compressed oracle.

\begin{theorem}\label{thm:intertwining_oracles}
    There exists a purification $\ket{\psi_{\alg,t}^{(\cmfWT)}}_{\A\regP}$ of $\rho^{(\wtO)}_{\A}$ and an intertwiner $\calI_{\regP}: \calH(\regP) \to \calH(\regP) \otimes \calH(\I)$ such that,
    \begin{align*}
        \norm{(\calI_{\regP} \cdot \cmfWT - \cpwt \cdot \calI_{\regP})\ket{\psi_{\alg,t}^{(\cmfWT)}}} = O\left( \frac{t^2}{N^{1/8}} \right).
    \end{align*}
\end{theorem}

The proof of \cref{thm:intertwining_oracles} may be found in \cref{sec:intertwining_proof}. We note \cref{thm:intertwining_oracles} also directly implies a version of \cref{thm:soundness} for \cpwt. Next, we will use \cref{thm:intertwining_oracles} to show the relationship between solving a problem using \wtO and checking whether a solution appears in the compressed database when using \cpwt. For a database $I=(L,C)$ and a colored candidate path $\pairPath=(l_0,\alpha_1,l_1,\dots,\alpha_m,l_m)$ as above, we write $\pairPath \in I$ if
\begin{align*}
    l_0=0,\qquad L(\groot(T_2))=l_m,\qquad
    L(\stepFunct(l_{i-1},\alpha_i,L,C))=l_i\quad\text{for every $i\in[m]$.}
\end{align*}
Thus $\pairPath \in I$ means that $I$ records this particular colored path from entrance to exit. We define
\begin{align}
    \calR &\coloneq \{I:\exists\pairPath\text{ such that }\pairPath \in I\},\label{eq:databases_which_exit} \\
    I\cap \calR &\coloneq \left\{(L^\prime,C^\prime): (L^\prime,C^\prime)\in \calR, L^\prime \preceq L, C^\prime \preceq C\right\}\\
    \calR_{\mathrm{vis}} &\coloneq \{\pairPath:\exists I\text{ such that }\pairPath \in I\}.
\end{align} 
We define $\Pi_\calR$ as the projector on injective databases which contain a path,
\begin{align*}
    \Pi_\calR \ket{I}_{\I} &\coloneq \begin{cases}
        \ket{I}_{\I} &\text{if $I \cap \calR \neq \emptyset$,}\\
        0 &\text{otherwise.}
    \end{cases}
\end{align*}
Note that this definition also extends over the permutations $(\labelPerm, \colorPerm)$ as they are instances of injective databases. Similarly to~\cite{Car25}, we may compare the success probability of a pathfinding algorithm $\alg$ using the following \say{games}.

\begin{algobox}{Standard Oracle Pathfinding}\alglabel{alg:standard_game}
    \begin{enumerate}
        \item Choose uniformly random permutations $\sigma = (\labelPerm, \colorPerm)$ in order to construct $\wtO_{\sigma}$.
        \item Run \alg which makes $q$ queries to $\wtO_{\sigma}$ and outputs some candidate path $\pairPath = (l_0, \alpha_1, l_1, \dots ,\alpha_m, l_m)$.
        \item\label{itm:standard_game_verification} The verifier outputs $1$ if for all $i\in [m]$, $\wtF_{\sigma}(l_{i-1}, \alpha_i) = l_i$, $l_m$ is the exit and $\pairPath\in \calR_{\mathrm{vis}}$.
    \end{enumerate}
\end{algobox}

Let \probWTO denote the probability \cref{alg:standard_game} succeeds. Similarly, we may define a compressed version of \cref{alg:standard_game}.

\begin{algobox}{Pathfinding using Compressed Oracle}\alglabel{alg:compressed_game}
    \begin{enumerate}
        \item Initiate the compressed databases $I=(L,C)$ to $\ket{\bot}_{\I}$.
        \item Run \alg which makes $q$ queries to $\cpwt$.
        \item Apply the measurement $\calM_{\calR} = \{ \Pi_{\calR}, \Pi_{\calR}^{\bot}\}$ to $\I$, accepting if and only if we measure the first outcome.
    \end{enumerate}
\end{algobox}

We denote the success probability of \cref{alg:compressed_game} by \probCWT.

\begin{theorem}\label{thm:success_probability_standard_compressed}
    For any algorithm $\alg$ which makes $q\leq \tfrac{\sqrt{N}}{100}$ queries,
    \begin{align*}
        \probWTO = O\left( \probCWT + \frac{q^6}{N^{1/4}}  \right).
    \end{align*}
\end{theorem}
\begin{proof}
    Without loss of generality we may assume the algorithm outputs $x\in\calR_{\mathrm{vis}}$, as otherwise it fails and that $q\geq 1$ as otherwise it does not know anything about the instance. Let $\ket*{\psi_{\alg, q}^{(\cdot)}}$ be the final state of $\alg$, where $\cdot$ is one of $\cpwt$ or $\wtO$.

    After $q$ queries, \cref{lem:cpwt_bounded_growth} tells us that $\abs{L}\leq 2q$, so no database can contain a path with more than $2q$ nonzero labels. Therefore the verification process can be cut off after $2q+1$ steps which by the argument we explain below may only succeed with probability at most $O(q^6/N^{1/4})$. Therefore we may assume that the output length is $m\leq 2q$.

    Let $\alg^\prime$ be the modification of \alg which at the end performs the verification procedure Step~\ref{itm:standard_game_verification} in \cref{alg:standard_game}. Formally, we first perform the queries $\wtO(l_{i-1}, \alpha_i)$ and store the results in some register $\regO$. For the exit check, the verifier additionally queries $\labelPerm^{-1}(l_m)$ into $\regO_{m+1}$. Afterwards, we accept if and only if $\regO$ contains $(l_1, \dots, l_m,\groot(T_2))$. Let us use $\ket*{\Psi_{\alg^\prime, q+m+1}^{(\cdot)}}$ to denote the resulting state and let $\Pi_{O}$ be the corresponding accepting projector. Therefore,
    \begin{align*}
        \probWTO \leq \norm{\Pi_{O}  \ket*{\Psi_{\alg^\prime, q+m+1}^{(\wtO)}}}^2.
    \end{align*}
    Note that we are assuming that $\ket*{\Psi_{\alg^\prime, q+m+1}^{(\wtO)}}$ is a pure state which has been purified according to \cref{thm:intertwining_oracles}. The product intertwiner also handles the verifier's additional inverse-label query by \cref{lem:general_intertwiner_lemma}, with an error of $O((q+m)^2/M^{1/4})$ absorbed below. Therefore,
    \begin{align*}
        \sqrt{\probWTO} &\leq \norm{ \calI_{\regP} \, \Pi_{O}\, \ket*{\Psi_{\alg^\prime, q+m+1}^{(\wtO)}}}\\
        & = \norm{\Pi_{O}\, \calI_{\regP}\, \ket*{\Psi_{\alg^\prime, q+m+1}^{(\wtO)}}}\\
        &\leq \norm{\Pi_{O} \ket*{\Psi_{\alg^\prime, q+m+1}^{(\cpwt)}}} + O\left(\frac{(q+m)^3}{N^{1/8}} \right),
    \end{align*}
    where we used the fact that since $\Pi_O$ and $\calI$ are on different registers, they commute. Suppose that for $i\in[m]$, each register $\regO_i$ in $\regO$ contains the output of the query $(l_{i-1}, \alpha_i)$ and let $\Pi_{O_i}$ be the corresponding projector on $\regO_i$. Then the edge verification is as follows,
    \begin{align*}
        \bra{l_i}_{\regO_i} \cpwt_{l_{i-1}, \alpha_i} \ket{0}_{\regO_i} &= \pcwt_{l_{i-1}, \alpha_i} \left[ \bra{l_i}_{\regO_i} \xorOP_{l_{i-1},\alpha_i} \ket{0}_{\regO_i}  \right] \pcwt_{l_{i-1}, \alpha_i}^\dagger\\
        &= \pcwt_{l_{i-1}, \alpha_i} \Pi^{(l_{i-1}, \alpha_i, l_i)}_{\I} \pcwt_{l_{i-1}, \alpha_i}^\dagger,
    \end{align*}
    meaning that we may replace the edge verification procedure in $\alg^\prime$ with $\Pi^{\mathrm{wt}}_{\A\I}$. For the exit check, let $E_{\A\I}$ project onto $L(\groot(T_2))=l_m$ and let $E^{\mathrm{wt}}_{\A\I}=\pC_{l_m}^{\mathrm{inv}}E_{\A\I}\pC_{l_m}^{\mathrm{inv}}$, controlled by the final output label. By \cref{eq:pc_step_interlace},
    \begin{align*}
        \norm{E^{\mathrm{wt}}_{\A\I}-E_{\A\I}}_{|L|\leq 2(q+m)}
        \leq \sqrt{\frac{2}{M-2(q+m)}}=O(M^{-1/2}).
    \end{align*}
    Therefore, using this estimate and the operator-norm estimate proved in \cref{lem:fundamental_cpwt},
    \begin{align*}
        \sqrt{\probWTO} &\leq \norm{\Pi_{O} \ket*{\Psi_{\alg^\prime, q+m+1}^{(\cpwt)}}} + O\left(\frac{(q+m)^3}{N^{1/8}} \right)\\
        &= \norm{E^{\mathrm{wt}}_{\A\I}\Pi^{\mathrm{wt}}_{\A\I} \ket*{\psi_{\alg, q}^{(\cpwt)}}}+ O\left(\frac{(q+m)^3}{N^{1/8}} \right)\\
        &\leq \norm{E_{\A\I}\Pi^{\cpwt}_{\A\I} \ket*{\psi_{\alg, q}^{(\cpwt)}}} + O\left(\frac{m}{N^{1/4}} + \frac{(q+m)^3}{N^{1/8}}\right)\\
        &\leq \sqrt{\probCWT} + O\left(\frac{q^3}{N^{1/8}}\right).
    \end{align*}
    The last inequality uses $m\leq2q$ and the fact that $E_{\A\I}\Pi^{\cpwt}_{\A\I}$ accepts only databases containing an entrance-to-exit path. The statement follows from squaring.
\end{proof}

\subsection{Proof of intertwiner theorem}\label{sec:intertwining_proof}

To prove \cref{thm:intertwining_oracles}, we will use an intermediate oracle which writes out its computation between each step of evaluating each permutation. Let $\Z$ be the intermediate computation register. Assume that the queried input is $\ket{l,\alpha}$. We define the oracle $\wtOZ$ to act as follows.
\begin{enumerate}
    \item Write into $\Z$ the inverse label $v = \labelPerm^{-1}(l)$.
    \item Let $e = \edgeVertex{\alpha}{v}$. If $e\neq\bot$ and $b=\edgeIndic(e)\in\middleEdge$, write $f$ into $\Z$, where
    \begin{align*}
        f = \begin{cases}
            \sigma_b(e) &\text{if $v=\tail(e)$,}\\
            \sigma^{-1}_b(e) &\text{if $v=\head(e)$.}
        \end{cases}
    \end{align*}
    Otherwise, set $f=\bot$. We encode $\bot$ by $0$ and all valid edge indices by nonzero strings.
    \item Let $w = \ncF_{\colorPerm}(v,\alpha)$ be the corresponding neighbor of $v$ which is a function of $(v,\alpha, f)$.
    \item XOR $\labelPerm(w)$ with the output $\Y$.
    \item Uncompute $f$ and then uncompute $v$.
\end{enumerate}

We use the following conventions, also in the compressed version below. If $l\neq0$ and $v=0$, skip the middle computation and only uncompute $v$. If an edge-permutation query returns $f=\bot$, set $w=\bot$ and skip the forward-label step. When $w=v$, XOR the original label $l$ directly, otherwise when $w=0$, do nothing. These conventions preserve the action of \wtO, since genuine permutation queries are always defined and $\labelPerm(v)=l$.

Therefore, the computation can be viewed as invoking the following permutations as oracles,
\begin{align*}
    \labelPerm^{-1} \cdot \sigma_b \cdot \labelPerm \cdot \sigma_b \cdot \labelPerm^{-1}.
\end{align*}
\begin{proposition}
    Running the oracle $\wtOZ$ with the intermediate register $\ket{0,0}_{\Z}$ returns it to $\ket{0,0}_{\Z}$ and behaves like \wtO.
\end{proposition}
\begin{proof}
    Conditioned on some query register $\ket{l,\alpha,a}_{\X\Y}$, \wtOZ acts as follows,
    \begin{align*}
        \ket{l,\alpha,a}_{\X\Y}\ket{0,0}_{\Z} &\xrightarrow{\labelPerm^{-1}} \ket{l,\alpha,a}_{\X\Y}\ket{v,0}_{\Z}\\
        &\xrightarrow{\sigma_b}  \ket{l,\alpha,a}_{\X\Y}\ket{v,f}_{\Z}\\
        &\xrightarrow{\labelPerm}  \ket{l,\alpha,a \oplus \labelPerm(w)}_{\X\Y}\ket{v,f}_{\Z}\\
        &\xrightarrow{\sigma_b}  \ket{l,\alpha,a \oplus \labelPerm(w)}_{\X\Y}\ket{v,0}_{\Z}\\
        &\xrightarrow{\labelPerm^{-1}}  \ket{l,\alpha,a \oplus \labelPerm(w)}_{\X\Y}\ket{0,0}_{\Z} = (\wtO_{\sigma} \ket{l,\alpha,a}_{\X\Y} )\ket{0,0}_\Z.
    \end{align*}
\end{proof}

We let \cpwtZ be the compressed version of \wtOZ, meaning that,
\begin{align*}
    \cpwtZ &\coloneq \cpLabInv \cdot \cpPerm \cdot \cpLabFor \cdot \cpPerm \cdot \cpLabInv.
\end{align*}

Note that as we are using compression operators, it is possible that the queries are not perfectly uncomputed. However, this is not a large issue.

\begin{lemma}\label{lem:separate_oracles_to_no_intermediate_comp}
    Let $\ket{\psi}$ be an arbitrary state such that $\abs{L} \leq t$ and $\abs{C}\leq t$. If $t < \sqrt{N}/4$, then,
    \begin{align*}
        \norm{\cpwtZ (\ket{\psi}{\ket{0}_\Z }) - (\cpwt\ket{\psi})\ket{0}_{\Z}} \leq 4\sqrt{\frac{2}{\sqrt{N}/4-t}}\norm{\ket{\psi}}.
    \end{align*}
\end{lemma}
\begin{proof}
    As each operator fixes the query register $\X$, we may fix some query $l,\alpha$. Suppose $l\neq 0$, as otherwise the result trivially follows. We may expand \cpLabInv in \cpwtZ as,
    \begin{align*}
        \cpwtZ = \pCLabelInv_{,l} \xorOP^{\mathrm{inv}}_\mathrm{lab} \pCLabelInv_{,l} \cdot (\cpPerm \cdot \cpLabFor \cdot \cpPerm )\pCLabelInv_{,l} \xorOP^{\mathrm{inv}}_\mathrm{lab} \pCLabelInv_{,l}.
    \end{align*}
    Let $\cpwtZ^{(1)}$ be the following,
    \begin{align}
        \cpwtZ^{(1)} \coloneq \pCLabelInv_{,l} \xorOP^{\mathrm{inv}}_\mathrm{lab} \cdot (\cpPerm \cdot \cpLabFor \cdot \cpPerm ) \xorOP^{\mathrm{inv}}_\mathrm{lab} \pCLabelInv_{,l},\label{eq:drop_inner_pc_operator_label}
    \end{align}
    which is just removing the intermediate $\pCLabelInv_{,l}$ calls after XORing. Let $\ket{\phi}=\xorOP^{\mathrm{inv}}_{\mathrm{lab}}\pCLabelInv_{,l}\ket{\psi}\ket{0}_{\Z}$ be the state immediately after the first XOR. In the case that $v=0$, the middle computation is the identity, so the compression operators perfectly cancel out. Otherwise, the register $\Z$ records that $L(v)=l$ and this does not change until the last XOR operator. Therefore \cref{prop:recompression_unlikely} applies to both before and after the computation between the XORs. As the database increases by at most $2$, by the triangle inequality we have,
    \begin{align*}
        &\norm{(\cpwtZ - \cpwtZ^{(1)})\ket{\psi}\ket{0}_{\Z}} \leq \norm{(\pCLabelInv_{,l} - \Id) \Pi_{v\neq 0} \ket{\phi}}\\
        &\quad + \norm{(\pCLabelInv_{,l} - \Id) (\cpPerm \cdot \cpLabFor \cdot \cpPerm) \Pi_{v\neq 0} \ket{\phi}} \\
        &\quad\leq 2\sqrt{\frac{2}{M-t-1}} \norm{\ket{\psi}}. &\mbox{(By \cref{prop:recompression_unlikely})}
    \end{align*}
    Similarly, write $\pC_{v,\alpha}$ for the edge-compression operator, or the identity when the edge query is skipped. Then,
    \begin{align*}
        \cpPerm \cdot \cpLabFor \cdot \cpPerm = \pC_{v,\alpha} \cdot \xorOP_{\mathrm{edg}} \cdot \pC_{v,\alpha} \cdot \cpLabFor \cdot \pC_{v,\alpha} \cdot \xorOP_{\mathrm{edg}} \cdot \pC_{v,\alpha}
    \end{align*}
    And similar to $\cpwtZ^{(1)}$, we may compare this with the version where we remove the inner $\pC_{v,\alpha}$ compression primitives,
    \begin{align}
        \pC_{v,\alpha} \cdot \xorOP_{\mathrm{edg}} \cdot\cpLabFor \cdot \xorOP_{\mathrm{edg}}  \cdot\pC_{v,\alpha} \nonumber%\label{eq:drop_inner_pc_operator_permutation}
    \end{align}
    If the edge database is not queried or $f=\bot$, the two expressions are equal. Otherwise when $f\neq \bot$, the register $\Z$ records the queried edge assignment which is again preserved until XORed again. Therefore, the same triangle inequality and \cref{prop:recompression_unlikely}, together with $\abs{E_b}\geq\sqrt{N}/4$ by \cref{obs:min_num_nodes_per_perm}, give
    \begin{align*}
        \norm{ (\cpPerm \cdot \cpLabFor \cdot \cpPerm - \pC_{v,\alpha} \xorOP_{\mathrm{edg}} \cpLabFor \xorOP_{\mathrm{edg}} \pC_{v,\alpha}) \ket{\phi}} \leq 2\sqrt{\frac{2}{\sqrt{N}/4-t}}\norm{\ket{\psi}}.
    \end{align*}
    Notice that after removing these operators, each lookup in $\Z$ is equivalent to looking it up in the actual database. Furthermore, these values in $\Z$ are still always uncomputed. Therefore, suppressing the intermediate lookups, the remaining operators are,
    \begin{align}
        \pCLabelInv_{,l} \cdot \pC_{v,\alpha} \cdot  \cpLabFor \pC_{v,\alpha} \pCLabelInv_{,l}\label{eq:removing_operators_in_cpwtz}
    \end{align}
    The operations in \cref{eq:removing_operators_in_cpwtz}, now controlled by the corresponding database entries, are exactly \cpwt. Since $M-1\geq\sqrt{N}/4$, summing the two errors and using orthogonality of the query register proves the statement.
\end{proof}

Finally, we may prove the main statement.

\begin{proof}[Proof of \cref{thm:intertwining_oracles}]
    Let $\cmfWT$ be the oracle which behaves as $\wtOZ$, except each oracle call is replaced with the oracle $\cmfO$ from \cref{lem:general_intertwiner_lemma}, and let $\calI_{\regP}$ be the intertwiner over both the label and edge registers and $\calV$ be the product of the isometries $\calV_S$ from \cref{lem:appendix_local_permutation_comparison} over these registers. Note that the $\Z$ register is implicit.

    Let $\ket*{\psi^{(\cdot)}}$ be the state of running an algorithm with $t$ calls to either $\cmfWT$ or $\cpwt$. By \cref{lem:cpwt_bounded_growth}, the databases in support of the state have $\abs{L}\leq 2t$ and $\abs{C}\leq t$. Also, we may assume that $t\leq \tfrac{\sqrt{N}}{16}$ as otherwise the bound is trivial. By applying \cref{lem:appendix_local_permutation_comparison,lem:separate_oracles_to_no_intermediate_comp},
    \begin{align*}
        \norm{(\cmfWT\calV-\calV\cpwt)\ket{\psi_{\alg,t}^{(\cpwt)}}} &= O\left(\frac{t+1}{N^{1/8}}\right).
    \end{align*}
    Telescoping over all queries gives,
    \begin{align*}
        \norm{\ket{\psi_{\alg,t}^{(\cmfWT)}}-\calV\ket{\psi_{\alg,t}^{(\cpwt)}}}=O(t^2/N^{1/8}).
    \end{align*}
    The construction in \cref{sec:general_intertwiner_lemma} satisfies $\calI_{\regP}\calV=\ket{\bot}_{\regP}\otimes\Id$. Hence,
    \begin{align*}
        &\norm{(\calI_{\regP}\cmfWT-\cpwt\calI_{\regP})\ket{\psi_{\alg,t}^{(\cmfWT)}}}\\
        &\quad\leq2\norm{\ket{\psi_{\alg,t}^{(\cmfWT)}}-\calV\ket{\psi_{\alg,t}^{(\cpwt)}}}+\norm{(\cmfWT\calV-\calV\cpwt)\ket{\psi_{\alg,t}^{(\cpwt)}}}\\
        &\quad=O\left(\frac{t^2}{N^{1/8}}\right).\qedhere
    \end{align*}
\end{proof}

\section{Hardness of pathfinding}

The goal of this section is proving the following theorem.

\begin{theorem}\label{thm:pathfinding_hard}
    Assume that $n\geq16$ is even. For any quantum algorithm $\alg$ which makes $1\leq q\leq \tfrac{\sqrt{N}}{100}$ queries to the glued trees oracle,
    \begin{align*}
        \Pr[\text{$\alg$ outputs a path from entrance to exit}] = O\left(\frac{q^6}{N^{1/4}}\right).
    \end{align*}
\end{theorem}

\subsection{Fresh oracle}\label{sec:fresh_oracle}

We define a modified version of \cpwt, denoted by \mcpwt, which we call \emph{fresh} as every single assignment in the edge database is completely untouched by the current database.

Let us formalize how the database $I = (L,C)$ describes a known graph $\rG(I) = (\rV(I), \rE(I))$ in the middle region,
\begin{align}
    \rE(I) &\coloneq \left\{ \{ \tail(e), \head(f) \}: e \in \dom(C), C(e) = f \right\},\\
    \rV(I) &\coloneq \bigcup_{e \in \rE(I)} e.
\end{align}
We call the edges in $\rE(I)$ \emph{logical} edges as they represent the logic of how $\ncF_C$ connects nodes.
Let $\vPerm = \cup_{{b\in \middleEdge}}\cup_{ e\in E_b} \{ \tail(e), \head(e)\}$ be all the vertices which are connected to edges in $\middleEdge$. Furthermore, we also define the set of all endpoints used,
\begin{align}
    \supp(I) &\coloneq \left(\dom(L) \cap \vPerm  \right) \cup \left( \bigcup_{e\in \dom(C) \cup \im(C)} \{\tail(e), \head(e)\} \right).\label{eq:support_database_def}
\end{align}
Intuitively, $\supp(I)$ denotes the set of points we want to avoid in order to keep the assignments \emph{fresh}. For some fixed $b\in \middleEdge$ and $e\in E_b$, define,
\begin{align}
    \calA_{e,b,I}^{\rightarrow} &\coloneq \{f\in E_b \setminus \left( \dom(C_b) \cup \im(C_b) \right): \head(f) \notin \supp(I) \},\label{eq:available_assignments_fresh_right}\\
    \ket{\tilde{+}_{e,b,I}^{\rightarrow}} &\coloneq \frac{1}{\sqrt{\abs{\calA_{e,b,I}^{\rightarrow}}}}\sum_{f\in \calA_{e,b,I}^{\rightarrow}} \ket{C[e\to f]},\label{eq:plus_state_fresh_right}\\
    \calA_{e,b,I}^{\leftarrow} &\coloneq \{f\in E_b \setminus \left( \dom(C_b) \cup \im(C_b) \right): \tail(f) \notin \supp(I) \},\label{eq:available_assignments_fresh_left}\\
    \ket{\tilde{+}_{e,b,I}^{\leftarrow}} &\coloneq \frac{1}{\sqrt{\abs{\calA_{e,b,I}^{\leftarrow}}}} \sum_{f\in \calA_{e,b,I}^{\leftarrow}} \ket{C[f\to e]}\label{eq:plus_state_fresh_left}.
\end{align}
We define \mpCPerm as \pCPerm, except we replace the plus states from \cref{eq:plus_state_upcolor,eq:plus_state_downcolor} with those from \cref{eq:plus_state_fresh_right,eq:plus_state_fresh_left}. Uf the corresponding set is empty, \mpCPerm acts as the identity. We define the modified fresh oracle.

\begin{definition}The \emph{fresh compressed welded trees} oracle is defined as,
    \begin{align}
        \mcpwt &\coloneq \pCLabelInv \cdot \mpCPerm \cdot \pCLabel \cdot \xorOP \cdot \pCLabel\cdot \mpCPerm  \cdot \pCLabelInv.\label{eq:mcpwt_def}
    \end{align}
\end{definition}

First, we may observe that using \mcpwt avoids creating any cycles.

\begin{lemma}\label{lem:growth_argument}
    Fix some query basis $\ket{l,\alpha}_{\X}$ and some database $I$ which is not acted as identity by $\mpCPerm$, and assume that the queried edge assignment is absent. For any database $J$ which appears as a child in the plus state in \cref{eq:plus_state_fresh_left,eq:plus_state_fresh_right}, there exists some \emph{logical} edge $g$ such that,
    \begin{align}
        \rE(J) &= \rE(I) \dot{\cup} \{g\},\label{eq:increase_edge_db}\\
        \abs{g \cap \supp(I)} &= 1\label{eq:support_intersect}\\
        \abs{g \cap \rV(I)} &\leq 1.\label{eq:bound_number_known_edges}
    \end{align}
\end{lemma}
\begin{proof}
    Let $v=L^{-1}(l)$ and $e=\edgeVertex{\alpha}{v}$. Since $I$ is not acted as identity by $\mpCPerm$, we have that $e\neq \bot$, $b=\edgeIndic(e) \in \middleEdge$ and $v\neq 0$. Therefore, $v\in \dom(L)\cap \vPerm$. Without loss of generality, assume that $v=\tail(e)$ and let $J = (L,C[e\to f])$ for some $f\in E_b$. Therefore the new logical edge is $g=\{\tail(e), \head(f)\}$.

    As $v\in \dom(L)\cap \vPerm$, we have that $\tail(e) \in \supp(I)$, while $\head(f)\notin\supp(I)$ by \cref{eq:available_assignments_fresh_right}, implying \cref{eq:support_intersect}. Moreover, by definition $\rV(I) \subseteq \supp(I)$, leading to \cref{eq:bound_number_known_edges}. Lastly, the disjoint union in \cref{eq:increase_edge_db} follows from the fact that every edge in $\rE(I)$ lies in $\supp(I)$, meaning that $g$ could not be present. The case when $v=\head(e)$ follows analogously.
\end{proof}

\begin{observation}\label{obs:bounded_growth_mcpwt}
    After $q$ applications of $\mcpwt$, every database in the support of the state $\ket{\psi_q}$ satisfies $\abs{L} \leq 2q$ and $\abs{C} \leq q$. Hence, $\abs{\supp(I)} \leq \abs{L} + 4\abs{C} \leq 6q$.
\end{observation}
\begin{proof}
    The same proof as of \cref{lem:cpwt_bounded_growth} applies here, bounding $\abs{L}$ and $\abs{C}$. The bound on $\abs{\supp(I)}$ follows by \cref{eq:support_database_def}.
\end{proof}

Furthermore, we may observe that using $\mcpwt$ does not drastically differ from $\cpwt$.

\begin{lemma}\label{lem:cpwt_mcpgt}
    For any quantum algorithm $\alg$ and projector $\Pi$ on $\I$,
    \begin{align*}
        \norm{\Pi \ket{\psi_{\alg, q}^{(\cpwt)}}}^2 \leq 2 \norm{\Pi \ket{\psi_{\alg, q}^{(\mcpwt)}}}^2 + O\left( \frac{q^3}{\sqrt{N}} \right).
    \end{align*}
\end{lemma}
\begin{proof}
    Let us first bound the norm of the difference between \pCPerm and \mpCPerm after $q$ queries. Notice that the color databases in the direct sum in \cref{eq:pC_right_direction,eq:pC_left_direction} are orthogonal. By \cref{lem:operator_norm_max}, it suffices to bound the difference between any arbitrary $\pCfromPlus$ applied to some $C$. By \cref{obs:superposition_similarity}, this depends on the relationship between the available sets. 
    
    Suppose we fix some $e,b$ and $C$. Without loss of generality, let us consider the $\rightarrow$ case as the case with $\leftarrow$ follows analogously. Notice that for any $\calA_{e,b,I}^{\rightarrow} \subseteq E_b \setminus \im(C_b)$. By \cref{obs:min_num_nodes_per_perm}, we have that,
    \begin{align*}
        \abs{E_b \setminus \im(C_b)} &\geq \frac{\sqrt{N}}{4} - q\\
        \abs{\calA_{e,b,I}^{\rightarrow}} &\geq \frac{\sqrt{N}}{4} - \abs{C} - 2\abs{\supp(I)} \geq \frac{\sqrt{N}}{4} - 13q
    \end{align*}
    Therefore, under the assumption $q \leq \tfrac{\sqrt{N}}{100}$,
    \begin{align*}
        \norm{\pCPerm - \mpCPerm}_{\leq q} = O\left(\frac{\sqrt{q}}{N^{1/4}}\right)
    \end{align*}
    By repeatedly applying the triangle inequality,
    \begin{align*}
        \norm{\ket{\psi_{\alg, q}^{(\cpwt)}} - \ket{\psi_{\alg, q}^{(\mcpwt)}}} &\leq \norm{ \sum_{i \in [2q]}(\pCPerm - \mpCPerm)}\\
        &\leq \sum_{i \in [2q]} \norm{\pCPerm - \mpCPerm} = O\left( \frac{q^{3/2}}{N^{1/4}} \right).
    \end{align*}
    The statement follows by squaring and using that fact that $(a+b)^2 \leq 2a^2 + 2b^2$.
\end{proof}

\subsubsection{Valley lemma}

In the following section, we formalize the idea that due to \cref{lem:growth_argument}, any path on $\rG(I)$ must grow one step at a time in some direction. First, we expand the oracle $\mcpwt$ into the $7$ primitives in \cref{eq:mcpwt_def}. Hence, any $q$-query algorithm \alg may be written as $T=8q+1$ steps $\calV = (V_1, V_2,\dots, V_T)$ where each $V_i$ is either one of the algorithm unitaries $A_0,\dots,A_q$, a compression operator $\pCLabelInv, \pCLabel, \mpCPerm$ or the XOR operator $\xorOP$. Using this decomposition, we define the basis history of the state.

\begin{definition}[Basis History] A \emph{basis history} is a sequence
    \begin{align*}
        h = ((a_0,I_0), (a_1, I_1),\dots,(a_T, I_T)),
    \end{align*}
    of basis states in $\A\I$ with $I_0=\bot$ such that
    \begin{align*}
        \Pi_{t\in [T]} \bra{a_t, I_t} V_t \ket{a_{t-1}, I_{t-1}} \neq 0.
    \end{align*}
\end{definition}

Let us first show that throughout any basis history, each database over $\middleEdge$ represents a tree.

\begin{lemma}\label{lem:no_cycles_in_graph}
    Fix some $t\in [T]$ and let $A_r = \rE(I_t) \setminus \rE(I_{t-1})$. If $A_r \neq \emptyset$, then $A_r = \{e\}$ such that $\abs{e \cap \rV(I_{t-1})} \leq 1$ and $\rG(I_t)$ is a forest.
\end{lemma}
\begin{proof}
    Let us apply induction on $t$. By definition, $\rG(I_0)$ does not contain any edges and hence is a forest. Assume $\rG(I_{t-1})$ is a forest. If $V_t$ is one of $\pCLabelInv, \pCLabel, \xorOP$ or $A_i$, then $\rE(I_{t-1}) = \rE(I_{t})$, which implies the statement.

    Assume that $V_t$ is $\mpCPerm$. If $\abs{I_t} < \abs{I_{t-1}}$, we must erase an edge and hence the statement also holds. Therefore, we either replaced one edge with another or added a new edge. If we added a new edge, then by \cref{lem:growth_argument}, we added one unique edge $e$ such that $\abs*{e \cap \rV(I_{t-1})} \leq 1$, meaning that $e$ cannot connect two vertices in $\rV(I_{t-1})$ and hence $\rG(I_t)$ is a forest. If we replaced an edge, by \cref{lem:expansion_pc_operator} we have essentially removed the original edge and then added a new edge, meaning that \cref{lem:growth_argument} applies to the database after removal. The new structural vertex it points to is different from the old one and must be fresh, meaning it is not in $\rV(I_{t-1})$.
\end{proof}

For each edge which appears in the final database $\rE(I_T)$, we will be interested in the moment it was added to the database.

\begin{definition}[Last-addition time] For a history $h$ and an edge $e\in \rE(I_T)$, the last-addition time is defined to be,
    \begin{align*}
        \lastAddTime_h(e) &\coloneq \max \{t \in [T]: e\notin \rE(I_{t-1}), e\in \rE(I_{t})\}
    \end{align*}
\end{definition}

Notice that by \cref{obs:bounded_growth_mcpwt}, for any two edges $e\neq e^\prime \in \rE(I_T)$, $\lastAddTime_h(e) \neq \lastAddTime_h(e^\prime)$. Lastly, we will define \emph{simple paths}, which are paths which never retract.

\begin{definition} For $\{v_i\}_{i\in[0,m]} \in \vPerm^{m+1}$ and $\{e_i\}_{i\in [m]} \in (\vPerm \times \vPerm)^{m}$, a \emph{simple path} in $\rG(I)$ is a sequence of vertices and edges
    \begin{align}
        \calP = \{v_0, e_1, v_1, \dots e_m, v_m\}
    \end{align}
    such that each edge is in $\rE(I)$ and no two edges repeat. The length of $\calP$ is denoted by $\abs{\calP} = m$. For a given graph $G$ and $u,v\in V$, we use $u\xleftrightarrow{G} v$ to denote that there exists a simple path on $G$ between $u$ and $v$.
\end{definition}

By \cref{lem:no_cycles_in_graph}, no simple path can form a cycle.

\begin{lemma}[Valley Shape]
    For any simple path $\calP$ on $\rG(I_T)$ of a basis history $h$, there exists a unique $j \in [\abs{\calP}]$ such that,
    \begin{align*}
        \lastAddTime_h(e_1) > \dots > \lastAddTime_h(e_{j-1}) > \lastAddTime_h(e_{j}) < \lastAddTime_h(e_{j+1}) < \dots < \lastAddTime_h(e_{m}).
    \end{align*}
\end{lemma}
\begin{proof}
    Let us use $b_i = \lastAddTime_h(e_{i})$ for all $i\in [\abs{\calP}]$. For the sake of contradiction, assume that there exists some $i$ such that $b_{i-1} < b_i > b_{i+1}$. That means that $v_{i-1}, v_{i} \in \rV(I_{b_i-1})$, meaning that $\abs{e_i \cap \rV(I_{b_i-1})} = 2$, contradicting \cref{lem:no_cycles_in_graph}. Therefore, the sequence of $b_j$ cannot contain a local maximum. As each $b_j$ is unique, there must exist a unique $i$ such that for both $j< i$ and $i<j$, the sequence of $b_j$ is increasing as otherwise we would produce a local maximum.
\end{proof}

\subsection{Reduction to a fixed weld}\label{sec:reduction_to_fixed_weld}

First, let us define a weaker version of $\Pi_{\calR}$ which solely conditions on the presence of edges which cross from column $n/2$ to $3n/2$.
\begin{align}
    \eventCrossEdge(I) &\coloneq \Id\left[ \exists u,v\in \rV(I) \text{ s.t. } \col(u) = n/2, \col(v)=3n/2 \text{ and } u \xleftrightarrow{\rG(I)} v \right]\\
    \prCross &\coloneq \sum_{I: \eventCrossEdge(I)=1} \ketbra{I}_{\I}
\end{align}
In our case, $u \xleftrightarrow{\rG(I)} v $ implies that the edge database $C$ contains a simple path between $u$ and $v$, even if the labels of the intermediate nodes are not known. As any path which crosses from the entrance to the exit must contain a subpath between columns $n/2$ and $3n/2$ and both projectors are diagonal,
\begin{align}
    \Pi_{\calR} \preceq \prCross.\label{eq:crossing_subsumes_finding_path}
\end{align}
The remainder of the proof will be concerned with bounding the probability that a database $I$ in the support of a state satisfies $\eventCrossEdge(I)$. Next, we partition the set of databases satisfying $\eventCrossEdge(I)$ based on which welded edge they cross.

A \emph{weld record} $\rho$ is the tuple $\rho = (\alpha, e,f)$ where $\alpha\in \colorSet$ and $e,f\in E_{\mathrm{wld}, \alpha}$. We say that $\rho \in I$ when $C(e)=f$ and the associated edge with $\rho$ is $E(\rho) = \{\tail(e), \head(f)\}$.

Let us assume that there exists an order for all pairs $(u,v) \in \col(n/2) \times \col(3n/2)$ and all weld records $\rho$ (for example using the structural vertices). By \cref{lem:no_cycles_in_graph}, we have that at all times $\rG(I)$ is forest, meaning that any simple path between two vertices in $\rG(I)$ is unique. For some $I$ such that $\eventCrossEdge(I)=1$, let $\calP_I$ denote the unique simple path between the first pair $(u,v)\in \col(n/2) \times \col(3n/2)$ which satisfies the event $\eventCrossEdge$. The \emph{canonical weld record} of $I$ is,
\begin{align*}
    \rho^*(I) &\coloneq \min \{ \rho\in I: E(\rho) \in \calP_I \},
\end{align*}
where the minimization is using the aforementioned ordering. As this is unique, we may partition the $I$ in $\prCross$ as follows,
\begin{align}
    \prCanonicalCross &\coloneq \sum_{\substack{I: \eventCrossEdge(I) = 1 \\ \rho^*(I) = \rho}} \ketbra{I},\\
    \prCross &= \sum_{\rho} \prCanonicalCross.\label{eq:decomposition_pr_cross_canonical}
\end{align}

\paragraph{Last surviving insertion.} Next, we will partition the state further based on the direction. For simplicity of notation, we will use the following notation for intermediate states,
\begin{align*}
    \ket{\psi_0} &\coloneq \ket{0}_{\A} \ket{\bot}_{\I},\\
    \ket{\psi_t} &\coloneq V_t\dots V_1 \ket{\psi_0},
\end{align*}
where we are using the decomposition $\calV$ of an algorithm $\alg$. For a weld record $\rho$, we define the following projectors,
\begin{align*}
    &\prWeld \coloneq \sum_{I: \rho\in I} \ketbra{I} &\prWeldOrtho \coloneq \Id - \prWeld.
\end{align*}
Notice that $ \prCanonicalCross \preceq \prWeld$ and $\prWeld\ket{\psi_0} = 0$. For an arbitrary $t\in [T]$, we define the \emph{future retained weld} projector as,
\begin{align*}
    \prFuture &\coloneq \prWeld V_T \prWeld \dots \prWeld V_{t+1} \prWeld.
\end{align*}
We will use $\prFuture$ to decompose the state as follows,
\begin{align}
    \prCanonicalCross \ket{\psi_T} &= \prCanonicalCross V_T \dots V_1 \ket{\psi_0} \nonumber\\
    &= \prCanonicalCross V_T \prWeldOrtho \dots V_1 \ket{\psi_0} + \prCanonicalCross V_T \prWeld \dots V_1 \ket{\psi_0} \nonumber\\
    &= \left(\sum_{t\in [T]} \prCanonicalCross \prFuture V_t \prWeldOrtho V_{t-1} \dots V_1 \ket{\psi_0}\right) + \prCanonicalCross V_T \prWeld \dots V_1 \prWeld \ket{\psi_0} \nonumber\\
    &=\sum_{t\in [T]} \prCanonicalCross \prFuture V_t \prWeldOrtho \ket{\psi_{t-1}}.\label{eq:decomp_canonical_weld_proj}
\end{align}
Finally, we will decompose the projectors based on the \emph{direction} of insertion. Fix some weld $\rho = (\alpha, e,f)$. We will use $\sigma \in \{\rightarrow, \leftarrow\}$ to denote the insertion, which denotes which one of the available sets in \cref{eq:available_assignments_fresh_right,eq:available_assignments_fresh_left} was used.

\begin{definition}
    Fix some $t\in [T]$. We define the edge-insertion direction operator \insertDir as follows. If $V_t \in \{A_i, \pCLabel, \pCLabelInv, \xorOP\}$, $\insertDir = 0$. Otherwise, let $\ket{x} = \ket{a,L,C}_{\A\I}$ and $\ket{y} = \ket{a, L, C^\prime}_{\A\I}$ be two basis states. We define $\insertDir$ to satisfy,
    \begin{align*}
        \bra{y} \insertDir \ket{x} = \bra{y} \mpCPerm \ket{x}
    \end{align*}
    exactly when,
    \begin{enumerate}
        \item $\rho \notin C$ and $\rho \in C^\prime$,
        \item Letting $B = C \cap C^\prime$, we have $C^\prime = B[e\to f]$
        \item If $\sigma = \rightarrow$, the applied action is $\pCfromPlus\left(\ket{\tilde{+}^{\rightarrow}_{e,(\mathrm{wld}, \alpha), B}}, \ket{B}\right)$, while $\sigma = \leftarrow$ implies that the action was $\pCfromPlus\left(\ket{\tilde{+}^{\leftarrow}_{f,(\mathrm{wld}, \alpha), B}}, \ket{B}\right)$.
    \end{enumerate}
    Furthermore, we define the endpoint of $\rho$ in $\sigma$ as,
    \begin{align*}
        r_{\rho,\sigma} &\coloneq \begin{cases}
            \head(f) &\text{if $\sigma = \rightarrow$,}\\
            \tail(e) &\text{if $\sigma = \leftarrow$.}
        \end{cases}
    \end{align*}
\end{definition}
We use $r$ as it is the root of the connected component. The other side of this edge is the seed, denoted by $s$.
We immediately get the following.

\begin{observation}\label{obs:split_futures}
    For any $t\in [T]$ and weld $\rho$,
    \begin{align*}
        \prWeld V_t \prWeldOrtho = A_{\rho, t}^{\rightarrow} + A_{\rho, t}^{\leftarrow}
    \end{align*}
\end{observation}
\begin{proof}
    We may assume that $V_t = \mpCPerm$ as otherwise $C$ does not change, meaning that both sides of the equation are $0$. Let $\ket{x} = \ket{a,L,C}_{\A\I}$ and $\ket{y}=\ket{a,L,C^\prime}_{\A\I}$. Then,
    \begin{align*}
        \bra{a,L,C^\prime} \prWeld \mpCPerm \prWeldOrtho \ket{a,L,C}_{\A\I} = \Id[\rho\notin C, \rho\in C^\prime] \bra{a,L,C^\prime}\mpCPerm \ket{a,L,C}_{\A\I}.
    \end{align*}
    Let $B= C\cap C^\prime$ be the base database and $h$ the queried edge. By \cref{lem:expansion_pc_operator}, we have that $C^\prime \setminus C = \{\rho\}$ and $C^\prime = B[e\to f]$. Let $(l,\beta)$ be the query in $a$ and $v=L^{-1}(l)$. By \cref{eq:plus_state_fresh_right,eq:plus_state_fresh_left},
    \begin{align*}
        (\beta, v) = \begin{cases}
            (\alpha, \tail(e)) &\text{if $\sigma = \rightarrow$,}\\
            (\alpha, \head(f)) &\text{if $\sigma=\leftarrow$.}
        \end{cases}
    \end{align*}
    As the direction is unique, every nonzero entry appears in one of $A_{\rho,t}^{\sigma}$. Therefore,
    \begin{align*}
        \sum_{\sigma \in \{\rightarrow, \leftarrow\}} \bra{a,L,C^\prime} A_{\rho,t}^{\sigma} \ket{a,L,C}_{\A\I} &=\Id[\rho\notin C, \rho\in C^\prime] \bra{a,L,C^\prime}\mpCPerm \ket{a,L,C}_{\A\I}\\
        &=\bra{a,L,C^\prime} \prWeld \mpCPerm \prWeldOrtho \ket{a,L,C}_{\A\I}.\qedhere
    \end{align*}
\end{proof}

Combining \cref{eq:decomp_canonical_weld_proj,obs:split_futures},
\begin{align}
    \prCanonicalCross \ket{\psi_T} &= \sum_{t\in [T]} \sum_{\sigma\in \{\leftarrow, \rightarrow\}}  \prCanonicalCross \prFuture \insertDir \ket{\psi_{t-1}}\label{eq:decomp_canonical_weld_directions}
\end{align}
Note that the summands are not necessarily orthogonal. Furthermore, the following lemma will be useful.
\begin{lemma}\label{lem:decomp_insertions_per_state}
    For any $t\in [T]$ and vector $\ket{\psi}$,
    \begin{align}
        \sum_{\rho} \sum_{\sigma \in \{\rightarrow, \leftarrow\}} \norm{\insertDir \ket{\psi}}^2 \leq 4\norm{\ket{\psi}}^2
    \end{align}
\end{lemma}
\begin{proof}
    We may assume that $V_t = \mpCPerm$ as otherwise $A_{\rho,t}^{\sigma} = 0$, making the statement trivial. By \cref{eq:pC_right_direction,eq:pC_left_direction}, we may decompose $\mpCPerm$ into orthogonal $\pCfromPlus$ blocks. By \cref{lem:operator_norm_max}, it suffices to prove the bound for $\ket{\psi}$ which is solely supported on one such block whose base database is $B$.

    As we cannot add the same record twice,
    \begin{align*}
        \sum_{\rho,\sigma} \norm{A_{\rho,t}^{\sigma} \ket{\psi}}^2 &= \sum_{\rho\notin B} \norm{\prWeld \mpCPerm \prWeldOrtho \ket{\psi}}^2 &\mbox{(By \cref{obs:split_futures})}\\
        &= \sum_{\rho \notin B}  \norm{\prWeld \mpCPerm \ket{\psi} - \prWeld \mpCPerm \prWeld \ket{\psi}}^2 \\
        &\leq 2\sum_{\rho \notin B}\left(\norm{\prWeld \mpCPerm \ket{\psi}}^2 + \norm{\prWeld \mpCPerm \prWeld \ket{\psi}}^2 \right)\\
        &\leq 2\left(\norm{\mpCPerm \ket{\psi}}^2 + \sum_{\rho \notin B} \norm{\prWeld \ket{\psi}}^2 \right)\\
        &\leq 4\norm{\ket{\psi}}^2,
    \end{align*}
    where the second term in the last inequality follows due to the fact that the projectors $\prWeld$ are orthogonal when restricted to some fixed block in $\pCfromPlus$.
\end{proof}

We will use $K_{\rho, \sigma}(I) = \mathrm{Comp}_{\rG(I) \setminus E(\rho)}(r_{\rho,\sigma})$ to denote the connected component of $\rG(I)$ which is created after the inserted edge. Next, let us define the projector that a path reaches the boundary after crossing $\rho$ in direction $\sigma$,
\begin{align*}
    \eventBoundary_{\rho,\sigma}(I) &\coloneq \Id\left[ \exists u\in V(K_{\rho, \sigma}(I)): \col(u) \in \{n/2, 3n/2\}\right],
\end{align*}
and let $\prBoundary_{\rho,\sigma}$ be the corresponding diagonal projector. Notice that for any $\rho, \sigma$,
\begin{align}
    \prCanonicalCross \preceq \prBoundary_{\rho,\sigma},\label{eq:prCanonical_subsumed_boundary}
\end{align}
as crossing through $\rho$ implies that we must have reached one of the exit columns after crossing in this direction. This leads to the main result of this section.

\begin{lemma}\label{lem:bounding_future_bounds_exit_probability}
    Suppose that there exists some $\Gamma \geq 0$ such that for all $\rho, t, \sigma$,
    \begin{align}
        \norm{\prBoundary_{\rho, \sigma} \prFuture \insertDir \ket{\psi_{t-1}}}^2 \leq \Gamma \norm{\insertDir\ket{\psi_{t-1}}}^2.\label{eq:future_bound_goal}
    \end{align}
    Then,
    \begin{align*}
        \norm{\Pi_{\calR}\ket{\psi_T}}^2 \leq 8T^2 \Gamma.
    \end{align*}
\end{lemma}
\begin{proof}
    Suppose we fix some $\rho$. Notice,
    \begin{align*}
        \norm{\prCanonicalCross \ket{\psi_T}}^2 &= \norm{\sum_{t,\sigma} \prCanonicalCross \prFuture \insertDir \ket{\psi_{t-1}}}^2 &\mbox{(By \cref{eq:decomp_canonical_weld_directions})}\\
        &\leq 2T\sum_{t,\sigma} \norm{\prCanonicalCross \prFuture \insertDir \ket{\psi_{t-1}}}^2 &\mbox{(By Cauchy-Schwarz)}\\
        &\leq 2T \sum_{t,\sigma} \norm{\prBoundary_{\rho, \sigma} \prFuture \insertDir \ket{\psi_{t-1}}}^2 &\mbox{(By \cref{eq:prCanonical_subsumed_boundary})}\\
        &\leq 2T\Gamma \sum_{t,\sigma} \norm{\insertDir\ket{\psi_{t-1}}}^2 &\mbox{(By \cref{eq:future_bound_goal})}
    \end{align*}
    Using the orthogonal partition of $\prCross$,
    \begin{align*}
        \norm{\prCross \ket{\psi_T}}^2 &= \sum_{\rho} \norm{\prCanonicalCross \ket{\psi_T}}^2 &\mbox{(By \cref{eq:decomposition_pr_cross_canonical})}\\
        &\leq 2T\Gamma \sum_{t\in [T]}\sum_{\rho, \sigma} \norm{\insertDir\ket{\psi_{t-1}}}^2 &\mbox{(By \cref{lem:decomp_insertions_per_state})}\\
        &\leq 8T\Gamma \sum_{t\in [T]} \norm{\ket{\psi_{t-1}}}^2 = 8T^2\Gamma.
    \end{align*}
    The final statement follows due to \cref{eq:crossing_subsumes_finding_path}.
\end{proof}

Hence the remainder of the proof is concerned with establishing the constant $\Gamma$ in \cref{eq:future_bound_goal}.

\subsection{Fixed seed escaping probability}\label{sec:escaping_probability_bound}

Unless stated otherwise, assume that we are working with some fixed weld crossing $\rho = (\alpha, e,f)$ and direction $\sigma \in \{\rightarrow, \leftarrow\}$. Remember that these define the root $r \coloneq r_{\rho, \sigma}$.

As a reminder, we assume that $n\geq16$. The desired escape bound is trivial unless
\begin{align}
    \frac{2^{14}(T+1)^4}{\sqrt{N}} < 1. \label{eq:assumption_number_queries_bound}
\end{align}
We therefore assume this inequality throughout the subsection. This implies $T\leq\sqrt{N}/100$ and $\delta<1/2$, where $\delta$ is defined in \cref{lem:goodness_probability}.

\subsubsection{Fixed color paths do not escape}

For some $u\in \middleNodes$, $e= \edgeVertex{\alpha}{u}$ and $b=\edgeIndic(e)$, we let $\probAssignEdge{\alpha}{u,v}$ be the probability that an edge is randomly assigned to go between $u$ and $v$,
\begin{align}
    \probAssignEdge{\alpha}{u,v} &\coloneq \begin{cases}
        \tfrac{1}{\abs{E_b}} &\text{if $u=\tail(e)$ and $v\in \{\head(g): g\in E_b\}$,}\\
        \tfrac{1}{\abs{E_b}} &\text{if $u=\head(e)$ and $v\in \{\tail(g): g\in E_b\}$,}\\
        0 &\text{otherwise.}
    \end{cases}\label{eq:prob_assign_edge_distr}
\end{align}

Notice that this is a valid probability distribution over all possible vertices as $\sum_{v\in [M] \cup \{\bot\}} \probAssignEdge{\alpha}{u,v} = 1$. Next, we discuss color sequences $\mathbf{\alpha} = (\alpha_1, \alpha_2,\dots \alpha_h)$. Intuitively, we will think of them as paths from the root $r$. We call a color sequence \emph{valid} if it does not immediately cross $\rho$ and never retraces an edge. We use $\emptySeq$ to denote the empty color sequence. Let $\colSeqSet_h$ denote the set of valid color sequences up to length $h$,
\begin{align*}
    \colSeqSet_h &\coloneq \{\emptySeq\}\cup\{\mathbf{\alpha} = (\alpha_1,\dots,\alpha_k): 1\leq k\leq h, \alpha_1 \neq \alpha\text{ and } \forall i\in [k-1], \alpha_{i} \neq \alpha_{i+1} \}.
\end{align*}
Let $H \subseteq \colSeqSet_h$ be some set of color sequences.
In general, we will use valid color sequences to describe the general structure of the tree whose origin is $r$. We will use $\colorVertMap: H \to V_n \cup \{\bot\}$ to be some arbitrary mapping from color sequences to structural nodes. This will be useful to define the probability distribution $\mu_H(\colorVertMap)$ over assignments,
\begin{align}
    \mu_{H}(\colorVertMap) &\coloneq \Id[\colorVertMap(\emptySeq) = r_{\rho,\sigma}] \prod_{\tau \beta \in H} \probAssignEdge{\beta}{\colorVertMap(\tau),\colorVertMap(\tau \beta)},\label{eq:probability_distribution_paths}
\end{align}
where for $\tau \in \colSeqSet$ and $\beta \in \colorSet$, $\tau\beta$ denotes the concatenation of $\beta$ with $\tau$. We call a color sequence $\tau \in H$ a leaf if there does not exist another longer color sequence in $H$ which contains it. For a color sequence of length $h$ and $i\in [h]$, we let $\tau_{\leq i}$ be the first $i$ colors of the sequence. Notice that $\sum_{\colorVertMap} \mu_H(\colorVertMap) = 1$, meaning that this is a valid probability distribution.

We assume that $H$ must contain $\emptySeq$ and all prefixes of its sequences in order to represent a tree which starts from $r_{\rho,\sigma}$. Furthermore, after a path has left $\middleNodes$, the distribution on those nodes is $\probAssignEdge{\beta}{u,v} = \Id[v= \bot]$.

\begin{lemma}[Fixed color sequences do not escape]\label{lem:fixed_sequences_do_not_escape}
    Fix some \emph{valid} color sequence $\mathbf{\tau} = (\tau_1,\dots, \tau_h)\in H$. Then,
    \begin{align*}
        \Pr_{\colorVertMap \sim \mu_{H}(\colorVertMap)}[\exists i \in [h] \text{ s.t. } \colorVertMap(\mathbf{\tau}_{\leq i}) \notin \middleNodes] \leq \frac{8h}{\sqrt{N}}.
    \end{align*}
\end{lemma}
\begin{proof}
    Define the variable $X_0 = r_{\rho,\sigma}$ and $X_i = \colorVertMap(\tau_{\leq i})$. By \cref{eq:prob_assign_edge_distr,eq:probability_distribution_paths}, we have that,
    \begin{align*}
        \Pr[X_i = v|X_0, \dots, X_{i-1}] = \probAssignEdge{\tau_i}{X_{i-1}, v}.
    \end{align*}
    We say that a step goes \emph{upward} if it goes towards the root of $T_1$ or $T_2$ and downward otherwise. Notice that,
    \begin{align*}
        \upcolor(X_i) = \tau_i &\implies \text{step $i+1$ is not upward,}
    \end{align*}
    as if we just crossed the upcolor of the vertex $X_i$, we must have just gone downward. Therefore, any valid color sequence next either also goes downwards or crosses the weld. Therefore, an upward run starts either from $r_{\rho,\sigma}$ or after another crossing of the weld. Thus in order to reach $\midBoundNodes$, we must make at least $n/2-1$ steps upwards.

    Fix some path $X_0,\dots X_{i-1}$ with nonzero probability where step $i$ goes upward from depth $k$ and $\midBoundNodes$ has not been reached. By the definition of $\mu_H(\colorVertMap)$,
    \begin{align*}
        \Pr[\text{step $i+1$ is upward}|X_0,\dots X_{i-1}] &\leq \begin{cases}
            \frac{\colClassSize{k-1, \tau_{i+1}}}{\colClassSize{k,\tau_i}} &\text{in $T_1$,}\\
            \frac{\colClassSize{2n+2 -k, \tau_{i+1}}}{\colClassSize{2n+1 -k,\tau_i}} &\text{in $T_2$,}
        \end{cases}\\
        &\leq \frac{2^{k-1}/3 + 1}{2^{k}/3 - 1} &\mbox{(By \cref{lem:colors_per_column})}\\
        &\leq \frac{1}{2} + \frac{6}{\sqrt{N}} \coloneq \eta. &\mbox{(By \cref{eq:assumption_number_queries_bound})}
    \end{align*}
    Let $l=\tfrac{n}{2}-1$ and assume that $h\geq 1$ as otherwise it cannot reach the boundary.
    Let $R_{i,m}$ be the event that no vertex in $\midBoundNodes$ was reached before step $i$, $X_{i-1}$ is a leaf and that steps $i,\dots i+m-1$ are upwards. Therefore,
    \begin{align*}
        \{\exists j\in [h]: X_j \notin \middleNodes\} \subseteq \bigcup_{i\in [h-l+1]} R_{i,l}.
    \end{align*}
    Notice that for $m\in [l-1]$, $\Pr[R_{i,m+1}] \leq \eta \Pr[R_{i,m}]$. Therefore,
    \begin{align*}
        \Pr[R_{i,l}] \leq \eta^{l-1} \leq \frac{8}{\sqrt{N}}.
    \end{align*}
    The result follows by union-bounding over $i\in[h-l+1]$.
\end{proof}

In addition to tracking how the graph develops from $r$, we also need to track how the rest of the database develops. In our approach, we will view the tree from the root simply through the general structure of color assignments. On the other hand, the rest of it, which is connected component to the seed $s_{\rho,\sigma}$ and other disconnected components, are going to be considered using some database instance (which can also develop during the algorithm).

First, we define a set on mappings $\mapLabels$ which specify whether the label is inside the connected component of $r_{\rho,\sigma}$ or outside of it,
\begin{align*}
    \mapLabels &\coloneq \left(([0,M] \setminus \{r_{\rho,\sigma}\}) \times \{o\}\right) \cup (H \times \{i\}).
\end{align*}
Let $P: \mapLabels \rightharpoonup [0,M]$ be some injective map which represents the labels of these points. Without loss of generality we always assume that $P(0,o) = 0$, and count this pair in $\abs{P}$. The tuples in $P$ with $o$ represent points outside of the connected component of $r$ and those with $i$ are connected to $r$ via some path. The database $\outsideColorDb$ is an instance of an edge database such that $\{\theta \in \outsideColorDb: r_{\rho,\sigma} \in E(\theta)\} = \emptyset$ and $\outsideColorDb$ forms a valid instance which can be created using $\mcpwt$. Also, we require that $\outsideColorDb\cup\{\rho\}$ is an injective database.
We will group these variables together as follows,
\begin{align*}
    \reprTuple &\coloneq \{(H,P, \outsideColorDb): \abs{P} + \abs{\outsideColorDb}+ \abs{H} \leq T+1\}.
\end{align*}
We condition on the sizes as each variable represents a part of the database and by \cref{lem:cpwt_bounded_growth}, after $T$ steps the size of the domain is at most $T$. The additional $1$ is to account for the entrance vertex $0$ being assigned a priori. Under some fixed $D\in \reprTuple$, we use $\outsideLabelDb: [M] \rightharpoonup [M]$ to be the injective label database defined as $\outsideLabelDb(u) = P(u,o)$.
Furthermore, we will define the \emph{neighborhoods} of each structural vertex $v \in [M]$ as,
\begin{align*}
    &\nbhood(v) = \{v\} \cup \{u: \{u,v\}\in E_n\}, &\nbhood(S) = \bigcup_{s\in S}\nbhood(s).
\end{align*}
For the undefined value $\bot$, we have $\nbhood(\bot) = \emptyset$. Notice that as the degree of each node is at most $3$, we have that for any set $S$, $\abs{\nbhood(S)} \leq 4\abs{S}$. Suppose that we have some fixed $D\in \reprTuple$ and map $\colorVertMap$ such that $\mu_{H}(\colorVertMap) > 0$. We say that $\colorVertMap$ is \emph{good}, which we indicate using $\goodMap_{D}(\colorVertMap)$ as,
\begin{align}
    \goodMap_D(\colorVertMap) = \Id\left[ \begin{array}{l}
        \colorVertMap(H) \subseteq \middleNodes\\
        \nbhood(\colorVertMap(\tau)) \cap \nbhood(\colorVertMap(\nu)) = \emptyset \quad(\tau \neq \nu)\\
        \nbhood(\colorVertMap(\tau)) \cap \supp(\outsideLabelDb, \outsideColorDb \cup \{\rho\}) = \emptyset \quad (\tau\in H\setminus \{\emptySeq\})\\
        \mu_{H}(\colorVertMap) > 0
    \end{array} \right]\label{eq:goodness_definition}
\end{align}

Intuitively, the assignment is good if it does not escape the middle portion of the graph, all neighborhoods are disconnected and it can occur under $\mu_{H}(\colorVertMap)$. We find that most assignments are good.

\begin{lemma}[Goodness probability]\label{lem:goodness_probability}
    For every $D\in \reprTuple$,
    \begin{align*}
        \Pr_{\mu_H}[\goodMap_D = 0] \leq \delta \coloneq \frac{256(T+1)^2}{\sqrt{N}}.
    \end{align*}
\end{lemma}
\begin{proof}
    Fix some $D = (H,P, \outsideColorDb)$ and let $S = \supp(\outsideLabelDb, \outsideColorDb \cup \{\rho\})$. Notice that by the growth bounds, $\abs{S} \leq 4(T+1)$. Let $A = \{\colorVertMap: \colorVertMap(H)\not \subseteq \middleNodes\}$. Then by union-bounding over at most $T$ different paths $\tau\in H$,
    \begin{align}
        \Pr_{\mu_H}[A] &\leq \sum_{\tau \in H} \Pr[\exists i \leq \abs{\tau}: \colorVertMap(\tau_{\leq i}) \notin \middleNodes]\nonumber\\
        &\leq \frac{8}{\sqrt{N}}\sum_{\tau\in H} \abs{\tau} &\mbox{(By \cref{lem:fixed_sequences_do_not_escape})}\nonumber\\
        &\leq \frac{8(T+1)^2}{\sqrt{N}}. \label{eq:bound_prob_any_path_escapes}
    \end{align}
    Suppose we order the positions $\tau_0 = \emptySeq,\dots \tau_{\abs{H}-1} \in H$ such that if $\tau_i$ is in $\tau_j$, then $i\leq j$. Let $H_i = \{\tau_0,\dots,\tau_i\}$ and define,
    \begin{align*}
        S_i \coloneq S \cup \bigcup_{j\in [0,i-1]} \nbhood(\colorVertMap(\tau_j)).
    \end{align*}
    Notice that $\abs{S_i} \leq \abs{S} + 4i \leq 8(T+1)$ and $\abs{\nbhood(S_i)} \leq 32(T+1)$. Let $\tau_i = \nu\beta$. By \cref{eq:probability_distribution_paths},
    \begin{align}
        \Pr[\colorVertMap(\tau_i)=v| \colorVertMap|_{H_{i-1}}] = \probAssignEdge{\beta}{\colorVertMap(\nu),v} \label{eq:writeup_conditional_placement}
    \end{align}
    If all previous images lie in $\middleNodes$, the next image is uniform on the appropriate endpoints of $E_b$ where $b=\edgeIndic(\edgeVertex{\beta}{\colorVertMap(\nu)})$. Therefore,
    \begin{align}
        \Pr[\nbhood(\colorVertMap(\tau_i)) \cap S_i \neq \emptyset| \colorVertMap|_{H_{i-1}}] &\leq \frac{\abs{\nbhood(S_i)}}{\abs{E_b}} &\mbox{(By \cref{eq:writeup_conditional_placement})}\nonumber\\
        &\leq \frac{128(T+1)}{\sqrt{N}}. &\mbox{(By \cref{obs:min_num_nodes_per_perm})} \label{eq:bound_bad_nbh_event}
    \end{align}
    Let $B_i$ be the event that all previous images lie in $\middleNodes$ and the neighborhood intersection just bounded is nonempty. Therefore,
    \begin{align*}
        \Pr[\goodMap_D = 0] &\leq \Pr[A] + \sum_{i\in [\abs{H}-1]} \Pr[B_i] &&\mbox{(By \cref{eq:goodness_definition})}\\
        &\leq \frac{256(T+1)^2}{\sqrt{N}} = \delta &&\mbox{(By \cref{eq:bound_bad_nbh_event,eq:bound_prob_any_path_escapes})}\qedhere
    \end{align*}
\end{proof}

Next, let us describe how the pair $(D,\colorVertMap)$ can be transformed into an instance of the database using the map $\shapeToDb:\mapLabels \rightharpoonup [0,M]$. For some $\goodMap_D(\colorVertMap)=1$, we define it as,
\begin{align*}
    \shapeToDb(x) &= \begin{cases}
        u &\text{if $x=(u,o)$,}\\
        \colorVertMap(\tau) &\text{if $x=(\tau,i)$,}\\
        \bot &\text{otherwise.}
    \end{cases}
\end{align*}
Furthermore, for $\tau\beta\in H$, let,
\begin{align*}
    \theta_{\tau\beta}(\colorVertMap) &\coloneq \begin{cases}
        (\edgeVertex{\beta}{\colorVertMap(\tau)}, \edgeVertex{\beta}{\colorVertMap(\tau\beta)}) &\text{if $\colorVertMap(\tau) = \tail(\edgeVertex{\beta}{\colorVertMap(\tau)})$}\\
        (\edgeVertex{\beta}{\colorVertMap(\tau\beta)}, \edgeVertex{\beta}{\colorVertMap(\tau)}) &\text{if $\colorVertMap(\tau) = \head(\edgeVertex{\beta}{\colorVertMap(\tau)})$}
    \end{cases}
\end{align*}
Using this, we define the database $I_{\colorVertMap} = (L_{\colorVertMap}, C_{\colorVertMap})$ as,
\begin{align}
    L_{\colorVertMap}(\shapeToDb(x)) &\coloneq P(x), \quad(x\in \dom(P))\label{eq:writeup_database_definition_label} \\
    C_{\colorVertMap} &\coloneq \outsideColorDb \cup \{\rho\} \cup \{ \theta_{\tau\beta}(\colorVertMap): \tau\beta \in H \}. \label{eq:writeup_database_definition_color}
\end{align}
By definition, we immediately get the following.
\begin{lemma}\label{lem:valid_representation_transformation}
    For any instance $(D,\colorVertMap) \rightarrow I_{\colorVertMap}$ mapped as described above, $I_{\colorVertMap}$ is a valid injective database, its logical graph it represents is a forest, the mapping itself is injective and we have that,
    \begin{align*}
        \abs{L_{\colorVertMap}} + \abs{C_{\colorVertMap}} + 1 &= \abs{P} + \abs{\outsideColorDb}+ \abs{H},\\
        K_{\rho, \sigma}(I_{\colorVertMap}) &= (\colorVertMap(H), \{ \{\colorVertMap(\tau), \colorVertMap(\tau \beta): \tau \beta \in H\}  \}).
    \end{align*}
\end{lemma}
\begin{proof}
    Fix some good map $\colorVertMap$. By \cref{eq:goodness_definition}, $\colorVertMap(\tau) = \colorVertMap(\nu)$ implies that $\tau = \nu$. Furthermore by \cref{eq:goodness_definition}, $ \colorVertMap(H\setminus\{\emptySeq\})\cap\dom(\outsideLabelDb) = \emptyset$ and $\colorVertMap(\emptySeq) = r_{\rho, \sigma} \notin \dom(\outsideLabelDb)$, so,
    \begin{align*}
        \colorVertMap(H)\cap\dom(\outsideLabelDb)=\emptyset.
    \end{align*}
    Therefore $\shapeToDb$ is injective on $\dom(P)$, meaning that $L_\colorVertMap$ is also a partial injection such that $\abs{L_\colorVertMap} = \abs{P}-1$. Similarly, all $C_{\colorVertMap, b}$ are partial injections as $\outsideColorDb\cup\{\rho\}$ is a partial injection by definition and goodness with the validity of color assignments means that the records in $H$ use available slots.

    Furthermore, as $H$ is a forest by definition and $\outsideColorDb$ is only attainable by \mcpwt, by \cref{lem:no_cycles_in_graph} $\rG(I_{\colorVertMap})$ is a forest. Additionally, $K_{\rho,\sigma}(I_{\colorVertMap}) = (\colorVertMap(H), \{\{ \colorVertMap(\tau), \colorVertMap(\tau\beta)\}: \tau\beta\in H \})$, meaning that we retain the same structure.

    All the records in \cref{eq:writeup_database_definition_label,eq:writeup_database_definition_color} are distinct, so,
    \begin{align*}
        \abs{C_\colorVertMap} &= \abs{\outsideColorDb} + 1 + (\abs{H} - 1)\\
        \abs{L_\colorVertMap} + \abs{C_\colorVertMap} + 1 &= \abs{P} + \abs{\outsideColorDb} + \abs{H}.
    \end{align*}
    Finally, we may see that these are bijections, meaning that fixing $(D,\colorVertMap)$ is the same as fixing the database $I$.
\end{proof}

\subsubsection{Arbitrary modification}

Next, we define each canonical good vector and its normalization factor. For each $D\in \reprTuple$, we have,
\begin{align*}
    \canVec{D} &\coloneq \sum_{\colorVertMap: \goodMap_D(\colorVertMap) = 1} \sqrt{\mu_H(\colorVertMap)} \ket{L_{\colorVertMap}, C_{\colorVertMap}}_{\Lab \C},\\
    p_D &\coloneq \sum_{\colorVertMap} \mu_H(\colorVertMap) \goodMap_D(\colorVertMap).
\end{align*}
Notice that by \cref{lem:goodness_probability,lem:valid_representation_transformation},
\begin{align}
    &\braket{g_D}{g_{D^\prime}} = \Id[D = D^\prime]p_D, &1-\delta \leq p_D \leq 1.\label{eq:kronecker_database_instances}
\end{align}
For $j\in [0,T+1]$, we define the uniform projector $\prUnif{j}$ as follows,
\begin{align*}
    \prUnif{j} &\coloneq \sum_{D\in \reprTuple: \abs{D} \leq j} \frac{\ketbra{g_D}}{p_D},
\end{align*}
where $\abs{D} = \abs{P} + \abs{\outsideColorDb}+ \abs{H}$. Notice that $\prBoundary_{\rho,\sigma} \prUnif{j} = 0$ as by the definition of goodness, $K_{\rho, \sigma}(I)$ has all of its vertices in $\middleNodes$.

We immediately have the following.
\begin{proposition}\label{prop:unif_proj_algo_xor}
    Letting $V$ be either an algorithm unitary $A_i$ or an XOR operator $\xorOP$,
    \begin{align*}
        (\Id - \prUnif{j}) (\prWeld V \prWeld) \prUnif{j} = 0.
    \end{align*}
\end{proposition}
\begin{proof}
    An algorithm unitary only acts on $\A$ while the XOR $\xorOP$ only alters the output register $\Y$ and not the database registers. Furthermore, for some fixed $D$ and query, the XOR answer is already determined. Hence, the same operation acts on every summand of $\canVec{D}$.
\end{proof}

The remainder of the proof is mainly focused on how \pC operators behave under $\prUnif{j}$. Next, we define several different types of queries we will analyze and then prove they all behave similarly.

\paragraph{Types and assignment groups.} We define query types $\typeInv-\typeEdgeIn$ below, which characterize respectively unlabeling, labeling, querying an edge outside of $H$ and querying an edge inside $H$ except for $\rho$. 

Fix a step $j\in [0,T]$ and some algorithm basis state $a$ with query inputs $(l,\beta)$. We let $V$, which is one of $\pCLabelInv, \pCLabel, \mpCPerm$, denote some compression operator which is fixed by the type. The set $\famQParts_{V,a,j}$ contains sets of descriptions $\mathfrak{b}\subseteq \reprTuple$, which we call \emph{assignment groups}. Each group has a \emph{base} $D_\bot=(H,P,\outsideColorDb)\in\reprTuple$ where $\abs{D_\bot}\leq j$ and the value which is queried is absent, meaning it maps to $\bot$ in the database. Within a group, $I_{\colorVertMap}=(L_{\colorVertMap},C_{\colorVertMap})$ denotes the base database represented by $(D_\bot,\colorVertMap)$.

For types $\typeInv,\typeFor$ and $\typeEdgOut$, we specify a set $\assignVar$ of possible assignments and the corresponding extensions $D_z$. Formally, the assignment group is the set,
\begin{align*}
    \mathfrak{b} &\coloneq \{D_\bot\}\cup\{D_z:z\in\assignVar\}.%\label{eq:pc_complete_group}
\end{align*}
For type $\typeEdgeIn$, the group consists of the base $D_\bot$ and a single extension $D_1$, defined below.

\paragraph{Type $\typeInv$ - Inverse labels.} Let $V =\pCLabelInv$ and suppose that $l\neq 0$ and $l\in [M]\setminus \im(P)$, meaning that we are assigning the label to some free location. Define,
\begin{align*}
    \assignVar &\coloneq \mapLabels\setminus\dom(P),&D_z &\coloneq (H,P[z\to l],\outsideColorDb)\quad(z\in\assignVar).%\label{eq:pc_inverse_family}
\end{align*}
Notice that $z$ is a location in $\mapLabels$, whose structural value under $\colorVertMap$ is $\shapeToDb(z)$.

\paragraph{Type $\typeFor$ - Forward labels.} Let $V=\pCLabel$. For every good map $\colorVertMap$, the structural source $v= L^{-1}_{\colorVertMap}(l)$. Suppose there exists some structural location $w\in \mapLabels\setminus\dom(P)$ such that,
\begin{align*}
    v^\prime &= \stepFunct(l,\beta, L_{\colorVertMap}, C_{\colorVertMap}) = \shapeToDb(w), &v^\prime \notin \dom(L_{\colorVertMap})\cup \{0,\bot,v\}.
\end{align*}
This means that the query reaches the unlabeled location $w$, which is the queried argument in this case. We define,
\begin{align*}
    \assignVar &\coloneq [M]\setminus \im{P}, &D_z &\coloneq (H,P[w\to z],\outsideColorDb)\quad(z\in\assignVar).
\end{align*}
Hence $w$ is labeled with $z$.

For types $\typeEdgOut$ and $\typeEdgeIn$, we have $V=\mpCPerm$. At each good base map, we require that $l\in \im(L_\colorVertMap) \setminus \{0\}$ which maps to a structural vertex $v\in V_n$. We use $h$ for the queried structural edge, $b\in \middleEdge$ for its indicator and $\varsigma$ the direction of the query,
\begin{align*}
    h &\coloneq \edgeVertex{\beta}{v}\neq \bot, &b\coloneq \edgeIndic(h)\in\middleEdge, &\varsigma \coloneq \begin{cases}
        \rightarrow &\text{if $v=\tail(h)$,}\\
        \leftarrow &\text{if $v=\head(h)$.}
    \end{cases}
\end{align*}
\paragraph{Type $\typeEdgOut$ - Outside edges.} Suppose that $l=P(u,o)$ for some $u\neq r_{\rho,\sigma}$, meaning that $v=u$. Let $I^{\mathrm{out}} \coloneq (\outsideLabelDb,\outsideColorDb\cup\{\rho\})$. We require $h$ to be such that,
\begin{align*}
    \begin{cases}
        h\notin\dom((\outsideColorDb\cup\{\rho\})_b) &\text{if $\varsigma=\rightarrow$,}\\
        h\notin\im((\outsideColorDb\cup\{\rho\})_b) &\text{if $\varsigma=\leftarrow$,}
    \end{cases}
\end{align*}
meaning that the queried assignment is absent. Furthermore, remember that no edges in $\outsideColorDb$ touch the root $r_{\rho,\sigma}$. Let $F_H$ denote the set of edges used by the root, $F_H = \{\edgeVertex{\gamma}{r_{\rho,\sigma}}: \gamma\in H, \abs{\gamma}=1\}$. Using the fresh assignment sets from \cref{eq:available_assignments_fresh_right,eq:available_assignments_fresh_left},
\begin{align}
    \assignVar &\coloneq \begin{cases}
        \{z\in\calA^{\rightarrow}_{h,b,I^{\mathrm{out}}}:\head(z)\notin\bigcup_{g\in F_H}g\} &\text{if $\varsigma=\rightarrow$,}\\
        \{z\in\calA^{\leftarrow}_{h,b,I^{\mathrm{out}}}:\tail(z)\notin\bigcup_{g\in F_H}g\} &\text{if $\varsigma=\leftarrow$,}
    \end{cases}\label{eq:pc_outside_assignments}\\
    D_z &\coloneq \begin{cases}
        (H,P,\outsideColorDb[h\to z]) &\text{if $\varsigma=\rightarrow$,}\\
        (H,P,\outsideColorDb[z\to h]) &\text{if $\varsigma=\leftarrow$,}
    \end{cases}\qquad(z\in\assignVar).\label{eq:pc_outside_family}
\end{align}
Thus $z$ is an edge in $E_b$ which avoids $H$. Note that these assignments ensure that $\outsideColorDb\cup\{\rho\}$ is injective as the available sets are defined using $(\outsideLabelDb,\outsideColorDb)$.

\paragraph{Type $\typeEdgeIn$ - Inside edges.} Lastly, we are concerned with the edges in $H$. Suppose that the source is some location $\tau\in H$ in $K_{\rho,\sigma}(I_{\colorVertMap})$ which adds a child to $\tau$, meaning that $\tau\beta \notin H$ such that,
\begin{align*}
    \beta\neq\begin{cases}
        \alpha &\text{if $\tau=\emptySeq$,}\\
        \gamma &\text{if $\tau=\nu\gamma\neq\emptySeq$,}
    \end{cases}
\end{align*}
ensuring that $\tau \beta$ remains valid. Define,
\begin{align*}
    D_1 &\coloneq (H \cup \{\tau\beta\}, P, \outsideColorDb), &\mathfrak{b} \coloneq \{D_\bot,D_1\}.
\end{align*}
Notice that as $P$ does not change, the new leaf does not yet have a label. For any $z\in E_b$, extend the map to the new leaf $\colorVertMap_z$ by letting $\colorVertMap_z|_H = \colorVertMap$ and,
\begin{align}
    \colorVertMap_z(\tau\beta) &\coloneq \begin{cases}
        \head(z) &\text{if $\varsigma=\rightarrow$,}\\
        \tail(z) &\text{if $\varsigma=\leftarrow$.}
    \end{cases}
\end{align}

\paragraph{Available assignments.} Fix some assignment group $\mathfrak{b}\in\famQParts_{V,a,j}$ whose base is $D_\bot=(H,P,\outsideColorDb)$. Assume that $H$ contains $\emptySeq$ and all prefixes of any present color sequence. Recall that by definition of $L_{\colorVertMap}$,
\begin{align}
    \dom(L_{\colorVertMap}) &= \shapeToDb(\dom(P)\setminus\{(0,o)\}), &L_{\colorVertMap}(\shapeToDb(x)) &= P(x)\quad(x\in\dom(P)), &\abs{L_{\colorVertMap}} = \abs{P}-1.\label{eq:pc_label_compatibility}
\end{align}
For each $\colorVertMap$, let $\assignVar_\colorVertMap$ denote the assignments available in the represented database. We let $m_\colorVertMap \coloneq \abs{\assignVar_\colorVertMap}$ and $\ket{\xi_\bot(\colorVertMap)} \coloneq \ket{a,L_{\colorVertMap},C_{\colorVertMap}}_{\A\I}$. For each type, we define the available sets and corresponding database states as,
\begin{align}
    \begin{array}{c|c|c}
        \text{Type} & \assignVar_{\colorVertMap} & \ket{\xi_z(\colorVertMap)}\quad(z\in\assignVar_{\colorVertMap})\\\hline
        \typeInv
        & \assignVar\setminus\{(\colorVertMap(\nu),o):\nu\in H\setminus\{\emptySeq\}\}
        & \ket{a,L_{\colorVertMap}[\shapeToDb(z)\to l],C_{\colorVertMap}}_{\A\I}\\[4pt]
        \typeFor
        & [M]\setminus\im(L_{\colorVertMap})
        & \ket{a,L_{\colorVertMap}[\shapeToDb(w)\to z],C_{\colorVertMap}}_{\A\I}\\[4pt]
        \typeEdgOut,\typeEdgeIn,\ \varsigma=\rightarrow
        & \calA^\rightarrow_{h,b,I_{\colorVertMap}}
        & \ket{a,L_{\colorVertMap},C_{\colorVertMap}[h\to z]}_{\A\I}\\[4pt]
        \typeEdgOut,\typeEdgeIn,\ \varsigma=\leftarrow
        & \calA^\leftarrow_{h,b,I_{\colorVertMap}}
        & \ket{a,L_{\colorVertMap},C_{\colorVertMap}[z\to h]}_{\A\I}
    \end{array}\label{eq:pc_physical_extensions}
\end{align}
Notice that for type $\typeInv,\typeFor$ and $\typeEdgOut$, an available assignment $z$ contributes to $\canVec{D_z}$ only when $\goodMap_{D_z}(\colorVertMap)=1$ and that
\begin{align}
    \{z\in\assignVar:\goodMap_{D_z}(\colorVertMap)=1\} &\subseteq\assignVar_{\colorVertMap}\subseteq\assignVar.\label{eq:short_choice_condition}
\end{align}
For type $\typeEdgeIn$, the corresponding condition is $\goodMap_{D_1}(\colorVertMap_z)=1$. For type $\typeEdgeIn$, we will also use the following compatibility assumptions in the later estimates,
\begin{align}
    \mu_{H\cup\{\tau\beta\}}(\colorVertMap_z) &= \frac{\mu_H(\colorVertMap)}{\abs{E_b}}\qquad(z\in E_b),\nonumber\\
    \ket{a,I_{\colorVertMap_z}}_{\A\I} &= \ket{\xi_z(\colorVertMap)}
    \qquad(z\in\assignVar_{\colorVertMap},\ \goodMap_{D_1}(\colorVertMap_z)=1),\label{eq:pc_leaf_compatibility}
\end{align}
where $I_{\colorVertMap_z}$ is the database represented by $(D_1,\colorVertMap_z)$. Additionally, we have that every good map $\psi$ for $D_1$ is obtained this way,
\begin{align}
    \goodMap_{D_1}(\psi)=1
    \implies
    \exists\,\colorVertMap,\ z\in E_b:\quad
    \goodMap_{D_\bot}(\colorVertMap)=1,\qquad \psi=\colorVertMap_z,\label{eq:pc_leaf_existence}
\end{align}
where $\colorVertMap=\psi|_H$. Next, we define the base database $I_\bot(V,a,I)$ by removing the queried assignment,
\begin{align}
    \begin{array}{c|c}
        \text{Operator and condition} & I_\bot(V,a,I)\\\hline
        V=\pCLabelInv,\ l\neq 0
        & (L\setminus\{(z,l):z\in[M]\},C)\\[4pt]
        V=\pCLabel,\ v^\prime\notin\{\bot,0,v\}
        & (L\setminus\{(v^\prime,z):z\in[M]\},C)\\[4pt]
        V=\mpCPerm,\ \varsigma=\rightarrow
        & (L,C\setminus\{(h,z):z\in E_b\})\\[4pt]
        V=\mpCPerm,\ \varsigma=\leftarrow
        & (L,C\setminus\{(z,h):z\in E_b\})
    \end{array}\label{eq:pc_base_database}
\end{align}
If the conditions do not hold, set $I_\bot(V,a,I)=I$. Notice that $I_\bot(V,a,I_{\colorVertMap})= I_{\colorVertMap}$ as we are already assuming that the queried argument is not assigned. Lastly, notice that for any $z\in \assignVar$ and map $\colorVertMap$,
\begin{align}
    \goodMap_{D_z}(\colorVertMap) &\leq \goodMap_{D_\bot}(\colorVertMap),\label{eq:pc_support_implies_nesting}
\end{align}
as removing an edge cannot decrease goodness. Next we show several useful properties.

\begin{proposition}[Available assignments for types $\typeInv$-$\typeEdgOut$]\label{prop:short_choices}
    Fix some type in $\{\typeInv,\typeFor,\typeEdgOut\}$ and an assignment group $\mathfrak{b}\in\famQParts_{V,a,j}$ of this type. For every $\colorVertMap$ with $\goodMap_{D_\bot}(\colorVertMap)=1$,
    \begin{align}
        \begin{array}{c|c|c}
            \text{Type} & m_{\colorVertMap} & \abs{\assignVar\setminus\assignVar_{\colorVertMap}}\\\hline
            \typeInv & M+1-\abs{P} & \abs{H}-1\\[3pt]
            \typeFor & M+1-\abs{P} & 0\\[3pt]
            \typeEdgOut & \geq\abs{E_b}-\abs{L_{\colorVertMap}}-6\abs{C_{\colorVertMap}} & \leq 6\abs{D_\bot}
        \end{array}\label{eq:short_counts}
    \end{align}
    Furthermore,
    \begin{align}
        \ket{a}_{\A}\canVec{D_z}
        &= \sum_{\colorVertMap:\goodMap_{D_z}(\colorVertMap)=1}
        \sqrt{\mu_H(\colorVertMap)}\ket{\xi_z(\colorVertMap)}
        \qquad(z\in\assignVar),\label{eq:pc_vector_identification}
    \end{align}
\end{proposition}
\begin{proof}
    Fix some $z\in\assignVar$ and $\colorVertMap$ with $\goodMap_{D_\bot}(\colorVertMap)=1$. For \cref{eq:pc_vector_identification}, by \cref{eq:pc_support_implies_nesting,eq:short_choice_condition}, the base map is good and $z\in\assignVar_{\colorVertMap}$. Substituting the label or outside-edge update gives exactly the database in $\ket{\xi_z(\colorVertMap)}$ from \cref{eq:pc_physical_extensions}. As $H$ does not change, the weight $\mu_H(\colorVertMap)$ does not change either. Substituting these states into the definition of $\canVec{D_z}$ proves \cref{eq:pc_vector_identification}.
    
    Next, let us consider each type separately.
    Start with type $\typeInv$. By definition of $\mu_H$, $\colorVertMap(\emptySeq) = r_{\rho,\sigma}$. By the definition of goodness in \cref{eq:goodness_definition}, $\colorVertMap$ is injective and for all $\nu \in H\setminus \{\emptySeq\}, (\colorVertMap(\nu), o) \notin \dom(P)$. Therefore, the removed outside locations in \cref{eq:pc_physical_extensions} are $\abs{H}-1$ distinct members of $\assignVar$.

    The map $\shapeToDb$ restricted to $\assignVar_{\colorVertMap}$ is a bijection onto $[M]\setminus\dom(L_{\colorVertMap})$, meaning that,
    \begin{align*}
        \abs{\assignVar} &= M+\abs{H}-\abs{P},\\
        \abs{\assignVar\setminus\assignVar_{\colorVertMap}} &= \abs{H}-1,\\
        m_{\colorVertMap} &= M+1-\abs{P}=M-\abs{L_{\colorVertMap}}.
    \end{align*}

    For type $\typeFor$, the definition gives us that $\assignVar = \assignVar_\colorVertMap$, which implies \cref{eq:short_counts}. Lastly, the analysis for $\typeEdgOut$ follows from the analysis on $\calA^{\varsigma}_{h,b,I_{\colorVertMap}}$ in \cref{sec:fresh_oracle}.
\end{proof}
Notice that \cref{obs:min_num_nodes_per_perm,eq:assumption_number_queries_bound,prop:short_choices} immediately imply that,
\begin{align}
    m_{\colorVertMap} &\geq \frac{\sqrt{N}}{8}, &
    \abs{\assignVar\setminus\assignVar_{\colorVertMap}} &\leq 6(T+1).\label{eq:pc_uniform_available_bounds}
\end{align}

Next, we characterize the orthogonality relationships between different states $\ket{\xi_z(\colorVertMap)}$.
\begin{proposition}[Base databases]\label{prop:base_databases_ortho}
    Fix some compression operator $V$ and a database $I$ such that $I = I_\bot(V,a,I)$. For every database $J$ in the associated $\ket{+}$ state,
    \begin{align}
        I_\bot(V,a,J)=I\label{eq:pc_recover_base}
    \end{align}
    Additionally,
    \begin{align}
        V\ket{\xi_\bot(\colorVertMap)}
        &= \frac{1}{\sqrt{m_{\colorVertMap}}}\sum_{y\in\assignVar_{\colorVertMap}}\ket{\xi_y(\colorVertMap)},\label{eq:pc_empty_insertion}\\
        V\ket{\xi_z(\colorVertMap)}
        &= \ket{\xi_z(\colorVertMap)}
        +\frac{\ket{\xi_\bot(\colorVertMap)}-V\ket{\xi_\bot(\colorVertMap)}}{\sqrt{m_{\colorVertMap}}}.\label{eq:pc_existing_insertion_specialized}
    \end{align}
    Furthermore, for every $\colorVertMap,\colorVertMap^\prime$ such that $\goodMap_{D_\bot}(\colorVertMap)=\goodMap_{D_\bot}(\colorVertMap^\prime)=1$,
    \begin{align}
        \braket{\xi_y(\colorVertMap)}{\xi_z(\colorVertMap^\prime)}
        &= \Id[\colorVertMap=\colorVertMap^\prime]\Id[y=z],\label{eq:pc_expanded_orthogonality}
    \end{align}
    where $y\in\{\bot\}\cup\assignVar_{\colorVertMap}$ and $z\in\{\bot\}\cup\assignVar_{\colorVertMap^\prime}$.
\end{proposition}
\begin{proof}
    Notice that the result that $I_\bot(V,a,J)=I$ is immediate from definition as every $\pC$ operator removes one assignment which is exactly the assignment removed in $I_\bot(V,a,J)$. Furthermore, \cref{eq:pc_existing_insertion_specialized} is directly due to \cref{lem:expansion_pc_operator}.
    
    For \cref{eq:pc_expanded_orthogonality}, notice that,
    \begin{align*}
        \ket{\xi_y(\colorVertMap)} = \ket{\xi_z(\colorVertMap^\prime)} &\implies I_\colorVertMap = I_{\colorVertMap^\prime} &\mbox{(By \cref{eq:pc_recover_base})}\\
        &\implies \colorVertMap = \colorVertMap^\prime &\mbox{(By \cref{lem:valid_representation_transformation})}.
    \end{align*}
    Therefore we have that $y=z$ as the assignments must be the same for the vectors to be equal.
\end{proof}

Next we fix the states we compare. Fix some type in $\{\typeInv,\typeFor,\typeEdgOut\}$ and an assignment group $\mathfrak{b}\in\famQParts_{V,a,j}$. We will be working in the orthonormal basis using $\ket{\bot}$ $\{\ket{z}\}_{z\in \assignVar}$. Define,
\begin{align}
    \ket{\assignVar} &\coloneq \frac{1}{\sqrt{\abs{\assignVar}}}\sum_{z\in\assignVar}\ket{z}, &
    \ket{\assignVar_{\colorVertMap}} &\coloneq \frac{1}{\sqrt{m_{\colorVertMap}}}\sum_{z\in\assignVar_{\colorVertMap}}\ket{z},\nonumber\\
    \prGoodAssign &\coloneq \ketbra{\bot}+\sum_{\substack{z\in\assignVar\\\goodMap_{D_z}(\colorVertMap)=1}}\ketbra{z}.\label{eq:pc_coordinate_averages}
\end{align}
Thus $\prGoodAssign$ retains the good subspace. For arbitrary coefficients $c_\bot\in\mathbb{C}$ and $(c_z)_{z\in\assignVar}\in\mathbb{C}^{\assignVar}$, define,
\begin{align}
    c_\bot^\prime &\coloneq \frac{1}{\sqrt{\abs{\assignVar}}}\sum_{z\in\assignVar}c_z, &
    c_z^\prime &\coloneq c_z+\frac{c_\bot-c_\bot^\prime}{\sqrt{\abs{\assignVar}}},\nonumber\\
    \ket{c} &\coloneq c_\bot\ket{\bot}+\sum_{z\in\assignVar}c_z\ket{z}, &
    \ket{c^\prime} &\coloneq c_\bot^\prime\ket{\bot}+\sum_{z\in\assignVar}c_z^\prime\ket{z}.\label{eq:pc_coefficients}
\end{align}
The corresponding database vectors which we will compare are,
\begin{align}
    \zeta_{\mathfrak{b}} &\coloneq \ket{a}_{\A}\left(c_\bot\canVec{D_\bot}+\sum_{z\in\assignVar}c_z\canVec{D_z}\right),\nonumber\\
    \zeta_{\mathfrak{b}}^\prime &\coloneq \ket{a}_{\A}\left(c_\bot^\prime\canVec{D_\bot}+\sum_{z\in\assignVar}c_z^\prime\canVec{D_z}\right).\label{eq:pc_compared_vectors}
\end{align}
We first bound the weight removed by $\prGoodAssign$.
\begin{observation}[Good weight]\label{obs:pc_omitted_weight}
    Fix some vector $\ket{d}=d_\bot\ket{\bot}+\sum_{z\in\assignVar}d_z\ket{z}$ whose coefficients are independent of $\colorVertMap$. Then,
    \begin{align}
        \sum_{\colorVertMap:\goodMap_{D_\bot}(\colorVertMap)=1}\mu_H(\colorVertMap)\norm{(\Id-\prGoodAssign)\ket{d}}^2
        &= \sum_{z\in\assignVar}\abs{d_z}^2(p_{D_\bot}-p_{D_z})
        \leq \delta\norm{\ket{d}}^2.\label{eq:pc_omitted_weight}
    \end{align}
\end{observation}
\begin{proof}
    By \cref{eq:pc_coordinate_averages} and orthonormality of the basis,
    \begin{align*}
        \norm{(\Id - \prGoodAssign)\ket{d}}^2 = \sum_{z\in \assignVar} \abs{d_z}^2(1 - \goodMap_{D_z}(\colorVertMap)).
    \end{align*}
    Since $\goodMap_{D_z}(\colorVertMap)\leq\goodMap_{D_\bot}(\colorVertMap)$ by \cref{eq:pc_support_implies_nesting} and the coefficients are independent of $\colorVertMap$,
    \begin{align*}
        \sum_{\colorVertMap:\goodMap_{D_\bot}(\colorVertMap)=1}\mu_H(\colorVertMap)\norm{(\Id - \prGoodAssign)\ket{d}}^2 &= \sum_{z\in \assignVar} \abs{d_z}^2 \sum_{\colorVertMap} \mu_H(\colorVertMap)\goodMap_{D_\bot}(\colorVertMap)(1 - \goodMap_{D_z}(\colorVertMap))\\
        &=\sum_{z\in \assignVar} \abs{d_z}^2 \sum_{\colorVertMap} \mu_H(\colorVertMap)(\goodMap_{D_\bot}(\colorVertMap) - \goodMap_{D_z}(\colorVertMap))\\
        &= \sum_{z\in \assignVar} \abs{d_z}^2 (p_{D_\bot} - p_{D_z})\\
        &\leq \delta \sum_{z\in \assignVar} \abs{d_z}^2 \leq \delta \norm{\ket{d}}^2,
    \end{align*}
    where the last line uses the fact that $0\leq p_{D_\bot} - p_{D_z} \leq \delta$ due to \cref{eq:pc_support_implies_nesting,eq:kronecker_database_instances}.
\end{proof}

We may prove the main statement for types $\typeInv$-$\typeEdgOut$.

\begin{lemma}[Estimate for types $\typeInv$-$\typeEdgOut$]\label{lem:pc_types_i_iii_estimate}
    Fix some type in $\{\typeInv, \typeFor, \typeEdgOut\}$ and $\mathfrak{b} \in \famQParts_{V,a,j}$ of this type. For any choice of coefficients in \cref{eq:pc_coefficients} and corresponding vectors in \cref{eq:pc_compared_vectors},
    \begin{align}
        \norm{V\zeta_{\mathfrak{b}}-\zeta_{\mathfrak{b}}^\prime}^2
        &\leq 16\delta\left(\abs{c_\bot}^2+\sum_{z\in\assignVar}\abs{c_z}^2\right)
        \leq 32\delta\norm{\zeta_{\mathfrak{b}}}^2.\label{eq:pc_common_bound}
    \end{align}
\end{lemma}
\begin{proof}
    For each good base map $\colorVertMap$, let $\pC_{\assignVar}=\pCfromPlus(\ket{\assignVar}, \ket{\bot})$ and $\pC_\colorVertMap =\pCfromPlus(\ket{\assignVar_\colorVertMap}, \ket{\bot})$. By \cref{obs:superposition_similarity},
    \begin{align}
        \norm{\pC_{\assignVar} - \pC_\colorVertMap}^2 &\leq \frac{4\abs{\assignVar\setminus\assignVar_{\colorVertMap}}}{m_{\colorVertMap}}\nonumber\\
        &\leq \frac{256(T+1)^2}{\sqrt{N}} = \delta.\label{eq:pc_direct_comparison}
    \end{align}
    By definition of the coefficients in \cref{eq:pc_coefficients},
    \begin{align}
        \pC_{\assignVar}\ket{c} = \frac{\sum_{z\in\assignVar}c_z}{\sqrt{\abs{\assignVar}}}\ket{\bot}
        +\sum_{z\in\assignVar}\left(c_z+\frac{c_\bot-c_\bot^\prime}{\sqrt{\abs{\assignVar}}}\right)\ket{z} = \ket{c^\prime}.\label{eq:pc_c_to_c_prime}
    \end{align}
    As $\pC_{\assignVar}$ is a unitary, $\norm{\ket{c}}^2 = \norm{\ket{c^\prime}}^2$. Notice that both coefficient vectors are independent of $\colorVertMap$. Let $E_\colorVertMap = \pC_\colorVertMap \prGoodAssign - \prGoodAssign \pC_{\assignVar}$. Using \cref{eq:pc_c_to_c_prime}, we may decompose $E_\colorVertMap$ as,
    \begin{align*}
        E_\colorVertMap\ket{c} &= \pC_\colorVertMap \prGoodAssign\ket{c} - \prGoodAssign\ket{c^\prime}\\
        &=(\pC_\colorVertMap - \pC_{\assignVar})\prGoodAssign\ket{c} - \pC_{\assignVar}(\Id - \prGoodAssign)\ket{c} + (\Id - \prGoodAssign)\ket{c^\prime}
    \end{align*}
    Applying the triangle inequality and the fact that $\pC_{\assignVar}$ is a unitary,
    \begin{align*}
        \norm{E_\colorVertMap\ket{c}}^2 &\leq 3\left( \norm{\pC_\colorVertMap - \pC_{\assignVar}}^2 \norm{\prGoodAssign\ket{c}}^2 + \norm{(\Id - \prGoodAssign)\ket{c}}^2 + \norm{(\Id - \prGoodAssign)\ket{c^\prime}}^2 \right)\\
        &\leq 3 \left(\delta \norm{\ket{c}}^2 + \norm{(\Id - \prGoodAssign)\ket{c}}^2 + \norm{(\Id - \prGoodAssign)\ket{c^\prime}}^2\right). &\mbox{(By \cref{eq:pc_direct_comparison})}
    \end{align*}
    We now average this quantity over good maps. The first term depends on $\colorVertMap$ only through its weight. For the other two terms, we apply \cref{eq:pc_omitted_weight}. Therefore,
    \begin{align}
        \sum_{\colorVertMap:\goodMap_{D_\bot}(\colorVertMap)=1}&\mu_H(\colorVertMap) \norm{E_{\colorVertMap}\ket{c}}^2\nonumber\\
        &\leq 3\left(\sum_{\colorVertMap:\goodMap_{D_\bot}(\colorVertMap)=1}\mu_H(\colorVertMap) \left( \delta \norm{\ket{c}}^2 + \norm{(\Id - \prGoodAssign)\ket{c}}^2 + \norm{(\Id - \prGoodAssign)\ket{c^\prime}}^2\right) \right)\nonumber\\
        &\leq 3\left(\delta p_{D_\bot}\norm{\ket{c}}^2 + \sum_{z\in \assignVar}\abs{c_z}^2(p_{D_\bot} - p_{D_z}) + \sum_{z\in \assignVar}\abs{c_z^\prime}^2(p_{D_\bot} - p_{D_z}) \right) &\mbox{(By \cref{obs:pc_omitted_weight})}\nonumber\\
        &\leq 3\delta \left( p_{D_\bot}\norm{\ket{c}}^2 + \sum_{z\in \assignVar}\abs{c_z}^2 + \sum_{z\in \assignVar}\abs{c_z^\prime}^2 \right) &\mbox{(By \cref{eq:kronecker_database_instances})}\nonumber\\
        &\leq 9\delta \norm{\ket{c}}^2.\label{eq:bound_on_mu_e}
    \end{align}
    Next, we express the database vectors using these coefficients. For a good base map $\colorVertMap$, define,
    \begin{align*}
        T_\colorVertMap\ket{\bot} \coloneq \ket{\xi_{\bot}(\colorVertMap)}, &T_\colorVertMap\ket{z} \coloneq \ket{\xi_{z}(\colorVertMap)} \qquad(z\in\assignVar_{\colorVertMap}),
    \end{align*}
    and let $\calK_\colorVertMap \coloneq \Span\{\ket{\bot}, \ket{z}: z\in \assignVar_{\colorVertMap}\}$. Notice that by \cref{eq:pc_expanded_orthogonality}, $T_\colorVertMap$ preserves orthogonality and by \cref{eq:pc_empty_insertion,eq:pc_existing_insertion_specialized}, for any $\ket{x}\in \calK_\colorVertMap$,
    \begin{align*}
        V T_\colorVertMap \ket{x} = T_\colorVertMap\pC_\colorVertMap \ket{x}.
    \end{align*}
    By \cref{eq:pc_vector_identification,eq:pc_support_implies_nesting,eq:pc_compared_vectors},
    \begin{align*}
        \zeta_{\mathfrak{b}} &= \sum_{\colorVertMap:\goodMap_{D_\bot}(\colorVertMap)=1} \sqrt{\mu_H(\colorVertMap)}T_{\colorVertMap}\prGoodAssign\ket{c},\\
        \zeta_{\mathfrak{b}}^\prime &= \sum_{\colorVertMap:\goodMap_{D_\bot}(\colorVertMap)=1} \sqrt{\mu_H(\colorVertMap)}T_{\colorVertMap}\prGoodAssign\ket{c^\prime},\\
        V\zeta_{\mathfrak{b}}-\zeta_{\mathfrak{b}}^\prime &= \sum_{\colorVertMap:\goodMap_{D_\bot}(\colorVertMap)=1} \sqrt{\mu_H(\colorVertMap)}T_{\colorVertMap}E_{\colorVertMap}\ket{c}.
    \end{align*}
    Applying the properties of $T_\colorVertMap$,
    \begin{align*}
        \norm{V\zeta_{\mathfrak{b}}-\zeta_{\mathfrak{b}}^\prime}^2 &= \sum_{\colorVertMap:\goodMap_{D_\bot}(\colorVertMap)=1} \mu_H(\colorVertMap)\norm{E_{\colorVertMap}\ket{c}}^2\\
        &\leq 9\delta \norm{c}^2 &\mbox{(By \cref{eq:bound_on_mu_e})}
    \end{align*}
    Finally,
    \begin{align*}
        \norm{\zeta_{\mathfrak{b}}}^2 &= p_{D_\bot}\abs{c_\bot}^2 + \sum_{z\in \assignVar} p_{D_z} \abs{c_z}^2\\
        &\geq (1-\delta) \norm{\ket{c}}^2.
    \end{align*}
    As $\delta < \tfrac{1}{2}$, the statement follows.
\end{proof}

Next, we show that the same estimate for edges inside $H$ holds as well.

\begin{lemma}[Estimate for type $\typeEdgeIn$]\label{lem:pc_type_iv_similarity}
    Fix some assignment group $\mathfrak{b}=\{D_\bot,D_1\}\in\famQParts_{\mpCPerm,a,j}$ of type $\typeEdgeIn$ and coefficients $c_\bot,c_1\in\mathbb{C}$. Then letting $\zeta_\mathfrak{b}, \zeta_\mathfrak{b}^\prime$ be the states from \cref{eq:pc_compared_vectors}, with $c_\bot^\prime=c_1$ and $c_1^\prime=c_\bot$,
    \begin{align*}
        \norm{\mpCPerm\zeta_\mathfrak{b} - \zeta_\mathfrak{b}^\prime}^2 \leq 4\delta\norm{\zeta_\mathfrak{b}}^2.
    \end{align*}
\end{lemma}
\begin{proof}
    Fix some $\colorVertMap$ such that $\goodMap_{D_\bot}(\colorVertMap)=1$. Any good mapping is a valid label between edges, meaning that by \cref{obs:min_num_nodes_per_perm,eq:assumption_number_queries_bound}, the fresh assignment sets in \cref{eq:available_assignments_fresh_right,eq:available_assignments_fresh_left} satisfy $m_\colorVertMap \geq \tfrac{\sqrt{N}}{8}$.

    Notice that by goodness, $\supp(I_\colorVertMap) \subseteq \supp(\outsideLabelDb,\outsideColorDb\cup\{\rho\}) \cup \left( \cup_{\nu \in H} \nbhood(\colorVertMap(\nu)) \right)$. Therefore if $\goodMap_{D_1}(\colorVertMap_z) = 1$, by \cref{eq:goodness_definition}, $\{\tail(z),\head(z)\}\subseteq\nbhood(\colorVertMap_z(\tau\beta))$ and this neighborhood is disjoint from $\supp(I_\colorVertMap)$. Therefore,
    \begin{align}
        \{ z\in E_b: \goodMap_{D_1}(\colorVertMap_z)=1 \} \subseteq \assignVar_\colorVertMap \subseteq E_b.\label{eq:pc_leaf_good_inclusion}
    \end{align}
    Additionally, by \cref{eq:pc_leaf_existence} for each $z$ this assignment is unique, meaning under one mapping there is exactly one assignment $z$ which exactly sets the map. Therefore, by \cref{eq:pc_leaf_compatibility,eq:pc_leaf_existence,eq:pc_leaf_good_inclusion,eq:pc_empty_insertion},
    \begin{align}
        V\ket{a}\canVec{D_\bot}
        &= \sum_{\colorVertMap:\goodMap_{D_\bot}(\colorVertMap)=1}\sqrt{\mu_H(\colorVertMap)}
        \sum_{z\in\assignVar_{\colorVertMap}}\frac{\ket{\xi_z(\colorVertMap)}}{\sqrt{m_{\colorVertMap}}},\label{pc_leaf_vector_apply}\\
        \ket{a}\canVec{D_1}
        &= \sum_{\colorVertMap:\goodMap_{D_\bot}(\colorVertMap)=1}\sqrt{\mu_H(\colorVertMap)}
        \sum_{z\in\assignVar_{\colorVertMap}}\frac{\goodMap_{D_1}(\colorVertMap_z)}{\sqrt{\abs{E_b}}}\ket{\xi_z(\colorVertMap)}.\label{eq:pc_leaf_vector_goal}\\
        p_{D_1}
        &= \sum_{\colorVertMap:\goodMap_{D_\bot}(\colorVertMap)=1}
        \frac{\mu_H(\colorVertMap)}{\abs{E_b}}\sum_{z\in E_b}\goodMap_{D_1}(\colorVertMap_z).\label{eq:pc_leaf_weight}
    \end{align}
    Note that the index $b$ and direction $\varsigma$ are evaluated based on $\colorVertMap$. By \cref{prop:base_databases_ortho,pc_leaf_vector_apply,eq:pc_leaf_vector_goal},
    \begin{align*}
        \norm{\mpCPerm \ket{a}\canVec{D_\bot} - \ket{a}\canVec{D_1}}^2 &= \sum_{\colorVertMap:\goodMap_{D_\bot}(\colorVertMap)=1} \mu_H(\colorVertMap) \sum_{z\in \assignVar_{\colorVertMap}} \abs{\frac{1}{\sqrt{m_{\colorVertMap}}} - \frac{\goodMap_{D_1}(\colorVertMap_z)}{\sqrt{\abs{E_b}}}}^2\\
        &\leq \sum_{\colorVertMap:\goodMap_{D_\bot}(\colorVertMap)=1} \mu_H(\colorVertMap) \left(1 - \frac{1}{\abs{E_b}} \sum_{z\in E_b}\goodMap_{D_1}(\colorVertMap_z) \right) &\mbox{(As $m_\colorVertMap \leq \abs{E_b}$)}\\
        &= p_{D_\bot} - p_{D_1} \leq \delta &\mbox{(By \cref{eq:pc_leaf_weight,eq:kronecker_database_instances})}.
    \end{align*}
    Applying the bound above and using the fact that $\mpCPerm$ is unitary,
    \begin{align*}
        \norm{\mpCPerm \ket{a}\canVec{D_1} - \ket{a}\canVec{D_\bot}}^2 &=\norm{\mpCPerm(\mpCPerm \ket{a}\canVec{D_1} - \ket{a}\canVec{D_\bot})}^2\\
        &=\norm{\mpCPerm \ket{a}\canVec{D_\bot} - \ket{a}\canVec{D_1}}^2 \leq \delta.
    \end{align*}
    Therefore as $\mpCPerm$ approximately switches between $\canVec{D_\bot}$ and $\canVec{D_1}$,
    \begin{align*}
        \norm{\mpCPerm \zeta_\mathfrak{b} - \zeta_\mathfrak{b}^\prime}^2 &= \norm{c_\bot(\mpCPerm\ket{a}\canVec{D_\bot} - \ket{a}\canVec{D_1}) + c_1(\mpCPerm\ket{a}\canVec{D_1} - \ket{a}\canVec{D_\bot})}^2\\
        &\leq \delta(\abs{c_\bot} + \abs{c_1})^2 \leq 2\delta(\abs{c_\bot}^2 + \abs{c_1}^2).
    \end{align*}
    Finally, by \cref{eq:kronecker_database_instances}, $\norm{\zeta_\mathfrak{b}}^2 \geq (1-\delta)(\abs{c_\bot}^2 + \abs{c_1}^2)$. As $\delta < \tfrac{1}{2}$, the statement follows.
\end{proof}

Finally, we consider the specific subcase which has not been covered by any of the cases above, that being querying the fixed weld $\rho$ itself.

\begin{lemma}[Querying the weld $\rho$]\label{lem:mpcPerm_weld_similarity}
    For any $\rho = (\alpha, e,f)$, let $\zeta = \sum_{a,I} c_{a,I} \ket{a,I}_{\A\I}$ be a vector such that $\prWeld\zeta=\zeta$, $\abs{L} + \abs{C} \leq T+1$ and $a$ represents the query $(l,\alpha)$ where $L^{-1}(l) \in \{\tail(e), \head(f) \}$. Then,
    \begin{align*}
        \norm{(\prWeld \mpCPerm \prWeld - \Id) \zeta} \leq \frac{8 \norm{\zeta}}{\sqrt{N}}
    \end{align*}
\end{lemma}
\begin{proof}
    As the direction $\sigma$ of the queries act on orthogonal bases $(a,L)$, assume that $L^{-1}(l) = \tail(e)$. The other case follows analogously. Let $b=(\mathrm{wld},\alpha)$ and
    \begin{align*}
        m(a,L,C) = \abs{\calA_{e,b,(L,C\setminus \{\rho\})}^{\rightarrow}},
    \end{align*}
    use $m$ for $m(a,L,C)$ when the input is obvious. We may assume that $f$ is in the $\calA$ set as otherwise $\mpCPerm$ would act as the identity on this basis state. Therefore, as $\rho$ is in the database $C$ for all databases in $\zeta$, by \cref{lem:expansion_pc_operator},
    \begin{align}
        \mpCPerm \ket{a,L,C} = \left(1 - \frac{1}{m} \right) \ket{a,L,C} + \frac{1}{\sqrt{m}} \ket{a,L,C\setminus\{\rho\}} - \frac{1}{m} \sum_{\substack{z \in \calA_{e,b,(L,C\setminus \{\rho\})}^{\rightarrow}\\ z\neq f}} \ket{a,L,C\setminus \{\rho\}[e\to z]}.\label{eq:apply_mpcperm_db}
    \end{align}
    Notice that the database instances in the second and third term of \cref{eq:apply_mpcperm_db} are $0$ under $\prWeld$ as they do not include $\rho$. Notice that by \cref{obs:min_num_nodes_per_perm,eq:assumption_number_queries_bound},
    \begin{align*}
        m(a,L,C) \geq \abs{E_b} - \abs{L} - 6\abs{C} \geq \frac{\sqrt{N}}{8}.
    \end{align*}
    Therefore by \cref{eq:apply_mpcperm_db} and the orthogonality of each database $I$,
    \begin{align*}
        \norm{(\prWeld \mpCPerm \prWeld - \Id) \zeta}^2 &\leq \sum_{a,L,C} \frac{\abs{c_{a,L,C}}^2}{m(a,L,C)^2} \\
        &\leq \frac{64}{N}\sum_{a,L,C}\abs{c_{a,L,C}}^2 = \frac{64}{N}\norm{\zeta}^2.
    \end{align*}
    The statement follows by taking the square root.
\end{proof}

\paragraph{Error orthogonality.} Next, we show that any nonzero error coefficient determines its assignment group $\mathfrak{b}$. First, notice that for any fixed $\mathfrak{b}\in \famQParts_{V,a,j}$, only one database attains the minimum size $\abs{D}$, that being the base database $D_\bot$ and this is unique for each group $\mathfrak{b}$.

\begin{proposition}\label{prop:non_zero_support_good_map}
    Fix a group $\mathfrak{b}\in\famQParts_{V,a,j}$, an algorithm basis $a$ and a database $J$. If
    \begin{align*}
        \braket{a,J}{V \zeta_{\mathfrak{b}} - \zeta_{\mathfrak{b}}^\prime} \neq 0,
    \end{align*}
    then there exists some $\colorVertMap$ such that $\goodMap_{D_\bot}(\colorVertMap)=1$ and $I_\bot(V,a,J) = I_\colorVertMap$.
\end{proposition}
\begin{proof}
    Let $D_\bot$ be the base of $\mathfrak{b}$. For every $\colorVertMap$ such that $\goodMap_{D_\bot}(\colorVertMap)=1$, define,
    \begin{align*}
        &\calS_\colorVertMap \coloneq \Span\{\ket{\xi_y(\colorVertMap)}: y\in \{\bot\}\cup \assignVar_{\colorVertMap} \}, &&\calS_{\mathfrak{b}} \coloneq \sum_{\colorVertMap: \goodMap_{D_\bot}(\colorVertMap)=1} \calS_\colorVertMap.
    \end{align*}
    Notice that the base vector $\ket{a}\canVec{D_\bot}$ is in $\calS_{\mathfrak{b}}$ and for any $z\in \assignVar$, $\ket{a}\canVec{D_z} \in \calS_\mathfrak{b}$. Note that for type $\typeEdgeIn$, $\ket{a}\canVec{D_1} \in \calS_\mathfrak{b}$. By their definition in \cref{eq:pc_existing_insertion_specialized,eq:pc_empty_insertion} and linearity, we have that $V \zeta_\mathfrak{b} - \zeta_\mathfrak{b}^\prime \in \calS_{\mathfrak{b}}$.

    Thus, if $\braket{a,J}{V \zeta_\mathfrak{b} - \zeta_\mathfrak{b}^\prime} \neq 0$, there must exist a good map $\colorVertMap$ and $y\in \{\bot\}\cup \assignVar_{\colorVertMap}$ such that $\ket{a,J} = \ket{\xi_y(\colorVertMap)}$. Whatever $y$ is, we have that $I_\bot(V,a,J) = I_\colorVertMap$ is fixed.
\end{proof}

Next, we show that errors of the vectors are orthogonal to each other, which will be useful for measuring interference.

\begin{lemma}[Error orthogonality]\label{lem:error_orthogonality} Fix a primitive $V$ and two different groups $\mathfrak{b}, \mathfrak{c} \in \famQParts_{V,a,j}$. Then,
    \begin{align}
        \innerproduct{V \zeta_{\mathfrak{b}} - \zeta_{\mathfrak{b}}^\prime}{V \zeta_{\mathfrak{c}} - \zeta_{\mathfrak{c}}^\prime} = 0.\label{eq:error_orthogonality_res}
    \end{align}
\end{lemma}
\begin{proof}
    Let $D_\bot, D_\bot^\prime$ be the bases of $\mathfrak{b}, \mathfrak{c}$. Notice that for any group, there is exactly one base. Suppose that \cref{eq:error_orthogonality_res} does not hold, meaning there must exist at least one algorithm basis $a$ and database $J$ such that $\braket{a,J}{V \zeta_{\mathfrak{b}} - \zeta_{\mathfrak{b}}^\prime} \neq 0$ and $\braket{a,J}{V \zeta_{\mathfrak{c}} - \zeta_{\mathfrak{c}}^\prime} \neq 0$. By \cref{prop:non_zero_support_good_map}, there must exist some $\phi, \psi$ such that,
    \begin{align*}
        \goodMap_{D_\bot}(\phi)=\goodMap_{D_\bot^\prime}(\psi)=1\text{, and } I_\phi = I_\bot(V,a,J) = I^\prime_{\psi}.
    \end{align*}
    By \cref{lem:valid_representation_transformation}, this means that their database representations are the same. For fixed $V$ and $a$, this base also determines the query type: for an edge query, $P^{-1}(l)$ specifies whether the source is inside or outside $H$. Hence $\mathfrak{b} = \mathfrak{c}$, as the base and type determine the group. This is a contradiction, implying the statement.
\end{proof}

Lastly, the result below will be useful.

\begin{observation}\label{prop:proj_weld_on_group_states} Fix some type $\{ \typeInv, \typeFor, \typeEdgOut, \typeEdgeIn\}$ and a group $\mathfrak{b}$. Then,
    \begin{align}
        &\prWeld \zeta_{\mathfrak{b}} = \zeta_{\mathfrak{b}}, &&\prWeld V\zeta_{\mathfrak{b}} = V\zeta_{\mathfrak{b}}, &&\prWeld \zeta_{\mathfrak{b}}^\prime = \zeta_{\mathfrak{b}}^\prime. \label{eq:proj_weld_on_group_states}
    \end{align}
\end{observation}
\begin{proof}
    Every represented database contains $\rho$. Therefore, for any $D\in \mathfrak{b}$, $\prWeld \ket{a}\canVec{D} = \ket{a}\canVec{D}$. By \cref{eq:pc_compared_vectors},
    \begin{align}
        &\prWeld \zeta_{\mathfrak{b}} = \zeta_{\mathfrak{b}}, &\prWeld \zeta_{\mathfrak{b}}^\prime = \zeta_{\mathfrak{b}}^\prime. \label{eq:welded_edge_remains}
    \end{align}
    Notice that by the proof of \cref{prop:non_zero_support_good_map}, we still keep $\rho$ in the error, meaning that $\prWeld (V \zeta_{\mathfrak{b}} - \zeta_{\mathfrak{b}}^\prime)=V \zeta_{\mathfrak{b}} - \zeta_{\mathfrak{b}}^\prime$. By \cref{eq:welded_edge_remains},
    \begin{align*}
        \prWeld V\zeta_{\mathfrak{b}} = \prWeld (V \zeta_{\mathfrak{b}} - \zeta_{\mathfrak{b}}^\prime) + \prWeld \zeta_{\mathfrak{b}}^\prime = V\zeta_{\mathfrak{b}}.
    \end{align*}
\end{proof}

\subsubsection{Bounding escape probability}

We may apply the analysis above to get the following result.

\begin{proposition}\label{prop:combining_primitive_bounds}
    For any $j\in [0,T]$ and any primitive $V$,
    \begin{align*}
        \norm{(\Id - \prUnif{j+1}) (\prWeld V \prWeld) \prUnif{j}} \leq 8\sqrt{\delta}
    \end{align*}
\end{proposition}
\begin{proof}
    The case when $V$ is an algorithm unitary or XOR operator follows from \cref{prop:unif_proj_algo_xor}. When $j=0$, $\prUnif{0} = 0$, so this case trivially holds.
    
    Let us assume that $V\in \{\pCLabel,\pCLabelInv,\mpCPerm\}$ and $j\geq 1$. By \cref{eq:kronecker_database_instances}, the states $\canVec{D}$ form a basis, meaning that any vector $\zeta \in \ran(\prUnif{j})$ can be written as,
    \begin{align*}
        \zeta = \sum_a \sum_{D \in \reprTuple: \abs{D} \leq j} c_{a,D} \ket{a}\canVec{D}.
    \end{align*}
    Furthermore, we can partition the basis $(a,D)$ based on which assignment group they belong in, queries of $\rho$ or identity instances. These are all the options a query can be. Call such group $\mathfrak{b}$ and for each group, let $\zeta_{\mathfrak{b}}$ be subspace of $\zeta$ conditioned on $\mathfrak{b}$. Every $D^\prime\in\mathfrak{b}$ has $\abs{D^\prime}\leq j+1$ and if $\abs{D^\prime}=j+1$, their amplitude is $0$ due to \cref{lem:cpwt_bounded_growth}. Let $\zeta_{\mathfrak{b}}^\prime$ be the vector from \cref{lem:pc_types_i_iii_estimate,lem:pc_type_iv_similarity}. Then,
    \begin{align}
        \prUnif{j+1} \zeta_{\mathfrak{b}}^\prime = \zeta_{\mathfrak{b}}^\prime.\label{eq:prunif_j_group}
    \end{align}
    By \cref{lem:pc_types_i_iii_estimate,lem:pc_type_iv_similarity}, unless we are querying the weld $\rho$,
    \begin{align*}
        \norm{V\zeta_{\mathfrak{b}} - \zeta_{\mathfrak{b}}^\prime}^2 \leq 32 \delta \norm{\zeta_{\mathfrak{b}}}^2.
    \end{align*}
    By \cref{prop:proj_weld_on_group_states}, inserting $\prWeld$ on both sides of $V$ does not alter these errors. Define $\zeta^\prime$ by replacing each $\zeta_{\mathfrak{b}}$ with $\zeta_{\mathfrak{b}}^\prime$ except for the weld query. The error in the weld is orthogonal to other errors due to fixing the basis $(a,L)$. Letting $\zeta_\rho$ be the weld component, \cref{lem:error_orthogonality,lem:mpcPerm_weld_similarity} give
    \begin{align*}
        \norm{(\prWeld V \prWeld) \zeta - \zeta^\prime}^2 &= \sum_{\mathfrak{b}}\norm{V\zeta_{\mathfrak{b}}-\zeta_{\mathfrak{b}}^\prime}^2+\norm{(\prWeld V\prWeld-\Id)\zeta_\rho}^2\\
        &\leq32\delta\sum_{\mathfrak{b}}\norm{\zeta_{\mathfrak{b}}}^2+\frac{64}{N}\norm{\zeta_\rho}^2\leq64\delta\norm{\zeta}^2,
    \end{align*}
    where the input components are orthogonal by \cref{eq:kronecker_database_instances} and $64/N\leq\delta$. Therefore, by \cref{eq:prunif_j_group},
    \begin{align*}
        \norm{(\Id - \prUnif{j+1}) (\prWeld V \prWeld) \zeta}^2 &= \norm{(\Id - \prUnif{j+1}) ((\prWeld V \prWeld) \zeta - \zeta^\prime )}^2\\
        &\leq \norm{(\prWeld V \prWeld) \zeta - \zeta^\prime}^2 \leq 64\delta\norm{\zeta}^2.
    \end{align*}
    Hence by taking the square root, the statement follows.
\end{proof}

Let us prove the main theorem of the subsection.

\begin{theorem}\label{thm:bound_escape_probability}
    For every weld record $\rho = (\alpha,e,f)$, direction $\sigma \in \{\rightarrow,\leftarrow\}$ and step $t\in [T]$,
    \begin{align*}
        \norm{\prBoundary_{\rho, \sigma} \prFuture \insertDir \ket{\psi_{t-1}}}^2 \leq \frac{2^{14} (T+1)^4}{\sqrt{N}} \norm{\insertDir \ket{\psi_{t-1}}}^2
    \end{align*}
\end{theorem}
\begin{proof}
    Let $\xi = \insertDir \ket{\psi_{t-1}}$ and assume that $\xi \neq 0$ as otherwise the statement is trivially true. Notice that $\prUnif{t} \xi = \xi$ as by the freshness condition $r$ is isolated and thus $H=\{\emptySeq\}$. We define the following series of vectors for $j\in [t+1,T]$,
    \begin{align*}
        &\vecRem_t \coloneq \xi, & \vecRem_j \coloneq (\prWeld V_j \prWeld) \vecRem_{j-1},\\
        &\vecUnif_t \coloneq \xi, & \vecUnif_j \coloneq \prUnif{j}(\prWeld V_j \prWeld) \vecUnif_{j-1}.
    \end{align*}
    Notice that,
    \begin{align*}
        \vecRem_j - \vecUnif_j = (\prWeld V_j \prWeld)(\vecRem_{j-1} - \vecUnif_{j-1}) + (\Id - \prUnif{j})((\prWeld V_j \prWeld))\vecUnif_{j-1}.
    \end{align*}
    Therefore, by \cref{prop:combining_primitive_bounds} and the triangle inequality,
    \begin{align}
        \norm{\vecRem_T - \vecUnif_T} &\leq \norm{\vecRem_t - \vecUnif_t} + \sum_{j\in [t,T-1]} 8\sqrt{\delta} \norm{\vecUnif_j}\\
        &\leq 8(T-t)\sqrt{\delta} \norm{\xi}.\label{eq:bound_on_diff_vecs}
    \end{align}
    Notice that $\prBoundary_{\rho,\sigma} \vecUnif_T = 0$.
    Therefore,
    \begin{align*}
        \norm{\prBoundary_{\rho,\sigma} \prFuture \xi}^2 &= \norm{\prBoundary_{\rho,\sigma} (\vecRem_T - \vecUnif_T)}^2\\
        &\leq \norm{\vecRem_T - \vecUnif_T}^2 \leq \frac{2^{14}(T+1)^4}{\sqrt{N}} \norm{\xi}^2,
    \end{align*}
    where the last line follows from \cref{eq:bound_on_diff_vecs,lem:goodness_probability}.
\end{proof}

\subsection{Single-cycle distribution}\label{sec:single_cycle}

Notice that the distribution over welded tree graphs using $\colorPerm \sim \colorPresSet$ can cause the welded edges to not form a single cycle, but instead multiple cycles. Let $\multDistr$ be the corresponding distribution over welded tree graphs, \evSingle be the event that the graph contains a single cycle in the weld and $\singDistr \coloneq \multDistr|_{\evSingle}$.

\begin{lemma}\label{lem:mult_to_single_cycle}
    Let $p_{\calD}(\alg)$ be the acceptance probability of a quantum algorithm \alg making $q\geq 1$ queries to an oracle sampled from $\calD$. Then,
    \begin{align*}
        p_{\singDistr}(\alg) = O\left( p_{\multDistr}(\alg) + \frac{n\cdot q^2}{N}  \right).
    \end{align*}
\end{lemma}
\begin{proof}
    For each color $\eta\in \colorSet$, define the vertices at the ends of welded edges colored $\eta$ as,
    \begin{align*}
        &L_\eta \coloneq \col(n) \setminus \colClass{n,\eta}, &R_\eta \coloneq \col(n+1) \setminus \colClass{n+1,\eta}\\
        m_\eta&\coloneq\abs{L_\eta}=\abs{R_\eta}
        =N-\colClassSize{n,\eta}=\Theta(N),
    \end{align*}
    where the estimate on the size of $m_\eta$ is due to \cref{lem:colors_per_column}. A color-$\eta$ matching is some bijection $s:L_\eta\to R_\eta$ which represents the edge set $\{ \{u,s(u)\}: u\in L_\eta\}$.

    Let $h=(h_\eta)_{\eta\in \colorSet}$ be the weld edges of a graph sampled from \singDistr and $B=(B_\eta)_{\eta \in \colorSet}$ some set of recorded values which have been read, where $B_\eta \subseteq \{(u,h_\eta(u)):u\in L_\eta\}$. Conditioning on the record $B$ means that they are present. Let $U_\eta \coloneq L_\eta \setminus \dom(B_\eta)$, $V_\eta \coloneq R_\eta \setminus \im(B_\eta)$ and $r_\eta \coloneq \abs{U_\eta} = \abs{V_\eta} = m_\eta - \abs{B_\eta}$. Notice that for $u\in U_\eta, v\in V_\eta$ and $r_\eta>1$,
    \begin{align}
        p_B^\eta(u,v) \coloneq \Pr[h_\eta(u) = v | B] \leq \frac{1}{r_\eta-1}.\label{eq:single_cycle_sequential_marginals}
    \end{align}
    To see this, first condition on complete assignments of color edges for the other colors. These fixed edges along with $B_\eta$ form paths with $r_\eta$ free nodes on each side. Notice that both $u$ and $v$ must be endpoints of their respective paths connecting them. First, suppose that for one of $u$ or $v$, the other endpoint on this path is on the opposite side. Deleting this path from a completed cycle and joining its connecting edges allows it to be reinserted exactly in one place for each of the $r_\eta - 1$ possible partners. Hence $\{u,v\}$ has probability $\tfrac{1}{r_\eta - 1}$, unless it joins the paths endpoints in which case it is $0$.
    
    Otherwise, $u$ is part of a path whose both endpoints are on the left. Each right endpoint on a path with opposite-side endpoints has probability $\tfrac{1}{r_\eta-1}$ by the argument above, while the rest have equal probability due to symmetry. Therefore, their common probability must be less than $\tfrac{1}{r_\eta - 1}$ as otherwise the probabilities would sum over $1$. Hence \cref{eq:single_cycle_sequential_marginals} follows by averaging. For $r_\eta\geq 1$, let,
    \begin{align*}
        t_B^{\eta}(u,v) \coloneq 1 - (r_\eta-1)p_B^\eta(u,v).
    \end{align*}
    Notice that by \cref{eq:single_cycle_sequential_marginals}, $t_B^{\eta}(u,v)$ is nonnegative. Furthermore,
    \begin{align*}
        \sum_{v\in V_\eta}t_B^{\eta}(u,v) = \sum_{u\in U_\eta} t_B^{\eta}(u,v) = r_\eta-(r_\eta-1) = 1.
    \end{align*}
    We create the assignment $B$ using the following process. For every color $\eta$. start with $B_\eta =\emptyset$ and let $\phi_\eta:R_\eta \to R_\eta$ be the identity. One color at a time in some order, repeat the following process until $U_\eta$ is empty.
    \begin{enumerate}
        \item Uniformly sample $u~\sim U_\eta$ and let $y=h_\eta(u)$.
        \item With probability $1-\tfrac{1}{r_\eta}$, set $z=y$. Otherwise, sample $z\sim V_\eta$ with probability $t_B^{\eta}(u,\cdot)$.
        \item Add $(u,y)$ to $B_\eta$ and exchange the mappings of $\phi_\eta(y)$ and $\phi_\eta(z)$.
    \end{enumerate}
    Afterwards, let $g = (g_\eta)_{\eta \in \colorSet}$ where $g_\eta = \phi_\eta \circ h_\eta$. Notice that for each fixed $u$,
    \begin{align*}
        \Pr[z=v|B,u] = \left(1 - \frac{1}{r_\eta}\right) p_B^\eta(u,v) + \frac{1}{r_\eta} t_B^{\eta}(u,v) = \frac{1}{r_\eta},
    \end{align*}
    where $B$ is the record before the current $u$ was handled.
    Notice that after swapping $\phi_\eta(y)$ and $\phi_{\eta}(z)$, we have $\phi_\eta(h_\eta(u))$ is never changed again. Therefore each value of $g_\eta$ is uniform over the right vertices which were not used beforehand. Therefore all $g_\eta$ are independent and uniform. Note that this extends when conditioned on previous random choices as they depend on $h$ through $B$. Therefore, $g$ follows the distribution of welded edges under $\multDistr$ and hence that from \cref{sec:construction_welded_trees}.

    Fix some $h$ and some records $B$. On the second case in step 2, $y=h_\eta(u)$ is uniform over $V_\eta$ as it is bijective and $z$ is uniform as,
    \begin{align*}
        \frac{1}{r_\eta} \sum_{u\in U_\eta} t_B^{\eta}(u,v) = \frac{1}{r_\eta}.
    \end{align*}
    The fixed value of $\phi_\eta(v)$ only changes in this branch when $v\in \{y,z\}$. Hence,
    \begin{align*}
        \Pr[\phi_\eta(v) \text{ changes}|h,B] \leq \frac{1}{r_\eta} \left( \frac{1}{r_\eta} + \frac{1}{r_\eta}\right) = \frac{2}{r_\eta^2}.
    \end{align*}
    For a fixed $x\in L_\eta$, we have that $h_\eta(x)$ is in $V_\eta$ when $x$ has not been read yet. This has probability $\tfrac{r_\eta}{m_\eta}$ as the order is uniformly random and independent of $h$. By union bounding,
    \begin{align}
        \Pr[g_\eta(x) \neq h_\eta(x)|h] \leq \sum_{r\in [2,m_\eta]} \frac{r}{m_\eta} \frac{2}{r^2} = O\left(\frac{n}{N}\right).\label{eq:single_cycle_sequential_discrepancy}
    \end{align}
    The inverse mappings have the same bound due to symmetry. Sample some permutations on other edges besides the weld and vertex labels independently of $(g,h)$ and use them to obtain the instances of welded tree graph $\sigma\sim \multDistr$ and $\sigma^\prime \sim \singDistr$ built from $g$ and $h$ respectively. By \cref{eq:single_cycle_sequential_discrepancy},
    \begin{align*}
        \Pr[\wtF_\sigma(l,\eta) \neq \wtF_{\sigma^\prime}(l,\eta) | \sigma^\prime ] = O\left(\frac{n}{N}\right).
    \end{align*}
    Applying the hybrid argument of~\cite{BBBV97}, taking the expectation over $(\sigma,\sigma^\prime)$ and applying Cauchy-Schwarz,
    \begin{align*}
        \EX \norm{\ket{\psi_{\alg,q}^{\wtO_\sigma}} - \ket{\psi_{\alg,q}^{\wtO_{\sigma^\prime}}}}^2 &\leq q\sum_{j\in [0,q-1]} \EX\norm{(\wtO_\sigma - \wtO_{\sigma^\prime}) \ket{\psi_{\alg,j}^{\wtO_{\sigma^\prime}}} }^2\\
        &= O\left( \frac{n\cdot q^2 }{N}\right)
    \end{align*}
    Therefore by the triangle inequality and projecting on the accepting output,
    \begin{align*}
        p_{\singDistr}(\alg) &= O\left( p_{\multDistr}(\alg) + \frac{n\cdot q^2}{N}  \right).\qedhere
    \end{align*}
\end{proof}

\subsection{Putting it all together}

With all of the tools above, let us prove the main result of the paper.

\begin{proof}[Proof of \cref{thm:pathfinding_hard}]
    Let $\Pi_\calR$ be the projector on $\I$ which contains databases which exit, as defined in \cref{eq:databases_which_exit}.
    Due to \cref{thm:success_probability_standard_compressed}, it suffices to bound $\probCWT$, the probability $\ket*{\psi_{\alg,q}^{(\cpwt)}}_{\A\I}$ contains a database which records an exit. Furthermore, using \cref{lem:cpwt_mcpgt}, we may switch from \cpwt to \mcpwt. Let us use $\ket{\psi}_{\A\I} = \ket*{\psi_{\alg,q}^{(\mcpwt)}}_{\A\I}$ to denote the final state using \mcpwt. This means that it suffices to bound $p_{\mcpwt} = \norm*{\Pi_{\calR} \ket{\psi} }^2$. This is exactly done using \cref{lem:bounding_future_bounds_exit_probability,thm:bound_escape_probability}. Lastly, we apply \cref{lem:mult_to_single_cycle} to ensure that we are analyzing the distribution over graphs with a single cycle in the welded edges. In order to apply \cref{lem:mult_to_single_cycle}, we use the verification procedure from \cref{thm:success_probability_standard_compressed} in order to make the algorithm have a binary output and may assume that they only use $O(q)$ queries during the verification procedure, including checking if the path exits.
    Putting it all together,
    \begin{align*}
        \Pr[\text{$\alg$ succeeds}] &=  O\left(\probCWT + \frac{q^6}{N^{1/4}} \right) &\mbox{(By \cref{thm:success_probability_standard_compressed})}\\
        &= O\left(p_{\mcpwt} + \frac{q^3}{\sqrt{N}} + \frac{q^6}{N^{1/4}} \right) &\mbox{(By \cref{lem:cpwt_mcpgt})}\\
        &= O\left(\frac{q^6}{\sqrt{N}} + \frac{q^6}{N^{1/4}} \right) &\mbox{(By \cref{lem:bounding_future_bounds_exit_probability,thm:bound_escape_probability})}\\
        &= O\left(\frac{q^6}{N^{1/4}} + \frac{n \cdot q^2}{N} \right) &\mbox{(By \cref{lem:mult_to_single_cycle})}\\
        &= O\left(\frac{q^6}{N^{1/4}} \right).&\qedhere
    \end{align*}
\end{proof}

\section{Quantum walk analysis}\label{sec:quantum_walk_analysis}

In order to build intuition for the compressed oracle model, let us describe what information is (and is not) tracked in the database when running a discrete quantum walk. Suppose that we have access to an XOR oracle $O_G$ representing a graph $G=(V_G, E_G)$, which acts as follows,
\begin{align*}
    O_G \ket{v,\alpha, v^\prime} = \ket{v, \alpha, v^\prime \oplus E_G(v,\alpha)},
\end{align*}
where $E_G(v,\alpha)$ returns the neighbor of $v$ on color $\alpha$ if it exists, and returns $v$ otherwise. A discrete quantum walk repeatedly applies two operators, a coin operator $\gcoin$ and a shift operator $\shiftOp$. The coin operator is a Grover coin over the 3 colors in $\colorSet$,
\begin{align*}
    \ket{s_{3}} &\coloneq \frac{1}{\sqrt{3}} \sum_{\alpha\in \colorSet} \ket{\alpha},\\
    \gcoin &\coloneq 2 \ketbra{s_3} - \Id.
\end{align*}
The shift operator $\shiftOp$ behaves as an in-place oracle, $\shiftOp \ket{v,\alpha} = \ket{E_G(v,\alpha), \alpha}$. It can be simulated using 2 calls to $O_G$ and a $\SWAP$ operator,
\begin{align*}
    \ket{v,\alpha, 0} \xrightarrow{O_G} \ket{v,\alpha, E_G(v,\alpha)} \xrightarrow{\SWAP} \ket{E_G(v,\alpha), \alpha, v} \xrightarrow{O_G} \ket{E_G(v,\alpha), \alpha, 0}.
\end{align*}
Based on the continuous quantum walk result of~\cite{CFG02,CCD+03}, it was shown in~\cite{LLL23} that this walk traverses any welded tree in $\poly(n)$ applications of $\shiftOp \cdot \gcoin$. The main idea behind the speedup is that a quantum walk only remembers its current position, which allows for interference to make the walk behave like a walk on a line where each node is a column in $G$.

Intuitively, when replacing $O_G$ with \cpwt, one might expect that the database reflects this idea by solely remembering the current label (and nothing else). We show that this is (approximately) the case.

\paragraph{Directed half-edges.} Let us use $\nu \in \{+, 0, -\}$ to denote direction and $c\in [0,2n+1]$ denote the column.\footnote{We do not use $\rightarrow, \leftarrow$ as in \cref{sec:compressed_welded_tree} as they denote if we are going up/down the tree, here we are adding the columns.} Define,
\begin{align*}
    H_c^{\nu} \coloneq \{(v,\alpha): v\in \col(c), \col(\ncF(v,\alpha)) = c + (\nu \cdot 1) \}.
\end{align*}
Note that $H_c^{0}$ is empty except for the roots in the direction without an edge. The column-direction state is defined as,
\begin{align*}
    \ket{c, \nu} &\coloneq \begin{cases}
        \frac{1}{\sqrt{\abs{H_c^\nu} M} } \sum_{(v,\alpha)\in H_c^\nu} \sum_{l\in [M]} \ket{l,\alpha, 0}_{\X\Y} \ket{[v\to l]}_{\Lab}\ket{\bot}_{\C}, &\text{if $c\neq 0$,}\\
        \frac{1}{\sqrt{\abs{H_{0}^{\nu}}}} \sum_{(0,\alpha)\in H_{0}^{\nu}}\ket{0,\alpha,0}_{\X\Y}\ket{\bot}_{\I} &\text{if $c=0$.}
    \end{cases}
\end{align*}
Furthermore, let $\calR = \Span\{\ket{c, \nu}: \abs{H_c^\nu} > 0, \nu \in \{+, 0, -\}  \}$ denote the subspace of the column-direction states. Note that for any $\ket{\psi}\in \calR$, $\gcoin\ket{\psi}\in \calR$.

Letting $\shiftOp = \cpwt \cdot  \SWAP \cdot \cpwt$ be the shift operator, we define a slightly weaker version of the operator $\shiftComp$ as,
\begin{align*}
    \shiftComp &\coloneq \pCPerm\cdot \pCLabel \cdot \xorOP \cdot \SWAP \cdot \xorOP \cdot \pCLabel \cdot \pCPerm
\end{align*}
Notice that $\shiftComp$ differs from $\shiftOp$ in that we omit all applications of $\pCLabelInv$ and all the compression operators that occur between the two applications of the XOR operator $\xorOP$ in \cpwt. We find that under $\shiftComp$, the walk remains in $\calR$.
\begin{theorem}\label{thm:shift_remains_in_columns}
    For any $c\in [0,2n+1]$ and $\nu \in \{+,0,-\}$ such that $\abs{H_c^{\nu}} > 0$,
    \begin{align*}
        \shiftComp \ket{c,\nu} = \ket{c+ \nu, -\nu}.
    \end{align*}
\end{theorem}
\begin{proof}
    For simplicity we suppress the register subscripts. Fix some $b\in \middleEdge$ with color $\alpha$. Notice that for distinct nonzero vertices $v,w$, $l^\prime\in [M]$ and $f\in E_b$,
    \begin{align}
        \pC_{v}^{\mathrm{for}}\left( \frac{1}{\sqrt{M-1}} \sum_{l\in [M]\setminus \{l^\prime\}} \ket{[v\to l, w\to l^\prime]} \right) &= \ket{[w\to l^\prime]} \label{eq:walk_shift_label_erasure}\\
        \pC_{f,b}^{\leftarrow}\left( \frac{1}{\sqrt{\abs{E_b}}}\sum_{e\in E_b}\ket{[e\to f]} \right) &= \ket{\bot}. \label{eq:walk_shift_color_erasure}
    \end{align}
    Suppose that we start from the uniform state over the tails of $E_b$. By \cref{eq:walk_shift_label_erasure,eq:walk_shift_color_erasure},
    \begin{align}
        &\frac{1}{\sqrt{\abs{E_b}M}} \sum_{\substack{e\in E_b\\ l\in [M]}} \ket{l,\alpha,0}\ket{\tail(e) \to l}\ket{\bot} \nonumber\\
        &\xrightarrow{\xorOP \cdot \SWAP \cdot \xorOP \cdot \pCLabel \cdot \pCPerm} \frac{1}{\abs{E_b}\sqrt{M (M-1)}} \sum_{\substack{e,f\in E_b \\ l,l^\prime\in [M], l\neq l^\prime}} \ket{l^\prime, \alpha, 0} \ket{[\tail(e) \to l, \head(f) \to l^\prime], [e\to f]} \nonumber\\
        &\xrightarrow{\pCLabel} \frac{1}{\abs{E_b}\sqrt{M}} \sum_{\substack{e,f\in E_b \\ l^\prime \in [M]}} \ket{l^\prime, \alpha, 0}\ket{[\head(f)\to l^\prime], [e\to f]} \nonumber\\
        &\xrightarrow{\pCPerm} \frac{1}{\sqrt{\abs{E_b}M}} \sum_{\substack{f\in E_b, l^\prime\in [M]}}\ket{l^\prime, \alpha, 0} \ket{[\head(f)\to l^\prime], \bot}.\label{eq:applying_s_dot}
    \end{align}
    The reverse direction follows by applying $\shiftComp$ as it is self-inverse. The case when $\edgeIndic(e)=\bot$ follows analogously, except we ignore $\pCPerm$. Notice that $\abs{H_{c}^\nu} = \abs{H_{c+\nu}^{-\nu}}$. Therefore by \cref{eq:applying_s_dot} when $c\neq 0$ and $c+\nu\neq0$,
    \begin{align*}
        \shiftComp\ket{c,\nu} = \frac{1}{\sqrt{\abs{H_c^\nu} M}} \sum_{\substack{(w,\alpha) \in H_{c + \nu}^{-\nu} \\ l^\prime \in [M]}} \ket{l^\prime, \alpha, 0}\ket{[w\to l^\prime], \bot} = \ket{c + \nu, -\nu}.
    \end{align*}
    Lastly, the same process shows that $\shiftComp\ket{0,+} = \ket{1,-}$, $\shiftComp\ket{1,-} = \ket{0,+}$, and $\shiftComp\ket{c,0} = \ket{c,0}$ at the roots of the trees.
\end{proof}

Furthermore, as long as we stay in $\calR$, the shift operators do not drastically differ.

\begin{lemma}\label{lem:difference_shift_operators}
    For every $\ket{\psi}\in \calR$,
    \begin{align*}
        \norm{(\shiftOp - \shiftComp) \ket{\psi}} = O(2^{-n/4}) \norm{\ket{\psi}}.
    \end{align*}
\end{lemma}
\begin{proof}
    Let $\xi = \xorOP \cdot \pCLabel \cdot \pCPerm \ket{\psi}$ and $\zeta \in \{\ket{\psi}, \xi, \SWAP\xi, \shiftComp\ket{\psi}\}$. Next, decompose $\zeta$ into orthogonal blocks with fixed queried label $l\neq0$ and base database $L$, where $\abs{L}\leq1$ as $\ket{\psi}\in \calR$. Within each block,
    \begin{align*}
        \norm{(\pC_{l}^{\mathrm{inv}} - \Id) \sum_{v} \ket{\zeta_v}\ket{L[v\to l]}}^2 &= \frac{2}{M - \abs{L}}\norm{\sum_v \ket{\zeta_v}}^2\\
        &\leq \frac{2\abs{V_n}}{M-1} \sum_v \norm{\ket{\zeta_v}}^2.
    \end{align*}
    Summing over the orthogonal blocks gives $\norm{(\pCLabelInv-\Id)\zeta}^2\leq \frac{2\abs{V_n}}{M-1}\norm{\ket{\psi}}^2$, while the $l=0$ blocks are fixed. Similarly, for $\zeta\in \{\xi, \SWAP\xi\}$, the retained labels determine the queried assignment, so by \cref{prop:recompression_unlikely,obs:min_num_nodes_per_perm},
    \begin{align*}
        \norm{(\pCLabel - \Id) \zeta}^2 &\leq \frac{2}{M-1}\norm{\psi}^2\\
        \norm{(\pCPerm - \Id) \zeta}^2 &\leq \frac{8}{\sqrt{N}}\norm{\psi}^2.
    \end{align*}
    By applying the bounds above and telescoping,
    \begin{align*}
        \norm{(\shiftOp - \shiftComp) \ket{\psi}} &\leq \norm{(\pCLabelInv - \Id)\ket{\psi}} + \norm{(\pCLabelInv - \Id) \shiftComp\ket{\psi}}\\
        &\quad+\sum_{\zeta\in\{\xi,\SWAP\xi\}}\sum_{V\in \{\pCLabelInv, \pCLabel, \pCPerm\}} \norm{(V-\Id)\zeta}\\
        &\leq O(2^{-n/4})\norm{\ket{\psi}} \qedhere
    \end{align*}
\end{proof}

At the start, the state is in $\calR$. Hence, by a repeated application of \cref{lem:difference_shift_operators}, after $q$ walk steps its distance from an ideal state in $\calR$ is $O(q2^{-n/4})$. Thus, for polynomially many steps, it lies in $\calR$ with overwhelming probability and the database records only the current position.

\bibliographystyle{alpha}
\bibliography{main}

\appendix
\crefalias{section}{appendix}

\section{Proof of general intertwiner lemma}\label{sec:general_intertwiner_lemma}

This section is concerned with proving \cref{lem:general_intertwiner_lemma}. As the main ideas rely on the work of~\cite{Car25}, let us state what we will use.

We will use $D$ to represent an arbitrary domain and add the subscript $D$ to the operators defined in \cref{sec:permutation_oracles_definition}. For a function $g:D\to \{0,1\}^l$, define the XOR operation with $g$ as,
\begin{align}
    \xorOP_{D,x}^{g}\ket{z}_{\Y}\ket{I}_{\I_D} &\coloneq \begin{cases}
        \ket{z\oplus g(I(x))}_{\Y}\ket{I}_{\I_D},&x\in\dom(I),\\
        \ket{z}_{\Y}\ket{I}_{\I_D},&x\notin\dom(I),
    \end{cases}\label{eq:appendix_relabeled_queries}
\end{align}
Any corresponding operator with $g$ uses the XOR operator we just defined. The corresponding standard query with $g$ is,
\begin{align*}
    O_\phi^g\ket{b,x,z} \coloneq\ket{b,x,z\oplus g(\phi^{1-2b}(x))}.
\end{align*}

\begin{lemma}[Consequence of Section 5 in~\cite{Car25}]\label{lem:power_four_useful_facts}
    Let $D$ be a set with $\abs{D}=4^k$ for some $k\in \mathbb{N}_0$. There exists a constant $K$, a register $\regP_D$, an isometry $\calV_D: \calH(\I_D) \to \calH(\regP_D)$ where $\calV_D^\dagger \calV_D = \Id$ and $\ket{\bot}_{\regP_D} = \calV_D\ket{\bot}_{\I_D}$, and operators $\cmfO_{D,x}^{g,b}$ on $\calH(\Y) \otimes \calH(\regP_D)$ such that the following holds.

    For some algorithm $\alg$ and maps $g_1,\dots g_T: D\to \{0,1\}^l$, let,
    \begin{align*}
        \cmfO_D^g &\coloneq \sum_{b\in \{0,1\},x\in D} \ketbra{b,x}_{\X} \otimes \cmfO_{D,x}^{g,b},\\
        \ket{\Psi_T} &\coloneq A_T \cmfO_{D}^{g_T} A_{T-1} \dots A_1 \cmfO_{D}^{g_1} A_0 \ket{0}_{\A} \ket{\bot}_{\regP_D},\\
        \ket{\psi_{\sigma, T}} &\coloneq A_T O_\sigma^{g_T} A_{T-1} \dots A_1 O_{\sigma}^{g_1} A_0 \ket{0}_{\A}. 
    \end{align*}
    Then,
    \begin{align}
        \Tr_{\regP_D}\rho(\ket{\Psi_T}) = \EX_{\sigma\sim \Sym(D)} \rho(\ket{\psi_{\sigma, T}}).\label{eq:exact_purification}
    \end{align}
    Moreover, for any $d\in \mathbb{N}_0, x\in D, b\in \{0,1\}$ and $g:D\to \{0,1\}^l$,
    \begin{align}
        \norm{\left( \cmfO_{D,x}^{g,b} \calV_D - \calV_D \cpO_{D,x}^{g,b} \right)\Pi_{\leq d, \I_D}} \leq \min \left\{ 2, K\frac{d+3}{\abs{D}^{1/4}}\right\}.\label{eq:isometry_between_purification_compression}
    \end{align}
    When $\abs{D}=1$, \cref{eq:isometry_between_purification_compression} is exact.
\end{lemma}
\begin{proof}
    The construction of \cref{eq:exact_purification} is the uniform-twirl construction used in the proof of \cite[Theorem 5.19]{Car25}. The map $\calV_D$ sends $\ket{I}$, for $\abs{I}\leq\sqrt{\abs{D}}$, to the corresponding sophisticated state of~\cite[Section 5.1]{Car25}; for larger databases, extend it isometrically into an enlarged purification register, where the bound follows from the trivial estimate $2$.
    To prove \cref{eq:isometry_between_purification_compression}, split a sophisticated state into its flip-elegant component and its orthogonal complement. Lemma 5.14 of~\cite{Car25} gives the required comparison on the first component. The second component is spanned by the uniform-assignment states $\calV_D\ket{+_{D,x,I}}$ (with the database inverted when $b=1$). The compressed query acts as the identity on the corresponding states in $\I_D$: decompression removes the queried assignment, so the XOR has no effect. If $\Pi_{\mathrm{qv}}$ denotes Carolan's query-valid projector, Lemma 5.4 of~\cite{Car25} bounds $\norm{\Pi_{\mathrm{qv}}\ket{\psi}}$ on this second component by $C(d+3)\abs{D}^{-1/4}\norm{\ket{\psi}}$ for an absolute constant $C$. Since the purified query acts as the identity outside the query-valid subspace,
    \begin{align*}
        \norm{(\cmfO_{D,x}^{g,b}-\Id)\ket{\psi}} \leq 2\norm{\Pi_{\mathrm{qv}}\ket{\psi}}.
    \end{align*}
    Combining the two components gives the stated bound after increasing $K$. For $\abs{D}=1$, take $\calV_D=\Id$ and $\cmfO_D^g=\cpO_D^g$, which gives exact queries from the empty database.
    The same proof applies to arbitrary $g$, since XORing $g$ of the returned value changes only the answer register and the decompression estimates do not depend on $g$.
\end{proof}

Next, as an intermediate step in the proof of \cref{lem:general_intertwiner_lemma}, we construct a general version of $\calV_D$ in \cref{lem:power_four_useful_facts} for domains of arbitrary size. The proof is done by partitioning the set into sets whose size is a power of four and for any query, randomly choosing a set. This is similar to the twirling permutations in~\cite{Car25}.

\begin{lemma}\label{lem:appendix_local_permutation_comparison}
    Let $S$ be an arbitrary finite nonempty set. There exists a constant $K^\prime$, a register $\regP_S$, an isometry $\calV_S: \calH(\I_S) \to \calH(\regP_S)$, and an oracle $\cmfO_S$ such that for every $d\in \mathbb{N}_0$ such that $d\leq \abs{S}/2$,
    \begin{align}
        \norm{(\cmfO_S \calV_S - \calV_S \cpO_S) \Pi_{\leq d, \I_S}} \leq K^\prime\frac{d+3}{\abs{S}^{1/4}}
    \end{align}
\end{lemma}
\begin{proof}
    Let $s\coloneq\abs{S}$. When $s\leq 3$, the original argument gives the bound $2$ after increasing $K^\prime$. Therefore, assume $s\geq 4$ and let $s=\sum_{k\in [0,k^*]} a_k4^k$ be the base-four expansion of $s$ and $B=\sum_{k\in [0,k^*]} a_k$. Fix some partition of $S$ such that,
    \begin{align}
        S = \bigsqcup_{j\in [B]} S_j, \label{eq:partition_of_permutation_set}
    \end{align}
    with $a_k$ blocks of size $4^k$, and let $s_j \coloneq \abs{S_j}$. Notice that $B \leq 3(1+\floor{\log_4 s})$ and $\sum_j \sqrt{s_j} \leq 6\sqrt{s}$. We call this the \emph{block} partition.
    For $I\in \db_S$ and $\pi, \omega \in \Sym(S)$, define,
    \begin{align*}
        J_j(I,\pi,\omega) &\coloneq \{(\pi(u),\omega^{-1}(v)):(u,v)\in I\} \cap (S_j \times S_j),\\
        \dbJ(I,\pi,\omega) &\coloneq (J_1(I,\pi,\omega),\dots J_B(I,\pi,\omega)),\\
        \mathrm{Comp}(I) &\coloneq \{ (\pi, \omega)\in \Sym(S)^2: \sum_{j\in [B]} \abs{J_j(I,\pi,\omega)} = \abs{I} \}.
    \end{align*}
    Intuitively, $\mathrm{Comp}(I)$ denotes the masks that are compatible with some database.
    The block register and its basis are,
    \begin{align*}
        \calH(\I_{\mathrm{blocks}}) &\coloneq \calH(\Sym(S)^2)\otimes\bigotimes_{j=1}^{B}\calH(\I_{S_j}),\\
        \ket{\dbJ} &\coloneq \ket{J_1}\otimes\cdots\otimes\ket{J_B},\\
        \Pi_{\leq d, \I_{\blocks}} &\coloneq \sum_{\pi,\omega \in \Sym(S)} \ketbra{\pi,\omega} \otimes \sum_{\dbJ: \abs{\dbJ} \leq d} \ketbra{\dbJ},
    \end{align*}
    where $\abs{\dbJ} = \sum_j\abs{J_j}$. Let $t_j(I,\pi) \coloneq \abs{\pi(\dom(I)) \cap S_j}$ and for integers $m,t\geq 0$, $\decrFact{m}{t} = \prod_{h\in [0,t-1]} (m-h)$. Then we define the partition map as,
    \begin{align}
        w_I(\pi, \omega) &\coloneq \begin{cases}
            \frac{1}{s!} \sqrt{\frac{\decrFact{s}{\abs{I}}}{\prod \decrFact{s_j}{t_j(I,\pi)}}}&\text{if $(\pi,\omega) \in \mathrm{Comp}(I)$,}\\
            0 &\text{otherwise,}
        \end{cases}\\
        \intPart\ket{I} &\coloneq \sum_{\pi,\omega \in \Sym(S)} w_I(\pi, \omega)\ket{\pi,\omega} \ket{\dbJ(I,\pi,\omega)}.
    \end{align}
    On compatible masks, we have that $\abs{J_j(I,\pi,\omega)} = t_j(I,\pi)$ and $\sum_{\omega} w_I(\pi,\omega)^2 = \tfrac{1}{s!}$. Additionally, the original database can be recovered by,
    \begin{align}
        I &= \{(\pi^{-1}(u), \omega(v)): (u,v) \in \bigsqcup_j J_j(I,\pi,\omega) \}.\label{eq:appendix_database_recovery}
    \end{align}
    Disjoint databases $I$ therefore have disjoint supports under $\intPart$. Therefore,
    \begin{align}
        \intPart^\dagger \intPart = \Id,\nonumber\\
        \Pi_{\leq d, \I_{\mathrm{blocks}}} \intPart = \intPart \Pi_{\leq d, \I_S}.\label{eq:appendix_partition_isometry}
    \end{align}
    We let $w_I(\pi,\omega)^2$ be the probability of $(\pi,\omega)$ under some $I$ and $\tfrac{1}{s!}$ for $\pi$. The identity $\Id_{\neq j}$ refers to being the identity everywhere besides block $j$. Let us define block queries as follows,
    \begin{align}
        \Pi^{\blocks,0}_{x,j} &\coloneq \sum_{\substack{\pi,\omega\in\Sym(S) \\ \pi(x)\in S_j}} \ketbra{\pi,\omega},\\
        \Pi^{\blocks,1}_{x,j} &\coloneq \sum_{\substack{\pi,\omega\in\Sym(S) \\ \omega^{-1}(x)\in S_j}} \ketbra{\pi,\omega},\\
        \pC_{\blocks, x} &\coloneq \sum_{\pi,\omega} \ketbra{\pi,\omega} \otimes \sum_{j: \pi(x)\in S_j} (\pC_{S_j,\pi(x)} \otimes \Id_{\neq j}),\\
        \xorOP_{\blocks,x} &\coloneq \sum_{\pi,\omega} \ketbra{\pi,\omega} \otimes \sum_{j: \pi(x)\in S_j} (\xorOP_{S_j, \pi(x)}^{\omega| S_j} \otimes \Id_{\neq j}),\\
        \cOp{F}_{\blocks} \ket{\pi,\omega} \ket{J_1,\dots ,J_B} &\coloneq \ket{\omega^{-1},\pi^{-1}} \ket{J_1^{-1},\dots ,J_B^{-1}},\\
        \cpO_{\blocks,x}^b &\coloneq \begin{cases}
            \pC_{\blocks, x} \xorOP_{\blocks,x} \pC_{\blocks, x}^\dagger &\text{if $b=0$,}\\
            \cOp{F}_{\blocks} \pC_{\blocks, x} \xorOP_{\blocks,x} \pC_{\blocks, x}^\dagger \cOp{F}_{\blocks}^\dagger &\text{if $b=1$.}
        \end{cases}
    \end{align}
    We emphasize that the superscript $\omega|_{S_j}$ invokes the function-based version of $\xorOP$ from \cref{eq:appendix_relabeled_queries}. Notice that,
    \begin{align}
        \xorOP_{\blocks,x} \intPart &= \intPart \xorOP_{S,x},\\
        \cOp{F}_\blocks \intPart &= \intPart \cOp{F}_S\\
        \intPart^\dagger \Pi^{\blocks,b}_{x,j} \intPart &= \frac{s_j}{s}\Id.\label{eq:part_split}
    \end{align}
    Let us compare the compression primitives on $\calH_{S,x,I}$. Fix some $I\in \db_S$ and $x\notin \dom(I)$. For some $(\pi,\omega)\in \mathrm{Comp}(I)$, set $\dbJ = \dbJ(I,\pi,\omega)$, $j^*$ be the unique $j$ such that $\pi(x) \in S_j$ and $Y_{j} = (S\setminus \im(I)) \cap \omega(S_{j})$.
    Define,
    \begin{align}
        \ket{+_{\pi(x), \dbJ}^{\blocks}} &\coloneq \frac{1}{\sqrt{s_{j^*} - \abs{J_{j^*}}}} \sum_{y\in Y_{j^*}} \ket{J_1,\dots,J_{j^*}[\pi(x) \to \omega^{-1}(y)],\dots,J_B} \label{eq:appendix_block_uniform_extension}
    \end{align}
    Notice that for $y\in (S\setminus \im(I)) \cap \omega(S_{j^*})$, $\tfrac{w_{I[x\to y]}(\pi, \omega)}{\sqrt{s-\abs{I}}} = \tfrac{w_I(\pi,\omega)}{\sqrt{s_{j^*} - \abs{J_{j^*}}}}$. Therefore,
    \begin{align}
        \pC_{\blocks,x} \intPart \ket{I} &= \intPart \ket{+_{S,x,I}}, \label{eq:appendix_uniform_extension_forward}\\
        \pC_{\blocks,x} \intPart \ket{+_{S,x,I}} &= \intPart \ket{I}.\label{eq:appendix_uniform_extensions}
    \end{align}
    In the subspace defined above, the only remaining components are the states $\ket{\psi} $ such that,
    \begin{align*}
        &\ket{\psi} \coloneq \sum_{y\in S \setminus \im(I)} \alpha_y \ket{I[x\to y]}, &\sum_y \alpha_y = 0.
    \end{align*}
    By \cref{eq:pc_x_I}, $\pC_{S,x}\ket{\psi}=\ket{\psi}$. Expanding the block compression part, we have,
    \begin{align}
        &(\pC_{\blocks,x} \intPart - \intPart \pC_{S,x})\ket{\psi}\\
        &=\quad\sum_{(\pi,\omega)\in \mathrm{Comp}(I)} \frac{w_I(\pi,\omega) \sqrt{s - \abs{I}} }{s_{j^*} - \abs{J_{j^*}}} \left(\sum_{y\in Y_{j^*}} \alpha_y \right) \cdot \ket{\pi,\omega}\left(\ket{\dbJ} - \ket{+_{\pi(x), \dbJ}^{\blocks}}\right)\label{eq:appendix_compression_difference_psi}
    \end{align}
    Notice that,
    \begin{align}
        \EX_I\left[ \abs{\sum_{y\in Y_{j^*}} \alpha_y}^2 |\pi, \omega^{-1}|_{\im(I)} \right] = \frac{(s_{j^*} - \abs{J_{j^*}})(s - \abs{I} - (s_{j^*} - \abs{J_{j^*}}))}{(s-\abs{I})(s - \abs{I} - 1)} \norm{\ket{\psi}}^2. \label{eq:appendix_conditional_second_moment}
    \end{align}
    The database vectors in each summand in \cref{eq:appendix_compression_difference_psi} are orthogonal unit vectors. Therefore,
    \begin{align}
        \norm{(\pC_{\blocks,x} \intPart - \intPart \pC_{S,x})\ket{\psi}}^2 &= 2\EX_I \left[ \frac{s - \abs{I}}{(s_{j^*} - \abs{J_{j^*}})^2} \abs{\sum_{y\in Y_{j^*}} \alpha_y}^2 \right]\nonumber\\
        &= 2\EX_I \left[ \sum_{j: (s_j - \abs{J_j}) > 0} \frac{(s_j - \abs{J_j})}{s - \abs{I}} \frac{(s - \abs{I} - (s_{j} - \abs{J_{j}}))}{(s_{j} - \abs{J_{j}})(s - \abs{I} - 1)} \right] \norm{\psi}^2 \nonumber\\
        &=\frac{2\left(\EX_I[\abs{\{j: s_j - \abs{J_j} > 0\}}] - 1\right)}{s - \abs{I} - 1} \norm{\psi}^2.\label{eq:appendix_squared_compression_error}
    \end{align}
    where the second equality is due to \cref{eq:appendix_conditional_second_moment}. Under the assumption that $\ket{\psi}$ only has databases of size $d$ in its support, this implies that,
    \begin{align}
        &\norm{(\pC_{\blocks,x} \intPart - \intPart \pC_{S,x})\Pi_{\leq d, \I_S}} \leq \sqrt{\frac{2B}{s - d}}, &\mbox{(By \cref{eq:appendix_squared_compression_error,eq:appendix_uniform_extension_forward,eq:appendix_uniform_extensions})}
    \end{align}
    where we used the fact that error vectors for distinct $I$ are orthogonal due to \cref{eq:appendix_database_recovery}. Notice that
    \begin{align*}
         &\cpO_{\blocks,x}^0 \intPart - \intPart \cpO^0_{S,x} = \pC_{\blocks,x} \xorOP_{\blocks,x} (\pC_{\blocks,x} \intPart - \intPart \pC_{S,x}) \\
         &\quad + (\pC_{\blocks,x} \intPart - \intPart \pC_{S,x}) \xorOP_{S,x} \pC_{S,x}.
    \end{align*}
    Therefore,
    \begin{align}
        \norm{( \cpO_{\blocks,x}^b \intPart - \intPart \cpO^b_{S,x} ) \Pi_{\leq d, \I_S}} &\leq \sqrt{\frac{2B}{s-d}} + \sqrt{\frac{2B}{s-d-1}} \nonumber\\
        &= O \left( \sqrt{\frac{\log s}{s}} \right),\label{eq:appendix_partition_bound}
    \end{align}
    where we used the fact that the database can increase by at most one by \cref{eq:pc_x_I} and the bound on $B$. Note that the case when $b=1$ follows from the case with $b=0$ by conjugating the flips $\cOp{F}$.

    Next, we apply \cref{lem:power_four_useful_facts} with $D=S_j$ to obtain maps $\calV_{S_j}$ and oracle calls $\cmfO_{S_j}$. We define the space $\calH(\regP_S)$ and map $\calV_S$ as,
    \begin{align*}
        \calH(\regP_S) &\coloneq \calH(\Sym(S)^2) \otimes \bigotimes_{j\in [B]} \calH(\regP_{S_j})\\
        \intBlocks &\coloneq \Id_{\pi,\omega} \otimes \bigotimes_{j\in [B]} \calV_{S_j}\\
        \calV_S &\coloneq \intBlocks \intPart.
    \end{align*}
    Notice that $\calV_S^\dagger \calV_S = \Id$. The purified query operators are defined as follows,
    \begin{align*}
        \cmfO_{S,x}^0 &\coloneq \sum_{\pi,\omega}\ketbra{\pi,\omega} \otimes \sum_{j: \pi(x)\in S_j} \left(\cmfO_{S_j,\pi(x)}^{\omega|_{S_j}, 0} \otimes \Id_{\neq j}\right),\\
        \cmfO_{S,x}^1 &\coloneq \sum_{\pi,\omega}\ketbra{\pi,\omega} \otimes \sum_{j: \omega^{-1}(x)\in S_j} \left(\cmfO_{S_j,\omega^{-1}(x)}^{\pi^{-1}|_{S_j}, 1} \otimes \Id_{\neq j}\right),\\
        \cmfO_S &\coloneq \sum_{\substack{b\in \{0,1\}\\ x\in S}} \ketbra{b,x}_{\X} \otimes \cmfO_{S,x}^b.
    \end{align*}
    We initialize the register $\regP_S$ to $\ket{\bot}_{\regP_S}=\calV_S\ket{\bot}_{\I_S}$.
    Let $\ket{\xi}$ be such that $\Pi_{\leq d, \I_S} \ket{\xi} = \ket{\xi}$. For each $x,b$,
    \begin{align}
        &\norm{\left(\cmfO_{S,x}^b \intBlocks - \intBlocks \cpO_{\blocks,x}^b \right) \intPart \ket{\xi}}^2 \nonumber\\
        &\quad\leq K^2(d+3)^2 \sum_{j}\frac{\norm{\Pi^{\blocks,b}_{x,j}\intPart\ket{\xi}}^2}{\sqrt{s_j}} &\mbox{(By \cref{eq:isometry_between_purification_compression,eq:appendix_partition_isometry})}\nonumber\\
        &\quad= \frac{K^2(d+3)^2}{s} \sum_{j} \sqrt{s_j}\norm{\ket{\xi}}^2 &\mbox{(By \cref{eq:part_split})}\nonumber \\
        &\quad\leq \frac{6K^2 (d+3)^2}{\sqrt{s}}\norm{\ket{\xi}}^2 \label{eq:appendix_weighted_block_comparison}
    \end{align}
    Notice that,
    \begin{align*}
        &\cmfO_{S,x}^b \calV_S - \calV_S \cpO_{S,x}^b = \left(\cmfO_{S,x}^b \intBlocks - \intBlocks\cpO_{\blocks,x}^b \right) \intPart\\ 
        \quad&+ \intBlocks\left(\cpO_{\blocks,x}^b \intPart - \intPart \cpO^b_{S,x}\right).
    \end{align*}
    Therefore,
    \begin{align}
        &\norm{\left(\cmfO_{S,x}^b \calV_S - \calV_S \cpO_{S,x}^b\right) \Pi_{\leq d, \I_S}}\\
        &\quad\leq  \norm{\left(\cmfO_{S,x}^b \intBlocks - \intBlocks\cpO_{\blocks,x}^b \right) \intPart \Pi_{\leq d, \I_S}}\\
        &\quad + \norm{\left(\cpO_{\blocks,x}^b \intPart - \intPart \cpO^b_{S,x}\right) \Pi_{\leq d, \I_S}}\\
        &\quad \leq  \frac{\sqrt{6}K (d+3)}{s^{1/4}} + O \left( \sqrt{\frac{\log s}{s}} \right) &\mbox{(By \cref{eq:appendix_weighted_block_comparison,eq:appendix_partition_bound})}\label{eq:appendix_lemma_final_bound}
    \end{align}
    Let $K^\prime$ be some constant based on \cref{eq:appendix_lemma_final_bound}. The final result is due to \cref{lem:operator_norm_max}.
\end{proof}

Finally, we may prove the main result of the section.

\begin{proof}[Proof of \cref{lem:general_intertwiner_lemma}]
    We will use the variables from \cref{lem:appendix_local_permutation_comparison} and let $K$ be the constant in the bound. Define $\inter: \calH(\regP_S)\to \calH(\regP_S) \otimes \calH(\I_S)$ to be the following,
    \begin{align}
        \inter \ket{\xi} &\coloneq \ket{\bot}_{\regP_S} \otimes (\calV_S^\dagger\ket{\xi} )_{\I_S} + \left( (\Id - \calV_S \calV_S^\dagger)\ket{\xi} \right)_{\regP_S} \otimes \ket{\bot}_{\I_S}.\label{eq:appendix_intertwiner_definition}
    \end{align}
    Notice that as $(\Id - \calV_S \calV_S^\dagger)\ket{\bot}_{\regP_S} = 0$,
    \begin{align}
        \norm{\inter \ket{\xi}}^2 = \norm{\calV_S^\dagger \ket{\xi}}^2 + \norm{(\Id - \calV_S \calV_S^\dagger)\ket{\xi}}^2 = \norm{\ket{\xi}}^2. \label{eq:appendix_intertwiner_norm}
    \end{align}
    Furthermore,
    \begin{align}
        \inter \calV_S \ket{I} = \ket{\bot}_{\regP_S}\ket{I}_{\I_S} \label{eq:appendix_intertwiner_on_image}\\
        \inter \ket{\bot}_{\regP_S} = \ket{\bot}_{\regP_S} \ket{\bot}_{\I_S}.\nonumber
    \end{align}
    Assume that $q\leq s/2$ as otherwise the bound follows trivially. Using the algorithm $\alg$ up to $t\in [0,q]$ queries, define,
    \begin{align*}
        \ket{\psi_t^{(\cmfO)}}_{\A\regP_S} &\coloneq A_t \cmfO_S \dots A_1 \cmfO_S A_0 \ket{0}_\A\ket{\bot}_{\regP_S},\\
        \ket{\psi_t^{(\cpO)}}_{\A\I_S} &\coloneq A_t \cpO_S \dots A_1 \cpO_S A_0 \ket{0}_\A \ket{\bot}_{\I_S}.
    \end{align*}
    For each step $t$,
    \begin{align*}
        &\ket{\psi_{t+1}^{(\cmfO)}}_{\A\regP_S} - \calV_S \ket{\psi_{t+1}^{(\cpO)}}_{\A\I_S}\\
        &\quad= A_{t+1} \cmfO_S \left(\ket{\psi_t^{(\cmfO)}}_{\A\regP_S} - \calV_S \ket{\psi_t^{(\cpO)}}_{\A\I_S} \right) + A_{t+1} (\cmfO_S \calV_S - \calV_S \cpO_S)\ket{\psi_t^{(\cpO)}}_{\A\I_S}
    \end{align*}
    Taking the norm and considering the final states after $q\geq1$ queries,
    \begin{align}
        &\norm{\ket{\psi_q^{(\cmfO)}}_{\A\regP_S} - \calV_S \ket{\psi_q^{(\cpO)}}_{\A\I_S}}\nonumber\\
        &\quad\leq \norm{\ket{\psi_{q-1}^{(\cmfO)}}_{\A\regP_S} - \calV_S \ket{\psi_{q-1}^{(\cpO)}}_{\A\I_S}} + K\frac{q+3}{s^{1/4}} &\mbox{(By \cref{lem:bounded,lem:appendix_local_permutation_comparison})}\nonumber\\
        &\leq \frac{K q(q+3)}{s^{1/4}},\label{eq:appendix_reached_state_difference}
    \end{align}
    where we used the fact that $\ket*{\psi_{0}^{(\cmfO)}}_{\A\regP_S} - \calV_S \ket*{\psi_{0}^{(\cpO)}}_{\A\I_S}=0$. Notice that by \cref{eq:appendix_intertwiner_on_image},
    \begin{align}
        &(\inter \cmfO_S - \cpO_S \inter)\ket{\psi_q^{(\cmfO)}}\\
        &\quad= (\inter \cmfO_S - \cpO_S \inter)(\ket{\psi_q^{(\cmfO)}} - \calV_S \ket{\psi_q^{(\cpO)}}) + \inter(\cmfO_S \calV_S - \calV_S \cpO_S) \ket{\psi_q^{(\cpO)}}\label{eq:decomp_using_inter}
    \end{align}
    Therefore,
    \begin{align*}
        &\norm{(\inter \cmfO_S - \cpO_S \inter)\ket{\psi_q^{(\cmfO)}}}\\
        &\quad\leq 2 \norm{\ket{\psi_q^{(\cmfO)}} - \calV_S \ket{\psi_q^{(\cpO)}}} + \norm{(\cmfO_S \calV_S - \calV_S \cpO_S) \ket{\psi_q^{(\cpO)}}} &\mbox{(By \cref{eq:appendix_intertwiner_norm,eq:decomp_using_inter})}\\
        &\quad\leq 2 \norm{\ket{\psi_q^{(\cmfO)}} - \calV_S \ket{\psi_q^{(\cpO)}}} + K\frac{q+3}{s^{1/4}} &\mbox{(By \cref{lem:bounded,lem:appendix_local_permutation_comparison})}\\
        &\quad \leq 3K \frac{(q+1)^2}{s^{1/4}}. &\mbox{(By \cref{eq:appendix_reached_state_difference})}
    \end{align*}
\end{proof}

\end{document}